\documentclass[11pt,reqno]{amsart}
\usepackage{amssymb}

\usepackage{enumitem}
\usepackage{amsthm,amsfonts,amssymb,euscript,color,bbm}
\usepackage{comment}
\usepackage{mathtools}
\usepackage{stmaryrd}
\usepackage{mathrsfs}
\usepackage{pifont}
\usepackage{amsfonts,amsmath,amssymb,amsthm,amscd,cancel}
\usepackage{graphicx}
\usepackage[]{geometry}
\usepackage{verbatim}
\usepackage{mathrsfs}
\usepackage{fancyhdr}
\usepackage{footnpag,footmisc}
\usepackage[normalem]{ulem}
\usepackage{float}

\usepackage{cite}
\usepackage{hyperref}
\usepackage{engord}
\usepackage[displaymath]{lineno}
\usepackage[utf8]{inputenc}

\theoremstyle{definition}
\newtheorem{theorem}{Theorem}[section]
\newtheorem{lemma}{Lemma}[section]
\newtheorem{proposition}{Proposition}[section]

\newtheorem{remark}{Remark}[section]

\allowdisplaybreaks[4]

\DeclareMathAlphabet{\mathsfsl}{OT1}{cmss}{m}{sl}
\numberwithin{equation}{section}

\newcommand{\D}{\mathrm{d}}

\newcommand{\tr}{\mathrm{tr}}
\newcommand{\cir}[1]{\overset{\circ}{#1}}

\def\alphab{\underline{\alpha}}
\def\betab{\underline{\beta}}
\def\chib{\underline{\chi}}
\def\chibh{\widehat{\underline{\chi}}}
\def\chih{\widehat{\chi}}
\def\etab{\underline{\eta}}

\def\Lb{\underline{L}}
\def\mub{\underline{\mu}}

\def\tr{\mathrm{tr}}
\def\omegab{\underline{\omega}}

\def\sigmac{\check{\sigma}}
\def\tensor{\widehat{\otimes}}

\def\ub{\underline{u}}
\def\Cb{\underline{C}}

\def\Lh{\widehat{L}}
\def\Lbh{\widehat{\underline{L}}}

\def\Xib{\underline{\Xi}}
\def\Ib{\underline{I}}
\def\Lambdab{\underline{\Lambda}}
\def\Nb{\underline{N}}
\def\Thetab{\underline{\Theta}}
\def\Kb{\underline{K}}

\newcommand{\Db}{\underline{D}}
\newcommand{\Dh}{\widehat{D}}
\newcommand{\Dbh}{\widehat{\underline{D}}}

\def\nablas{\mbox{$\nabla \mkern -13mu /$ }}
\def\Deltas{\mbox{$\Delta \mkern -13mu /$ }}

\def\divs{\mbox{$\mathrm{div} \mkern -13mu /$ }}
\def\curls{\mbox{$\mathrm{curl} \mkern -13mu /$ }}
\def\omegas{\mbox{$\omega \mkern -13mu /$ }}
\def\omegabs{\mbox{$\omegab \mkern -13mu /$ }}
\def\ds{\mbox{$\nabla \mkern -13mu /$ }}
\def\gs{\mbox{$g \mkern -9mu /$}}
\def\epsilons{\mbox{$\epsilon \mkern -9mu /$}}

\def\Os{\mbox{$\mathcal{O} \mkern -11mu /$}}
\def\Es{\mbox{$\mathcal{E} \mkern -11mu /$}}

\def\Ks{\mbox{$K \mkern -13mu / $}}

\begin{document}
	
	\title[Interior gravitational perturbations to naked singularities]{Interior gravitational perturbations to naked singularities of a scalar field}

	\author[Junbin Li]{Junbin Li}
	\address{Department of Mathematics, Sun Yat-sen University, Guangzhou, China}
	\email{lijunbin@mail.sysu.edu.cn}
	\author[Tingting Li]{Tingting Li}
	\address{Department of Mathematics, Sun Yat-sen University, Guangzhou, China}
	\email{litt76@mail2.sysu.edu.cn}
	
	\begin{abstract}
		
	For the $k$-self-similar naked singularity solutions of Einstein--scalar field equations constructed by Christodoulou, we construct a family of interior gravitational perturbations leading to trapped surface formation that remain large but finite at the threshold H\"older regularity, while converging to zero in all regularities below the threshold. This extends the interior spherically symmetric result below the threshold in \cite{Li25} to non-spherically symmetric setting.  We further construct a genuinely localized family of perturbations whose angular support shrinks to a single point, with smooth convergence to zero away from that point. This exploits the additional angular freedom available beyond spherical symmetry.
Moreover,  when measured instead in Sobolev regularity,  such genuinely localized family of perturbations still converges to zero in all regularities below the same threshold.

	\end{abstract}
	
	\maketitle
	

	\setcounter{tocdepth}{1}


	\section{Introduction}

\subsection{Previous Results}
   In general relativity, the weak cosmic censorship conjecture is one of the most important problem. It says that generically naked singularities will not arise in gravitational collapses, i.e., solutions of the Einstein equations (coupled with reasonable matter field) of asymptotically flat initial data. In a series of papers \cite{Chr87, Chr91, Chr93, Chr99}, Christodoulou studied spherically symmetric solutions of the Einstein--scalar field equations
   \begin{equation}\label{ES}\mathbf{Ric}_{\alpha\beta}-\frac{1}{2}\mathbf{R}g_{\alpha\beta}=2\partial_\alpha\phi\partial_\beta\phi-g_{\alpha\beta}\partial_\mu\phi\partial^\mu\phi\end{equation}
    and verified the weak (and strong) cosmic censorship conjecture for this model in the class of BV solutions. In the paper \cite{Chr94}, Christodoulou was able to construct examples of so-called $k$-self-similar naked singularities of this model, which implies that the term ``generic'' is necessary. The main strategy  of the proof of the weak cosmic censorship is then to identify the naked singularities and study their instability. In the paper \cite{Chr99}, Christodoulou showed that for a solution with a first singularity at the center, if the past null cone $\mathcal{N}$ of the singularity has the property that the following blue-shift 
    $$\int_{\mathcal{N}}\frac{1}{r}\frac{\mu}{1-\mu}\D r$$
    is infinite, then the solution contains a trapped region terminating at a spacelike singularity, or satisfies this property after a small BV perturbation to the future of $\mathcal{N}$, which we call exterior perturbations. To gain insights to the non-spherically symmetric problem, Li--Liu \cite{Liu-Li, Li-Liu1} studied exterior gravitational perturbations (which must be non-spherically symmetric due to Birkhoff theorem) to spherically symmetric singular solution of the Einstein-scalar field equations, with the infinite blue-shift property. We showed that small $C^1$ gravitational perturbation can result in trapped surface formation, raising the codimension of the exceptional set, which consists of perturbations still resulting in naked singularities. In general, the perturbations produced in this way must be in a regularity strictly lower than $C^2$, and on the other hand it can be made $C^\infty$ away from the intersection of the initial cone and $\mathcal{N}$. The finite regularity only localizes around $\mathcal{N}$ and in the discussions below, all finite regularities are understood in this localized sense\footnote{A recent work \cite{Ci26} considered the weak cosmic censorship in $2+1$ circularly symmetric Einstein-scalar field system with negative cosmological constant. Because of the so-called mass gap phenomena in this model, the blue-shift in singular solutions tends to infinity in an extremely fast rate, so that one can find perturbations resulting in trappdd surface formation in $C^k$ for any integer $k$.}.
    
    Although the instability of general naked singularities was verified, it is also important to study specific examples,  which will advance our understanding of the weak cosmic censorship conjecture. The examples constructed in \cite{Chr94}, the $k$-self-similar ones, have been extensively studied recently. An \cite{A} constructed small $C^1$ anisotropic exterior perturbations (which means that the shear tensor is supported only in a small region of the spherical sections) to the $k$-self-similar naked singularities so that an apparent horizon (in addition to trapped surfaces) emerging from the singularity can be identified in the maximal development. 
    
    On the other hand, the stable construction of the exterior of the vacuum naked singularities \cite{R-Sh} suggested that exterior region is stable under exterior perturbations with sufficiently high regularity. To be more precise, let us recall that the $k$-self-similar naked singularity solution of \eqref{ES} suffers a finite regularity across the past null cone $\mathcal{N}$ of the singularity. Let $C_o$ be an outgoing null cone from the center in the solution, parametrized by the area radius $r$. Then the function $\partial_r(r\phi)\big|_{C_o}$, which is the initial data, belongs to $C^{\frac{k^2}{1-k^2}}$ (and $C^\infty$ alway from $S_0=\mathcal{N}\cap C_o$). In fact, the function $\partial_r(r\phi)\big|_{C_o}$, after subtracted its value $\frac{1}{k}$ on $S_0$, looks like a power function 
    $$a|r-r_0|^{\frac{k^2}{1-k^2}}+\text{more regular part}$$
     around $S_0$ with possibly different coefficients on both sides of $S_0$,  where $r_0=r\big|_{S_0}$. The (localized) regularity $C^{\frac{k^2}{1-k^2}}$ is called the threshold regularity. The analysis in \cite{JS1, JS2} suggested that, at least restricted within spherical symmetry,  small exterior perturbations above and equal to the threshold regularity lead to stability.  Examples of naked singularity exterior for \eqref{ES} outside spherical symmetry are constructed very recently in \cite{An2} and their instability to anisotropic apparent horizon formation was shown in \cite{An3}. For exterior perturbations below threshold, one expected instability. It is not hard to see from the arguments of previous works that trapped surface will form for H\"older perturbations with sufficiently low exponent,  that is, the perturbed data $\left(\partial_r(r\phi)\right)_t$ looks like
     $$\frac{1}{k}+t|r-r_0|^\alpha,$$
     for $\alpha>0$ sufficiently small and $t\ne0$. It is further confirmed in \cite{Li25} that, in spherical symmetry, one can find small perturbations resulting in trapped surface formation in $C^\alpha$ for any $\alpha<\frac{k^2}{1-k^2}$. 
     \begin{remark}
     The exterior perturbations with the form of power function produce not only a single trapped surface but a sequence of trapped surfaces with arbitrarily small areas. A consequence is that the apparent horizon emanates from the singularity and even a small piece of the Cauchy horizon from the singularity does not exist, justifying also the strong cosmic censorship. But if we only concern about the weak cosmic censorship, a single trapped surface is enough (see \cite{D}) and the differences of the perturbed and original data can be made at least $C^1$, while the convergence is still in $C^\alpha$. 
     \end{remark}

\begin{figure}
\includegraphics[width=1.5 in]{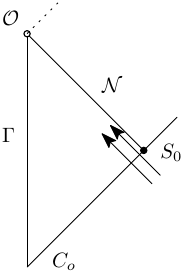}
\caption{Interior perturbations}
\label{fig:interior}
\end{figure}
    
     Although the exterior stability and instability were studied extensively, there are almost no results concerning interior perturbations before \cite{JS2}, in which solutions of linear wave equation in $k$-self-similar background (for $k^2\ll1$) was studied. Here interior perturbations refer to perturbations not supported to the future of $\mathcal{N}$ (see Figure \ref{fig:interior}). One obvious difference is that interior perturbations will change the causal past of the singularity and the end state of the perturbations is not a priori known. Interior perturbations may result in stability, instability to trapped surface (black hole) formation, or instability to dispersion, the latter of which cannot caused by exterior perturbations. Within spherical symmetry, in \cite{Li25} we found small interior perturbations resulting in trapped surface and black hole formation still in a localized $C^{\alpha}$ topology for any $\alpha<\frac{k^2}{1-k^2}$.  On the stability side, the works \cite{SZ26, Zheng}, completing the program initiating in \cite{JS2}, proved stability under small (interior) perturbations above the threshold (for $k^2\ll1$, the smallness is measured using a weighed norm). The stability result in \cite{Zheng} implies that perturbations above the threshold cannot result in dispersion, and the instability result in \cite{Li25} narrows the search of dispersive perturbations in the original BV topology. All these stability and instability results justify that the regularity $C^{\frac{k^2}{1-k^2}}$ is a genuine threshold regularity.

\subsection{The main results}
 
 In this paper, we generalize the instability result for interior perturbations below the threshold to non-spherically symmetric perturbations.  We demonstrate that the $k$-self-similar naked singularity solutions are also unstable under (interior) gravitational perturbations below the threshold in a certain sense. In fact, the family of perturbations we construct has uniform bound at the threshold.  Our theorems provide complementary information toward understanding how naked singularity solutions are embedded in the phase space of initial data. A concrete question in this direction is  whether there exist dispersive gravitational perturbations of the $k$-self-similar naked singularity solutions and our results rule out certain candidates. 
 
The main results in \cite{Li25} and this paper are set up as follows.  Let $\mathcal{M}$ be the $k$-self-similar naked singularity solution of spherically symmetric Einstein--scalar field system and $C_o$ be some outgoing null cone parametrized by $\ub$, which is one of the optical function in the double null coordinates. Let $\mathcal{N}$ be the past null cone of the singularity $\mathcal{O}$ and denote $S_0=\mathcal{N}\cap C_o$. Outside spherical symmetry, initial data imposed on $C_o$ consists of the scalar function $\partial_r(r\phi)$ and the shear $\chih$ of $C_o$. Let us recall that the following theorem established in \cite{Li25} within spherical symmetry. 
 \begin{theorem}[\cite{Li25}]\label{thm:Li25}
 For any $\alpha\in[0,\frac{k^2}{1-k^2})$, there exists a $t$-family of spherically symmetric perturbations of $\left(\partial_r(r\phi)\right)_t$ on $C_o$, supported before $S_0$, such that $\left(\partial_r(r\phi)\right)_t$ converges as $t\to0$ to the data of $\mathcal{M}$ in $C^\infty$ away from $S_0$ and in $C^{\alpha}$ on $C_o$. The maximal development $\mathcal{M}_t$ of perturbed data contains a closed trapped surface for $t\ne0$.
 \end{theorem}

The theorem remains true if we replace the modifier ``supported before $S_0$'' by ``supported after $S_0$'', which is the exterior instability arbitrarily close to the threshold from below, also established in \cite{Li25}.

The first main result of this paper is to generalize the above result to non-spherically symmetric case, allowing gravitational perturbations.  
 \begin{theorem}\label{main1}
There exists a $\delta$-family of perturbations of the shear $\chih_\delta$ on $C_o$, supported before $S_0$, such that the $C^{\frac{k^2}{1-k^2}}_{\ub}$ norm of $\chih_\delta$, with suitable angular control,  is uniformly bounded, and $\chih_\delta$ converges as $\delta\to0$ to zero in $C^\infty$ away from $S_0$ and in $C^{\alpha}_{\ub} C^\infty_\vartheta$ on $C_o$ for any $\alpha<\frac{k^2}{1-k^2}$. The maximal development $\mathcal{M}_\delta$ of perturbed data contains a closed trapped surface for $\delta\ne0$. 
 \end{theorem}
 
 
\begin{remark}
We refrain from including the statement of exterior perturbations: the theorem remains true if we replace the modifier ``supported before $S_0$'' by ``supported after $S_0$''. Although this is a new result, it follows directly from the proof and is therefore not stated separately.
\end{remark}
\begin{remark}
We note that Theorem \ref{main1} gives uniform boundedness of the perturbations at the threshold regularity. Although this was not stated in Theorem \ref{thm:Li25}, the same endpoint property can be obtained there by a minor modification of the argument  as in this paper.
\end{remark}
\begin{remark}
The argument becomes somewhat delicate once the angular derivatives of $L\phi$ are included, which is already present in our previous work \cite{Li-Liu1}, especially when approaching the threshold regularity in this paper.  The above statement is proven under the assumption that some top angular derivative of $\partial_r(r\phi)$ vanishes in a sufficiently fast way as approaching $S_0$. For convenience, we may assume that the scalar field is constant along each spherical section on $C_o$, or even leave it unperturbed.  If we choose to perturb the scalar field in a similar way of the shear simultaneously, then it may not hold that the family we constructed is uniformly bounded at the threshold, but it still tends to zero in all regularities below the threshold.
\end{remark}

The second main result further generalizes the interior instability to an anisotropic setting, allowing the family of perturbations to converge smoothly away of a \emph{point} of $S_0$, rather than the whole spherical section $S_0$ as before. This flexibility fully exploits the genuinely non-spherically symmetric nature. A natural question would then be whether the threshold can still be approached arbitrarily closely from below in such a genuinely localized setting. In this setting,  the perturbations are constructed so that their supports shrink to a point, while the amplitude of the perturbations must be increased to compensate for the shrinking angular support.  If we measure the size of these perturbations in  a  localized H\"older norm (in $\ub$), the H\"older exponent must be bounded away from the threshold $\frac{k^2}{1-k^2}$. A natural way to quantify the size of the perturbations, which reveals that our result can approach the threshold arbitrarily closely from below, is through the scaling of the the total incoming energy:
$$\int_{C_o}|\chih_\delta|^2.$$
The same scaling can equivalently be expressed in terms of the three-dimensional $H^s$ spaces on $C_o$. Note that the power function $\ub^\alpha$ lies in $H^{\frac{1}{2}+\alpha'}$ for $\alpha'<\alpha$, so the threshold regularity corresponds to $H^{\frac{1}{2}+\frac{k^2}{1-k^2}}$. The initial data $\partial_r(r\phi)$ of the $k$-self-similar naked singularity solution lies in $H_{\mathrm{loc}}^{\frac{1}{2}+\alpha'}$ for any $\alpha'<\frac{k^2}{1-k^2}$. In terms of the $H^s$ spaces, the second main result is the following.
 \begin{theorem}\label{main2}
 Consider the same setting as described before Theorem \ref{main1}. Then there exists a $\delta$-family of perturbations of the shear $\chih_\delta$ on $C_o$, supported before $S_0$, such that  $\chih_\delta$ converges as $\delta\to0$ to zero in $C^\infty$ away from a point on $S_0$ and in $H^{\frac{1}{2}+\alpha}(C_o)$ on $C_o$ for any $\alpha<\frac{k^2}{1-k^2}$. The maximal development $\mathcal{M}_\delta$ of perturbed data contains a closed trapped surface for $\delta\ne0$. 
  \end{theorem} 
 
\begin{remark} The family of perturbations constructed in Theorem \ref{main1} can also be measured in $H^s$ spaces and it is uniformly bounded at the threshold. In contrast, the family constructed in Theorem \ref{main2}, with shrinking angular support, is not uniformly bounded at the threshold. Similar to  Theorem \ref{main1}, we still have an anisotropic  exterior version of Theorem \ref{main2} but we will not state it separately. Note that this type of anisotropic result is new even when we consider trapped surface formation from Minkowski. See the discussion below. 
 \end{remark}

 At last, we remark that, although the threshold space $H^{\frac{1}{2}+\frac{k^2}{1-k^2}}$ arises naturally from the regularity of the $k$-self-similar naked singularity solutions, it is not clear whether it is a genuine threshold (outside spherical symmetry). It requests further studies on whether the $k$-self-similar naked singularity solutions are stable under small (anisotropic) perturbations at and above $H^{\frac{1}{2}+\frac{k^2}{1-k^2}}$.

	\subsection{Main Ideas of the Proof}
	
	The proof of Theorem \ref{main1} relies on two main ingredients. The first addresses the difficulties caused by interior perturbations, while the second allows us to obtain perturbations with regularity arbitrarily close to the threshold. Both ingredients were introduced in our previous work \cite{Li25} in spherical symmetry and are adapted here to the non-spherically symmetric setting. 
	
	\subsubsection{Interior perturbations} As mentioned before, the difficulty in studying instability under interior perturbations is that it is not a priori known the end state of the perturbed solution. When studying exterior perturbations, we started from the past null cone $\mathcal{N}$ of the singularity, and reduced it to the trapped surface formation problem exterior to $\mathcal{N}$.  For interior perturbations, we have no such natural starting point.
	
	The solution turns out to be surprisingly straightforward. We do not start from $\mathcal{N}$ but choose to start from an incoming null cone $\mathcal{N}_1$ before $\mathcal{N}$, whose vertex is regular, unlike that of $\mathcal{N}$ (see Figure \ref{fig:nakedinterior}). The perturbations we put in are ``exterior to''  $\mathcal{N}_1$ and hence the past of $\mathcal{N}_1$ will not be influenced by the perturbations. Since $\mathcal{N}_1$ has a regular vertex, it is impossible to find arbitrarily small perturbations exterior to $\mathcal{N}_1$ that lead to trapped surface formation within the future domain of dependence of $\mathcal{N}_1$. Nevertheless, we can quantify the the energy required for trapped surface formation, which will be reduced by the blue shift of $\mathcal{N}_1$, even when the blue shift of $\mathcal{N}_1$ is finite.  The key to resolving this problem is the following. We can let $\mathcal{N}_1$ approach  $\mathcal{N}$ and the blue shift of $\mathcal{N}_1$ tends to infinity. Based on the quantitative analysis of the trapped surface formation, we find that the energy required for trapped surface formation tends to zero, which is exactly the instability. The above idea works for more general singular solutions besides the $k$-self-similar naked singularities. However, establishing interior instability with arbitrarily regularity below the threshold requires a more refined quantitative analysis, which is closely related to the second ingredient discussed below.

	Obviously, the ideas above can be applied to construct non-spherically symmetric gravitational perturbations, which will not influence the past of  $\mathcal{N}_1$. The required quantitative analysis can also be established because $\mathcal{N}_1$ is still spherically symmetric after perturbations.

 \begin{figure}
\includegraphics[width=1.5 in]{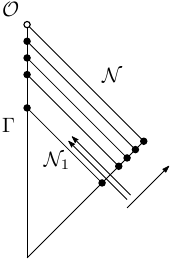}
\caption{}
\label{fig:nakedinterior}
\end{figure}

	\subsubsection{Regularity issues in trapped surface formation} In view of the above discussions,  we are going to focus on the trapped surface formation problem exterior to $\mathcal{N}_1$ (see Figure \ref{fig:ingredient2}). Letting $\mathcal N_1$ approach $\mathcal N$ amounts to varying the data on $\mathcal N_1$. Therefore, the problem is reduced to studying trapped surface formation  with varying data on $\mathcal{N}_1$. As $\mathcal{N}_1$ approaches $\mathcal{N}$, the data on $\mathcal{N}_1$ converge to the singular data on $\mathcal{N}$ with convergence being uniform away from the vertex. So, understanding the exterior instability (to $\mathcal{N}$) remains the key step in understanding interior instability.

 \begin{figure}
\includegraphics[width=1.5 in]{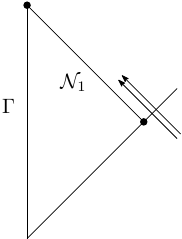}
\caption{}
\label{fig:ingredient2}
\end{figure}
	
	The exterior instability under gravitational perturbations was studied in \cite{Li-Liu1} and \cite{A} based on spherically symmetric background, and in \cite{An3} with a genuine non-spherically symmetric background.  However,  the instability in arbitrary regularity  below the threshold was not achieved.  On a technical level, they essentially fall within the framework of trapped surface formation pioneered by Christodoulou in \cite{Chr08}. He considered the trapped surface formation problem in vacuum, with initial data imposed on two intersecting null cones. Initial data on the incoming one is Minkowskian. Suppose that in a future neighborhood of the incoming null cone,  we have a double null coordinate system  $(\ub, u, \vartheta)$, where  $\ub$ and $u$ are two optical functions with values $\ub=0$ on the incoming null cone, and $u=u_0$ on the outgoing null cone respectively, and both increase towards the future. Christodoulou showed that if the initial shear $\chih$ on the outgoing cone obeys the following schematic upper bounds for $0\le\ub\le\delta$,
$$|\chih|\lesssim\delta^{-\frac{1}{2}}, |\partial_{\ub}\chih|\lesssim\delta^{-\frac{3}{2}},$$
then if $\delta>0$ is sufficiently small, the regular solution of the vacuum Einstein equations exists up to $|u|\approx c$, where the upper bounds still hold true. The choice of $\delta$ depends on $c$. If moreover a lower bound on the energy holds
\begin{equation}\label{Chr08lower}\int_{0}^\delta|\chih|^2(\ub,\vartheta)\D\ub\ge c, \forall\vartheta\in S^2,\end{equation}
then a closed trapped surface will form around $\ub=\delta, |u|\approx c$, within the existence region.  Heuristically, the lower bound on the energy corresponds to the pointwise bound 
$$|\chih|\ge c^{\frac{1}{2}}\delta^{-\frac{1}{2}},$$
whose power in $\delta$ is the same to the upper bound.  Here $c$ is an arbitrary positive constant, this result can be viewed as an $L^2$-instability result, that is, gravitational radiation with arbitrarily small $L^2$ norm leads to trapped surface formation. 

In a later work, An-Luk \cite{An-Luk} assumes the upper bound
$$|\chih|\lesssim a^{\frac{1}{2}},$$
for $0\le\ub\le\delta$, then the solution of vacuum Einstein equations remains regular up to $|u|=\delta a^{\frac{1}{2}}b$, where $b$ is a sufficiently large constant and $a\ge b$. Assuming in addition
$$|\chih|\ge 2\sqrt{ba^{\frac{1}{2}}},$$
then  a closed trapped surface forms at $\ub=\delta, |u|=\delta a^{\frac{1}{2}}b$. Here  $a$ is allowed to depend on $\delta$, then if we set $a^{\frac{1}{2}}=\delta^{-\frac{1}{2}}$, then the upper bound is the same to Christodoulou's. The existence part of the above result says that the solution remains regular up to $|u|=\delta^{\frac{1}{2}}b$, much deeper than Christodoulou's. This allows us to relax the lower bound of Christodoulou to
$$|\chih|\gtrsim 2\sqrt{b}\delta^{-\frac{1}{4}},$$
which is equivalent to that $\chih$ is large in $H^{\frac{1}{4}}$ at the level of scaling. Furthermore, An-Luk's result allows to produce closed trapped surface by putting in even less energy, by setting for example $a\approx b^2$. Now the lower bound becomes
$$|\chih|\ge 2 b,$$
equivalent to that $\chih$ is large in $H^{\frac{1}{2}}$, the scale-critical norm of the vacuum Einstein equations, so we need to solve the equation up to $|u|\approx \delta b^2$. This can be accomplished because the upper bound becomes
$$|\chih|\lesssim b,$$
in which case the solution does indeed remain regular up to $|u|\approx \delta b^2$.

The main idea of An-Luk's work is that if we want to produce closed trapped surface with less energy, then the closed trapped surface is expected to form in a deeper region, so we need to solve the equations in a deeper region, based on a smaller upper bound. The upper and lower bounds should be adjusted simultaneously.

The same idea works for the case when the background is a naked singularity solution. The Minkowski space is indeed a $0$-self-similar solution of the Einstein(-scalar field) equations. In the language of this paper, the above trapped surface formation results correspond to the case when $\mathcal{N}$ is a flat null cone. The singular behavior of $\mathcal{N}$ allows us to produce a closed trapped surface with even smaller incoming energy. To obtain a sharp trapped surface formation result, we need to establish a sharp existence result.

Suppose now $\mathcal{N}$ is the past null cone of the singularity in the $k$-self-similar naked singuarity solution. If we impose on $\mathcal{N}$ the gauge condition $u=-r$, then the null lapse $\Omega=\sqrt{-\frac{1}{2}g(\partial_{\ub},\partial_u)}$ restricted on $\mathcal{N}$ is simply $\Omega^2(0,u)=|u|^{k^2}$ (up to a constant multiple). As a comparison, $\Omega^2(0,u)\equiv1$ for when $\mathcal{N}$ is flat. Suppose that we have all estimates we want in the existence region.  The Raychaudhuri equation along outgoing direction reads 
$$\Omega^{-2}\partial_{\ub}(\Omega^{-1}\tr\chi)\le -\frac{1}{2}(\Omega^{-1}\tr\chi)^2-|\Omega^{-1}\chih|^2.$$
Integrating it along $\ub\in[0,\delta]$ along $C_u$ yields
$$\Omega^{-1}|u|\tr\chi\big|_{\ub=\delta}-\Omega^{-1}|u|\tr\chi\big|_{\ub=0}\lesssim -\delta |u|^{-1}\Omega^{-2}||u|\Omega\chih|^2.$$
On the other hand, the equation of $\partial_u(\Omega\chih)$ shows that
$$|u|\Omega\chih(\ub,u)\approx |u_0|\Omega\chih(\ub,u_0)=: f(\ub), 0\le\ub\le\delta,$$
we will have
$$\Omega^{-1}|u|\tr\chi\lesssim 1-\delta |u|^{-1}\Omega^{-2}|f(\ub)|^2, 0\le\ub\le\delta.$$
In order that a closed trapped surface forms along $C_u$, the initial data $\chih$ should verifies
$$|f(\ub)|^2\gtrsim \Omega^2(0,u)\delta^{-1}|u|=\delta^{-1}|u|^{1+k^2}, 0\le\ub\le\delta.$$
Suppose that the energy we put in are exactly at the threshold, that is, $|f(\ub)|\ge c\delta^{\frac{k^2}{1-k^2}}$, where $c$ is a large constant (small perturbations at the threshold may lead to stability, whose spherically symmetric version was established in \cite{Zheng}), then the closed trapped surface can only form not earlier than $c\delta^{\frac{2k^2}{1-k^2}}\approx \delta^{-1}|u|^{1+k^2}$, that is, 
\begin{equation}\label{expectedlocation}|u|\approx c^{\frac{1}{1+k^2}}\delta^{\frac{1}{1-k^2}}.\end{equation}
Note that this is exactly the intersection of $\ub=\delta$ and the future boundary of the region I of the $k$-self-similar naked singularity solutions  in the problem of naked singularity constructions (see \cite{R-Sh, JS1, An2}).   In fact, when the $k$-self-similar naked singularity solution is written in the double null coordinate system described as above, then the dimensionless variables are functions of $\frac{\ub}{|u|^{1-k^2}}$ alone, which is exactly the definition of $k$-self-similarity. On the other hand, in order to solve the solution up to $|u|\approx c^{\frac{1}{1+k^2}}\delta^{\frac{1}{1-k^2}}$, we need to propagate the self-similar bound $|\Omega^{-1}\chih|\lesssim C |u|^{-1}$, or
\begin{equation}\label{expectedupper}|u||\Omega\chih|\lesssim C|u|^{k^2},\end{equation}
which is consistence with the initial upper bound  (plugging in \eqref{expectedlocation})
\begin{equation}\label{expectedupperinitial}|u_0||\Omega\chih(\ub,u_0)|\lesssim c^{\frac{k^2}{1+k^2}}C\delta^{\frac{k^2}{1-k^2}},  \end{equation}
The above heuristic will be made rigorous in Theorem \ref{main4} in Section \ref{sec:trappedsurface},  a trapped surface formation theorem with data that are large at the threshold regularity, which is a nonzero-$k$ analogue of the result of An--Luk.  This result also gives exterior instability below the threshold since initial data with the form
$$\chih(\ub, u_0,\vartheta)=\delta^{\frac{k^2}{1-k^2}}\chih_0(\frac{\ub}{\delta},\vartheta),$$
which verifies \eqref{expectedupperinitial}, is already small in $C^\alpha_{\ub}$ for any $\alpha<\frac{k^2}{1-k^2}$ when $\delta$ is small enough. Previous works on exterior instability (like \cite{ Li-Liu1, A ,L-Z3,An3}) only propagate the estimates up to $|u|\approx \delta^{1-c(k^2)}$ where $c(k^2)$ is a positive number depending on $k^2$, so the exterior instability in arbitrary regularity below the threshold cannot be achieved. 
 
However, the above heuristic cannot be adapted directly to the proof of interior instability. The main additional difficulty is that, when $\mathcal{N}_1$ is strictly before $\mathcal{N}$, we will not have $\Omega^2(0,u)=|u|^{k^2}$. But instead, we have a weaker inequality
$$|u|^{k^2+\epsilon}\le \Omega^2(0,u)\le |u|^{k^2-\epsilon},$$
which only holds for $u\in [u_0,u_\epsilon]$ where $u_\epsilon\ne0$, and $\epsilon\to0$, $u_\epsilon\to0$ as $\mathcal{N}_1$ approaches $\mathcal{N}$. In fact, we have a stronger \emph{relative scaling control}
$$\left(\frac{|u|}{|u'|}\right)^{k^2+\epsilon}\le \frac{\Omega^2(0,u)}{\Omega^2(0,u')}\le \left(\frac{|u|}{|u'|}\right)^{k^2-\epsilon}, |u_\epsilon|\le |u|\le |u'|\le|u_0|.$$
The best we can hope for is to carry out the above heuristic with $k^2$ replaced by $k^2-\epsilon$. Since $\epsilon$ can be chosen arbitrarily small, we should very carefully choose the family of perturbations so that their size at the threshold is still uniformly bounded. One additional difficulty is that because the upper and lower bounds for $\Omega^2$ have different exponents.  Since invoking the lower bound would result in a loss in the power of $|u|$, making the estimates difficult to track, we will carefully arrange the argument so that only the upper bound $|u|^{k^2-\epsilon}$ is used throughout. The only place we use the lower bound is in the proof of Theorem \ref{main2} in Section \ref{sec:trappedsurface}.

\subsubsection{The anisotropic setting}

The proof of Theorem   \ref{main2} rely on one more ingredient, which can be traced back to the work of Klainerman--Luk--Rodnianski \cite{K-L-R} on anisotropic trapped surface formation, which was then generalized in \cite{An-Han, A, An3}. The main contribution of \cite{K-L-R} is relax the lower bound \eqref{Chr08lower} needed for trapped surface formation to 
\begin{equation*} 
\sup_{\vartheta\in S^2}\int_{0}^\delta|\chih|^2(\ub,\vartheta)\D\ub>0,\end{equation*}
or more quantitatively, 
\begin{equation}\label{K-L-Rlower} 
\inf_{\vartheta\in B_p(\epsilon)\subset S^2}\int_{0}^\delta|\chih|^2(\ub,\vartheta)\D\ub\ge M_*>0,\end{equation}
where $p\in S^2$ and $\epsilon, M_*$ be two constants.  It was shown that all spherical sections of the incoming null cone $\ub=\delta$ are not trapped but a closed trapped surface does exist along $\ub=\delta$. This closed trapped surface is obtained by deforming the foliation on $\ub=\delta$ and the deepest location of this trapped surface is around $|u|\approx M_*\epsilon^5$, which lies in the existence region established by Christodoulou \cite{Chr08} if $\delta$ is chosen sufficiently small (depending on $M_*,\epsilon$).  Note that under the condition \eqref{K-L-Rlower}  the square of the $L^2$ norm of $\chih$ is $\gtrsim M_*\epsilon^2$.

Now we consider an analogue to Theorem \ref{main2} on the $L^2$-trapped surface formation, that is, finding closed trapped surface when $\chih$ is a $\delta$-family data so that the support of $\chih$ shrinks to a point $p$ at $S_0=\{\ub=0, u=u_0\}$ as $\delta\to0$. At the same time, we require that the $L^2$ norm of $\chih_\delta$ is still larger than a positive constant independent of $\delta$. A possible way to define 
$$\chih_\delta(\ub,\vartheta)=\delta^{-\frac{1}{2}-\alpha_1}\chih_0(\frac{\ub}{\delta}, \frac{\vartheta}{\delta^{\alpha_2}}),$$
where $\chih_0$ is some fixed seed data with its support lying within a small neighborhood of $\vartheta=0$.  In fact, the case $\alpha_1=\alpha_2=0$ is exactly the family of data considered in Christodoulou's work \cite{Chr08}, and in a later generalization of Klainerman--Rodnianski \cite{K-R12}, existence theorem also up to $|u|\approx c$ is established in the case $\alpha_1=0, \alpha_2=\frac{1}{2}$.   The $L^2$ norm of $\chih_\delta$ is $\approx \|\chih_0\|_{L^2}\delta^{\alpha_2-\alpha_1}$ so we may set $\alpha_1=\alpha_2$ in which case $\|\chih_\delta\|_{L^2(C_{u_0})}\approx \|\chih_0\|_{L^2}$. To find a closed trapped surface,  we may apply the argument in \cite{K-L-R} for $\epsilon=\delta^{\alpha_2}$. Then the closed trapped surface is expected to form before $|u|\approx \|\chih_0\|_{L^2}\delta^{5\alpha_2}$. But this lies outside the existence region established in \cite{Chr08}. Moreover, the introduction of $\alpha_1$ and $\alpha_2$ may further shrink the existence region. At this point, the existence result in \cite{An-Luk} becomes crucial, since it extends the existence region to $|u|\approx\delta^{\frac{1}{2}}b$, thereby covering $|u|\approx \|\chih_0\|_{L^2}\delta^{5\alpha_2}$ whenever $5\alpha_2<\frac{1}{2}$. Although the introduction of $\alpha_1,\alpha_2$ may reduce the existence region, but we may choose $\alpha_2$ small enough so that the desired construction remains compatible with the reduced existence region, since we only need to construct one particular family of $\chih_\delta$ rather than pursue an optimal $\alpha_2$.  We remark that the techniques developed in \cite{K-R12} may help retain the existence region when introducing $\alpha_2>0$ even when it is not small, but we do not pursue this direction here.

 The same heuristic works when the background is $k$-self-similar, which is the main idea in proving Theorem \ref{main2}. The key point is to establish a trapped surface theorem allowing the initial data to satisfy upper and lower bounds with different scales. We need to make an additional effort to show that we can choose suitable parameters so that the threshold can be approached, especially when interior instability is considered.  We remark that the work \cite{A} of An considered the anisotropic (exterior) problem when the background is $k$-self-similar, with the number $\alpha_2$ being set to be zero.  In the language of Theorem \ref{main1} and \ref{main2},  the family constructed in \cite{A} of exterior gravitational perturbations converges to zero in $C^\infty$ outside an $\epsilon$-ball centered at  $p\in S_0$ and $C^\alpha_{\ub}C^\infty_\vartheta$ for some $\alpha$. The main technical novelties of this paper are to push the exponent arbitrarily close to the threshold, allow the ball radius $\epsilon$ to depend on $\delta$ and shrink to zero, and adapt the argument to the interior instability setting.

\subsection*{Acknowledgement}

 The authors are supported by National Key R\&D Program of China (No. 2022YFA1005400) and NSFC (12326602, 12141106).

	\section{Equations}
	
	In this section, we list all equations we need, which are written in double null frames.  The notation and setup in this section are taken largely from our previous work \cite{Li-Liu1}, following the notations in Christodoulou's monograph \cite{Chr08}, which may be slightly different from those in other literature.  We assume the solution $(\mathcal{M},g)$ is foliated by two optical functions, $u$ and $\ub$, that is
 $$g(\nabla u,\nabla u)=g(\nabla\ub,\nabla\ub)=0.$$
 In addition, we require that $u$ and $\ub$  increase towards the future. We use $C_u$ to denote the  null hypersurfaces that are the level sets of $u$ and use ${\Cb}_{\ub}$ to denote the incoming null hypersurfaces that are the level sets of $\ub$. We denote the intersection $S_{\ub,u}=\Cb_{\ub} \cap C_u$, which is a  space-like two-sphere. The space-time metric $g$ induces a Riemannian metric $\gs$ on $S_{\ub,u}$ and $\epsilons$ is the volume form of $\gs$ on $S_{\ub,u}$. We use   $\nablas$ to denote the  covariant derivative (with respect to $\gs$) on $S_{\ub,u}$ and $\Ks$ is the Gauss curvature of $S_{\ub,u}$.

The lapse function $\Omega$ is defined by the formula
$$ \Omega^{-2}=-2g(\nabla\ub,\nabla u).$$
  We then define the normalized null pair $(\Lbh, \Lh)$ by
  $$e_3=\Lbh=-2\Omega\nabla\ub,\ e_4=\Lh=-2\Omega\nabla u.$$
  We also define another null pair
  $$\Lb=\Omega \Lbh,\ L=\Omega \Lh.$$
 The flows generated by $\Lb$ and $L$ preserve the double null foliation. Moreover, we define
 $$\Lb'=-2\nabla\ub=\Omega^{-1}\Lbh,\ \Lb'=-2\nabla u=\Omega^{-1}\Lh$$
 which is geodesic. 

We can introduce the double null coordinate system $(\ub,u,\vartheta^A)$ on $\mathcal{M}$, where $A,B,C,\cdots$ are denoted an index from $1$ to $2$. In such a coordinate system, the Lorentzian metric $g$ takes the following form
\begin{align}\label{doublenullmetric}
g=-2\Omega^2(\D\ub\otimes\D u+\D u\otimes\D \ub)+\gs_{AB}(\D\vartheta^A-b^A\D u)\otimes(\D\vartheta^B-b^B\D u),
\end{align}
where the vectorfield $b=b^A\partial_{\vartheta^A}$ is tangent to $S_{\ub,u}$. The null vectors $\Lb$ and $L$ can be computed as $\Lb=\partial_u+b^A\partial_{\vartheta^A}$ and $L=\partial_{\ub}$.

We call $\psi$ to be a tangential tensorfield if $\psi$ is \textit{a priori} a tensorfield defined on the space-time $\mathcal{M}$ and all the possible contractions of $\psi$ with either $\Lbh$ or $\Lh$ are zeros. By choosing a tangential frame $(e_1,e_2)$, which is tangent to $S_{\ub,u}$, each tangential tensorfield can be express as, for example $\psi_A, \psi_{AB}, \psi^A,\cdots$. We use $D\psi$ and $\Db\psi$ to denote the projection to $S_{\ub,u}$ of usual Lie derivatives $\mathcal{L}_L\psi$ and $\mathcal{L}_{\Lb}\psi$, which are again tangential tensorfields. The notation $\nablas\psi$ also makes sense when $\psi$ is tangential. We also denote $\nablas_X\psi$ to be the projection of the spacetime covariant derivative $\nabla_X\psi$ to $S_{\ub,u}$, where $X$ is not necessarily a tangential vectorfield.

Using the null frame $(e_1,e_2,\Lbh,\Lh)$, the connection coefficients can be decomposed as the following tangential tensorfields:
\begin{align*}
\chi_{AB}&=g(\nabla_A\Lh,e_B),\quad \eta_A=-\frac{1}{2}g(\nabla_{\Lbh}e_A,\Lh),\quad \omega=\frac{1}{2}\Omega g(\nabla_{\Lh}\Lbh,\Lh),\\
\chib_{AB}&=g(\nabla_A\Lbh,e_B), \quad\etab_A=-\frac{1}{2}g(\nabla_{\Lh}e_A,\Lbh), \quad\omegab=\frac{1}{2}\Omega g(\nabla_{\Lbh}\Lh,\Lbh).
\end{align*}
We also define the following normalized quantities:
$$\chi'=\Omega^{-1}\chi,\ \chib'=\Omega^{-1}\chib,\ \zeta=\frac{1}{2}(\eta-\etab).$$
 The ($2$-dimensional) trace of $\chi$ and $\chib$ are denoted by
 $$\tr\chi = \gs^{AB}\chi_{AB},\ \tr\chib = \gs^{AB}\chib_{AB},$$
 and the trace-free parts of $\chi$ and $\chib$ are denoted by
 $$\chih=\chi-\frac{1}{2}\tr\chi\gs,\ \chibh=\chib-\frac{1}{2}\tr\chib\gs.$$ 
  By definition, we can check directly the following useful identities :
  $$\ds\log\Omega=\frac{1}{2}(\eta+\etab),\ D\log\Omega=\omega,\ \Db\log\Omega=\omegab.$$

Recall that the Weyl curvature tensor is
$$\mathbf{W}_{\alpha\beta\gamma\delta}=\mathbf{R}_{\alpha\beta\gamma\delta}+\frac{\mathbf{R}}{6}(g_{\alpha\gamma}g_{\beta\delta}-g_{\alpha\delta}g_{\beta\gamma})+\frac{1}{2}(g_{\alpha\delta}\mathbf{Ric}_{\beta\gamma}+g_{\beta\gamma}\mathbf{Ric}_{\alpha\delta}-g_{\alpha\gamma}\mathbf{Ric}_{\beta\delta}-g_{\beta\delta}\mathbf{Ric}_{\alpha\gamma}).$$
It can be decomposed as
\begin{align*}
\alpha_{AB}&=\mathbf{W}(e_A,\Lh,e_B,\Lh),\quad\beta_A=\frac{1}{2}\mathbf{W}(e_A,\Lh,\Lbh,\Lh),\quad\rho=\frac{1}{4}\mathbf{W}(\Lbh,\Lh,\Lbh,\Lh),\\
\alphab_{AB}&=\mathbf{W}(e_A,\Lbh,e_B,\Lbh),\quad\betab_A=\frac{1}{2}\mathbf{W}(e_A,\Lbh,\Lbh,\Lh),\quad\sigma=\frac{1}{4}\mathbf{W}(\Lbh,\Lh,e_A,e_B)\epsilons^{AB}.
\end{align*}
By the algrebric property of Weyl tensor, $\alpha$ and $\alphab$ are trace-free.

    The followings are the \textit{null structure equations} (where  $\Ks$ is the Gauss curvature of $S_{\ub,u}$):
    \begin{align*}
    \Db(\Omega\tr\chib)&=-\frac{1}{2}(\Omega\tr\chib)^2+2\omegab\Omega\tr\chib-|\Omega\chibh|^2-2(\Lb\phi)^2,\\
    D\tr\chi'&=-\frac{1}{2}(\Omega\tr\chi')^2-|\chih|^2-2(\Lh\phi)^2,\\
    \Dbh(\Omega\chih)&=\Omega^2(\nablas \tensor \eta + \eta \tensor \eta +\frac{1}{2}\tr\chib\chih-\frac{1}{2}\tr\chi \chibh+\ds\phi\tensor\ds\phi),\\
    \Dh(\Omega\chibh)&=\Omega^2(\nablas \tensor \etab + \etab \tensor \etab +\frac{1}{2}\tr\chi\chibh-\frac{1}{2}\tr\chib \chih+\ds\phi\tensor\ds\phi),\\
    \Db(\Omega\tr\chi)&=\Omega^2(2\divs\eta+2|\eta|^2-\tr\chi\tr\chib-2\Ks+2|\ds\phi|^2),\\
    D(\Omega\tr\chib)&=\Omega^2(2\divs\etab+2|\etab|^2-\tr\chi\tr\chib-2\Ks+2|\ds\phi|^2),\\
    D\eta &= (\Omega\chi)\cdot\etab-(\Omega\beta+L\phi\ds\phi),\\
    \Db\etab &= (\Omega\chib) \cdot\eta+(\Omega\betab-\Lb\phi\ds\phi),\\
    D  \omegab &=\Omega^2(2(\eta,\etab)-|\eta|^2-(\rho+\frac{1}{6}\mathbf{R}+\Lh\phi\Lbh\phi)),\\
    \Db  \omega &=\Omega^2(2(\eta,\etab)-|\etab|^2-(\rho+\frac{1}{6}\mathbf{R}+\Lh\phi\Lbh\phi)),\\
    \Ks+\frac{1}{4}\tr \chi\tr\chib-\frac{1}{2}(\chih,\chibh)&=-(\rho+\frac{1}{6}\mathbf{R})+|\ds\phi|^2,\\
    \divs(\Omega\chih)&=\frac{1}{2}\Omega^2\ds\tr \chi'+\Omega\chih\cdot\etab+\frac{1}{2}\Omega\tr \chi\eta-(\Omega\beta-L\phi\ds\phi),\\
    \divs(\Omega\chibh)&=\frac{1}{2}\Omega^2\ds\tr \chib'+\Omega\chibh\cdot\eta+\frac{1}{2}\Omega\tr \chib\etab+(\Omega\betab+\Lb\phi\ds\phi).\\
    \end{align*}

	The contracted second Bianchi identity will accordingly be the following inhomogeneous equation:
	\begin{equation*}
		\nabla^{\alpha}\mathbf{W}_{\alpha\beta\gamma\delta}=\nabla_{[\gamma}\mathbf{R}_{\delta]\beta}+\frac{1}{6}g_{\beta[\gamma}\nabla_{\delta]}\mathbf{R}
	\end{equation*}
	This equation can be decomposed using the null frame $(e_1,e_2,e_3,e_4)$ into components, which we call \textit{null Bianchi equations}:

    \begin{align*}
    &\Dbh\alpha-\frac{1}{2}\Omega\tr\chib \alpha+2\omegab\alpha+\Omega\{-\nablas\tensor\beta -(4\eta+\zeta)\tensor \beta+3\chih (\rho+\frac{1}{6}\mathbf{R})+3{}^*\chih \sigma\}
    \\=&-\Omega\left(\nablas\Lh\phi\tensor\ds\phi-\nablas\tensor\ds\phi\Lh\phi\right.\\
    &\left.-\frac{1}{2}\tr\chi\ds\phi\tensor\ds\phi-\chih\cdot\ds\phi\tensor\ds\phi+\frac{3}{2}\chih\Lh\phi\Lbh\phi-\chih|\ds\phi|^2+\frac{1}{2}\chibh(\Lh\phi)^2+\zeta\tensor\ds\phi\Lh\phi\right),\\
    &D\beta+\frac{3}{2}\Omega\tr\chi\beta-\Omega\chih\cdot\beta-\omega\beta-\Omega\{\divs\alpha+(\etab+2\zeta)\cdot\alpha\}\\
    =&\Omega\left(\ds\phi\Lh\Lh\phi-\Lh\phi\nablas\Lh\phi+\frac{1}{2}\tr\chi\ds\phi\Lh\phi+\chih\cdot\ds\phi\Lh\phi-\Omega^{-1}\omega\Lh\phi\ds\phi-(\Lh\phi)^2\zeta\right),\\
    &\Db\beta+\frac{1}{2}\Omega\tr\chib\beta-\Omega\chibh \cdot \beta+\omegab \beta-\Omega\{\ds (\rho+\frac{1}{6}\mathbf{R})+{}^*\ds \sigma+3\eta(\rho+\frac{1}{6}\mathbf{R})+3{}^*\eta\sigma+2\chih\cdot\betab\}\\
    =&-\Omega\left(\ds\phi\Deltas\phi-\nablas\Lbh\phi\Lh\phi-\frac{1}{2}\tr\chi\Lbh\phi\ds\phi+\chibh\cdot\ds\phi\Lh\phi-(\eta-\zeta)\Lbh\phi\Lh\phi+\eta|\ds\phi|^2\right),\\
    &D(\rho+\frac{1}{6}\mathbf{R})+\frac{3}{2}\Omega\tr\chi (\rho+\frac{1}{6}\mathbf{R})-\Omega\{\divs \beta+(2\etab+\zeta,\beta)-\frac{1}{2}(\chibh,\alpha)\}\\=&-\Omega\left(\Lh\phi\Delta\phi-\nablas\Lh\phi\cdot\ds\phi+\chih\cdot\ds\phi\cdot\ds\phi-\frac{1}{2}\tr\chib(\Lh\phi)^2-\zeta\cdot\ds\phi\Lh\phi\right),\\
    &D\sigma+\frac{3}{2}\Omega\tr\chi\sigma+\Omega\{\curls\beta+(2\etab+\zeta,{}^*\beta)-\frac{1}{2}\chibh\wedge\alpha\}\\
    =&-\Omega\left(\nablas\Lh\phi\wedge\ds\phi-\chih\cdot\ds\phi\wedge\ds\phi+\zeta\wedge\ds\phi\Lh\phi\right),\\
    &D\betab+\frac{1}{2}\Omega\tr\chi\betab-\Omega\chih \cdot \betab+\omega \betab+\Omega\{\ds (\rho+\frac{1}{6}\mathbf{R})-{}^*\ds \sigma+3\etab(\rho+\frac{1}{6}\mathbf{R})-3{}^*\etab\sigma-2\chibh\cdot\beta\}\\
    =&\Omega\left(\ds\phi\Deltas\phi-\nablas\Lh\phi\Lbh\phi-\frac{1}{2}\tr\chib\Lh\phi\ds\phi+\chih\cdot\ds\phi\Lbh\phi-(\etab+\zeta)\Lh\phi\Lbh\phi+\etab|\ds\phi|^2\right),\\
    &\Db(\rho+\frac{1}{6}\mathbf{R})+\frac{3}{2}\Omega\tr\chib (\rho+\frac{1}{6}\mathbf{R})+\Omega\{\divs \betab+(2\eta-\zeta,\betab)+\frac{1}{2}(\chih,\alphab)\}\\
    =&-\Omega\left(\Lbh\phi\Delta\phi-\nablas\Lbh\phi\cdot\ds\phi+\chibh\cdot\ds\phi\cdot\ds\phi-\frac{1}{2}\tr\chi(\Lbh\phi)^2+\zeta\cdot\ds\phi\Lbh\phi\right),\\
    &\Db\sigma+\frac{3}{2}\Omega\tr\chib\sigma+\Omega\{\curls\betab+(2\eta-\zeta,{}^*\betab)+\frac{1}{2}\chih\wedge\alphab\}\\
    =&\Omega\left(\nablas\Lbh\phi\wedge\ds\phi-\chibh\cdot\ds\phi\wedge\ds\phi-\zeta\wedge\ds\phi\Lbh\phi\right),\\
    &\Db\betab+\frac{3}{2}\Omega\tr\chib\betab-\Omega\chibh\cdot\betab-\omegab\betab+\Omega\{\divs\alphab+(\eta-2\zeta)\cdot\alphab\}\\
    =&-\Omega\left(\ds\phi\Lbh\Lbh\phi-\Lbh\phi\nablas\Lbh\phi+\frac{1}{2}\tr\chib\ds\phi\Lbh\phi+\chibh\cdot\ds\phi\Lbh\phi-\Omega^{-1}\omegab\Lbh\phi\ds\phi+(\Lbh\phi)^2\zeta\right),\\
    &\Dh\alphab-\frac{1}{2}\Omega\tr\chi \alphab+2\omega\alphab+\Omega\{\nablas\tensor\betab +(4\etab-\zeta)\tensor \betab+3\chibh (\rho+\frac{1}{6}\mathbf{R})-3{}^*\chibh \sigma\}\\
    =&-\Omega\left(\nablas\Lbh\phi\tensor\ds\phi-\nablas\tensor\ds\phi\Lbh\phi\right.\\
    &\left.-\frac{1}{2}\tr\chib\ds\phi\tensor\ds\phi-\chibh\cdot\ds\phi\tensor\ds\phi+\frac{3}{2}\chibh\Lbh\phi\Lh\phi-\chibh|\ds\phi|^2+\frac{1}{2}\chih(\Lbh\phi)^2-\zeta\tensor\ds\phi\Lbh\phi\right).\\
     \end{align*}

In this paper, we consider instead the following \textit{renormalized null Bianchi equations}:
    \begin{align*}
    D\Ks&+\Omega\tr\chi \Ks+\divs(\Omega\beta-L\phi\nablas\phi)\\
    &-\Omega\chih\cdot\nablas\etab+\frac{1}{2}\Omega\tr\chi\divs\etab+(\Omega\beta-L\phi\nablas\phi)\cdot\etab-\Omega\chih\cdot\etab\cdot\etab+\frac{1}{2}\Omega\tr\chi|\etab|^2=0\\
     D\sigmac&+\frac{3}{2}\Omega\tr\chi\sigmac+\curls(\Omega\beta-L\phi\nablas\phi)+\frac{1}{2}\Omega\chih\wedge(\etab\tensor\etab+\nablas\tensor\etab)\\
     &+\etab\wedge(\Omega\beta-L\phi\nablas\phi)+2\nablas L\phi\wedge\nablas\phi=0\\
     \Db(\Omega\beta-L\phi\nablas\phi)&+\frac{1}{2}\Omega\tr\chib(\Omega\beta-L\phi\nablas\phi)-\Omega\chibh\cdot(\Omega\beta-L\phi\nablas\phi)+\Omega^2\ds \Ks-\Omega^2{}^*\ds\sigmac\\
     &+3\Omega^2(\eta \Ks-{}^*\eta\sigmac)-\frac{1}{2}\Omega^2(\ds(\chih,\chibh)+{}^*\ds(\chih\wedge\chibh))-\frac{3}{2}\Omega^2(\eta(\chih,\chibh)+{}^*\eta(\chih\wedge\chibh))\\
     &+\frac{1}{4}\Omega^2\ds(\tr\chi\tr\chib)+\frac{3}{4}\Omega^2\tr\chi\tr\chib\eta-2\Omega\chih\cdot(\Omega\betab+\Lb\phi\nablas\phi)\\
    =&-2\Omega^2\Deltas\phi\nablas\phi+\Omega^2\ds|\ds\phi|^2-2\Omega\chih\cdot\nablas\phi\Lb\phi+\Omega\tr\chi\Lb\phi\nablas\phi\\
    &-2\Omega^2\eta\cdot\nablas\phi\nablas\phi+2\Omega^2\eta|\ds\phi|^2,\\
    \Db(\Ks-\frac{1}{|u|^2})&+\frac{3}{2}\Omega\tr\chib (\Ks-\frac{1}{|u|^2})+(\Omega\tr\chib+\frac{2}{|u|})\frac{1}{|u|^2}\\
    =&\divs(\Omega\betab+\Lb\phi\nablas\phi)+\Omega\chibh\cdot\nablas\eta+\frac{1}{2}\Omega\tr\chib\mu\\
    &+(\Omega\betab+\Lb\phi\nablas\phi)\cdot\eta+\Omega\chibh\cdot\eta\cdot\eta-\frac{1}{2}\Omega\tr\chib|\eta|^2,\\
    \Db\sigmac&+\frac{3}{2}\Omega\tr\chi\sigmac+\curls(\Omega\betab+\Lb\phi\nablas\phi)\\
    &-\frac{1}{2}\Omega\chibh\wedge(\eta\tensor\eta+\nablas\tensor\eta)+\eta\wedge(\Omega\betab+\Lb\phi\nablas\phi)-2\nablas \Lb\phi\wedge\nablas\phi=0\\
    D(\Omega\betab+\Lb\phi\nablas\phi)&+\frac{1}{2}\Omega\tr\chi(\Omega\betab+\Lb\phi\nablas\phi)-\Omega\chih\cdot(\Omega\betab+\Lb\phi\nablas\phi)-\Omega^2\ds \Ks-\Omega^2{}^*\ds\sigmac\\
    &-3\Omega^2(\etab \Ks+{}^*\etab\sigmac)+\frac{1}{2}\Omega^2(\ds(\chih,\chibh)-{}^*\ds(\chih\wedge\chibh))+\frac{3}{2}\Omega^2(\etab(\chih,\chibh)+{}^*\etab(\chih\wedge\chibh))\\
    &-\frac{1}{4}\Omega^2\ds(\tr\chi\tr\chib)-\frac{3}{4}\Omega^2\tr\chi\tr\chib\etab-2\Omega\chibh\cdot(\Omega\beta-L\phi\nablas\phi)\\
    =&2\Omega^2\Deltas\phi\nablas\phi-\Omega^2\ds|\ds\phi|^2+2\Omega\chibh\cdot\nablas\phi L\phi-\Omega\tr\chib L\phi\nablas\phi\\
    &+2\Omega^2\etab\cdot\nablas\phi\nablas\phi-2\Omega^2\etab|\ds\phi|^2.
   \end{align*}
   where $\check{\sigma}=\sigma-\frac{1}{2}\chih\wedge\chibh$ and $\mu$ is defined through
   $$\divs\eta=\Ks-\frac{1}{|u|^2}-\mu.$$

	We can rewrite the wave equations in the double null coordinate system, which reads
	\begin{align*}
		\Db L\phi+\frac{1}{2}\Omega\tr\chib L\phi&=\Omega^2\Deltas\phi+2\Omega^2(\eta,\ds\phi)-\frac{1}{2}\Omega\tr\chi\Lb\phi,\\
		D\ds\phi&=\nablas L\phi,\\
		D\Lb\phi+\frac{1}{2}\Omega\tr\chi\Lb\phi&=\Omega^2\Deltas\phi+2\Omega^2(\etab,\ds\phi)-\frac{1}{2}\Omega\tr\chib L\phi,\\
		\Db\ds\phi&=\nablas\Lb\phi.
	\end{align*}

	\section{Existence theorem and Geometric lemmas}
	
	\subsection{The existence theorem}
	
	In this section, we establish the main existence theorem. It follows by standard argument once the following a priori estimates are derived. The set up of the theorem is as follows. Consider the initial value problem of \eqref{ES}. Let $k^2\in(0,\frac{1}{3})$ be fixed. Suppose that we have initial data imposed on $\Cb_0\cup C_{u_0}$. The data on $\Cb_0$ is spherically symmetric. The function $u$ is chosen on $\Cb_0$ such that $u=-r$, and suppose that  on $\Cb_0$
\begin{equation}\label{initialboundCb0}|r\Lb\phi|\le 1,\quad |r\Omega^{-2}L\phi|\le K, \quad |r\Omega^{-1}\tr\chi|\le 2\end{equation}
for some $K\ge1$ and
\begin{equation}\label{omegablower}-2\omegab\big|_{\Cb_0}\ge (k^2-\epsilon) |u|^{-1}, u\in [u_0, u_1]\end{equation}
 for some $\epsilon\in[0,k^2)$ and $u_1\in (u_0,0)$. We choose $p\in(0,k^2-\epsilon)$. Consequently, we have the relative scaling control
 \begin{equation}\label{relativecontrol} \frac{\Omega^2(0,u)}{\Omega^2(0,u')}\le \left(\frac{|u|}{|u'|}\right)^{k^2-\epsilon}\le \left(\frac{|u|}{|u'|}\right)^{p}, |u_1|\le |u|\le |u'|\le|u_0|.\end{equation}
 We may assume $\Omega^2_0(u_0)=|u_0|^p$ then we have $\Omega^2_0(u)\le|u|^p, u\in [u_0,u_1]$. On $C_{u_0}$, we choose $\omega\equiv0$ and hence $\Omega$ is constant along $C_{u_0}$, and $\chih$, $L\phi$ obey for $\ub\in[0,\delta]$
\begin{equation}\label{initialbound} |\nablas^i(\Omega\chih)|, |\nablas^i \widetilde{L\phi}|\le  a\delta^{\gamma}\end{equation}
for some $a\ge K$ and for all $i\le N+3$ where $N$ is some large integer and $\gamma\ge\frac{p}{1-p}$. Here $a$ is allowed to depend on $\delta$. 

\begin{remark}
Throughout the paper, only the case $\gamma=\frac{p}{1-p}$ is needed. Nevertheless, we allow $\gamma\geq\frac{p}{1-p}$ for possible future applications. 
\end{remark}
\begin{remark}
The constant $K$ is introduced for application in Section \ref{sec:trappedsurface} for trapped surface formation. In proving existence theorem, we simply need the bound $K\le a$. In view of the $k$-self-similar background, $K$ is almost $\frac{1}{k}$. 
\end{remark}
    
	We introduce the following norms of different quantities. The integral is defined to be
	$$\int_{C_u}:=\int_0^\delta\D\ub\int_{S_{\ub,u}}\D\mu_{\gs}, \quad\int_{\Cb_{\ub}}:=\int_{u_0}^{u_1}\D\ub\int_{S_{\ub,u}}\D\mu_{\gs},$$ $$\int_{\mathcal{M}}:=\int_0^\delta\D\ub\int_{u_0}^{u_1}\D u\int_{S_{\ub,u}}\D\mu_{\gs}.$$
	Note that the integration on $\mathcal{M}$ is not taken with respect to the spacetime volume form.
    \begin{itemize}
    \item\textit{Curvature components:}
    \begin{align*}
     \|\beta\|^2_{\mathcal{A}}\triangleq\sup_u\sup_{0\le i\le N}&\left[ a^{-2}\delta^{-1-2\gamma}\int_{C_u}|u|^{2+2\gamma-p(3+2\gamma)+2i}|\nablas^i(\Omega^2\beta-L\phi\Omega\nablas\phi)|^2\right]
    \end{align*}
    \begin{align*}
    \|\Ks\|^2_{\mathcal{A}}\triangleq\sup_{\ub}\sup_{0\le i\le N}&a^{-2}\delta^{-1-2\gamma}\int_{\Cb_{\ub}} |u|^{2+2\gamma-p(3+2\gamma)+2i}|\nablas^i(\Omega^2\widetilde{\Ks})|^2\\
    +\sup_u\sup_{1\le i\le N}& a^{-3}\delta^{-3-2\gamma}\int_{C_u} |u|^{4+2\gamma-2p(2+\gamma)+2i} |\nablas^i(\Omega^2\widetilde{\Ks})|^2
    \end{align*}
        \begin{align*}
    \|\sigma\|^2_{\mathcal{A}}\triangleq\sup_{0\le i\le N}&\left[\sup_{\ub}a^{-2}\delta^{-1-2\gamma}\int_{\Cb_{\ub}} |u|^{2+2\gamma-p(3+2\gamma)+2i}|\nablas^i(\Omega^2\sigmac)|^2\right.\\
    &\left.+ \sup_ua^{-3}\delta^{-3-2\gamma}\int_{C_u} |u|^{4+2\gamma-2p(2+\gamma)+2i} |\nablas^i(\Omega^2\sigmac)|^2\right]
    \end{align*}
    \begin{align*}
    \|\betab\|^2_{\mathcal{A}}\triangleq \sup_{\ub}\sup_{0\le i\le N} a^{-3}\delta^{-3-2\gamma}\int_{\Cb_{\ub}} |u|^{4+2\gamma-2p(2+\gamma)+2i}|\nablas^i(\Omega^2 \betab+\Lb\phi\Omega\nablas\phi)|^2
    \end{align*}
  We introduce the notation
    $$\mathcal{R}^2\triangleq\sum_{\Psi\in\{\beta,\Ks,\sigma,\betab\}}\|\Psi\|^2_{\mathcal{A}}.$$
    Note that the norm of $\Omega^2\widetilde{K}$ on $C_u$ only defines for $i>0$. $\Omega^2\widetilde{K}$ obeys worse estimate. 
     
    \item\textit{Connection coefficients:}
    \begin{itemize}
    \item [1)]lower order derivatives
     $$\|\chih\|^2_{\mathcal{B}}\triangleq\sup_{\ub,u}\sup_{0\le i\le N} a^{-2}\delta^{-2\gamma}\int_{S_{\ub,u}}|u|^{2\gamma-2p(1+\gamma)+2i}|\nablas^i(\Omega\chih)|^2$$
          $$\|\omega\|^2_{\mathcal{B}}\triangleq\sup_{\ub,u}\sup_{0\le i\le N} a^{-2}\int_{S_{\ub,u}}|u|^{-2p+2i}|\nablas^i(\omega)|^2$$
     $$\|\eta,\etab\|^2_{\mathcal{B}}\triangleq\sup_{\ub,u}\sup_{0\le i\le N} a^{-2}\delta^{-2-2\gamma}\int_{S_{\ub,u}}|u|^{2+2\gamma-p(3+2\gamma)+2i}|\nablas^i(\Omega\eta,\Omega\etab)|^2$$
     $$\|\chibh\|^2_{\mathcal{B}}\triangleq\sup_{\ub,u}\sup_{0\le i\le N} a^{-2}\delta^{-2-2\gamma}\int_{S_{\ub,u}}|u|^{2+2\gamma-2p(1+\gamma)+2i}|\nablas^i(\Omega\chibh)|^2$$
        $$\|\tr\chib,\omegab\|^2_{\mathcal{B}}\triangleq\sup_{\ub,u}\sup_{0\le i\le N} a^{-2}\delta^{-2}\int_{S_{\ub,u}}|u|^{2-2p+2i}|\nablas^i(\widetilde{\Omega\tr\chib},\omegab)|^2$$
     $$\|\tr\chi\|^2_{\mathcal{B}}\triangleq\sup_{\ub,u}\sup_{0\le i\le N} a^{-4}\delta^{-2}\int_{S_{\ub,u}}|u|^{2-4p+2i}|\nablas^i(\widetilde{\Omega\tr\chi})|^2$$
     $$\Os^2\triangleq\sum_{\psi\in\{\chih,\tr\chi,\chibh,\tr\chib,\eta,\etab,\omega,\omegab\}}\|\psi\|^2_{\mathcal{B}}$$
    
     \item [2)]top order derivatives
     $$\|\tr\chi\|^2_{\mathcal{C}}\triangleq\sup_{\ub,u}\sup_{1\le i\le N+1}a^{-4}\delta^{-2-2\gamma}\int_{S_{\ub,u}}|u|^{2+2\gamma-p(5+2\gamma)+2i}|\nablas^i(\Omega^2\tr\chi)|^2$$
     $$\|\tr\chib\|^2_{\mathcal{C}}\triangleq\sup_{\ub,u}\sup_{1\le i\le N+1} a^{-2}\delta^{-2-2\gamma}\int_{S_{\ub,u}}|u|^{2+2\gamma-p(3+2\gamma)+2i}|\nablas^i(\Omega^2\tr\chib)|^2$$
     $$\|\chih,\omega\|^2_{\mathcal{C}}\triangleq\sup_{u}\sup_{1\le i\le N+1} a^{-2}\delta^{-1-2\gamma}\int_{C_u}|u|^{2\gamma-p(3+2\gamma)+2i}|\nablas^i(\Omega^2\chih,\Omega\omega)|^2$$
     $$\|\omegab\|^2_{\mathcal{C}}\triangleq\sup_{\ub}\sup_{1\le i\le N+1}a^{-2} \delta^{-2-2\gamma}\int_{\Cb_{\ub}}|u|^{1+2\gamma-p(3+2\gamma)+2i}|\nablas^i(\Omega\omegab)|^2$$
     $$\|\chibh\|^2_{\mathcal{C}}\triangleq\sup_{\ub}\sup_{1\le i\le N+1} a^{-2}\delta^{-2-2\gamma}\int_{\Cb_{\ub}}|u|^{2\gamma-p(3+2\gamma)+2i}|\nablas^i(\Omega^2\chibh)|^2$$
     \begin{align*}\|\eta,\etab\|^2_{\mathcal{C}}\triangleq\sup_{1\le i\le N+1} &\left[ \sup_{\ub}a^{-2}\delta^{-1-2\gamma}\int_{\Cb_{\ub}}|u|^{2\gamma-p(3+2\gamma)+2i}|\nablas^i(\Omega^2\eta)|^2\right.\\&\left.+\sup_{u}a^{-2}\delta^{-2-2\gamma}\int_{C_u}|u|^{1+2\gamma-p(3+2\gamma)+2i}|\nablas^i(\Omega^2\eta,\Omega^2\etab)|^2\right]\end{align*}
     $$\mathcal{O}^2\triangleq\sum_{\psi\in\{\chih,\tr\chi,\chibh,\tr\chib,\eta,\etab,\omega,\omegab\}}\|\psi\|^2_{\mathcal{C}}$$
     \end{itemize}

     \item\textit{Scalar field:}
     \begin{itemize}
     \item [1)]lower order derivatives
     $$\|L\phi\|^2_{\mathcal{D}}\triangleq\sup_{\ub,u}\sup_{0\le i\le N} a^{-2}\int_{S_{\ub,u}}|u|^{-2p+2i}|\nablas^i(L\phi)|^2$$
     $$\|\Omega\nablas\phi\|^2_{\mathcal{D}}\triangleq\sup_{\ub,u}\sup_{0\le i\le N}a^{-2} \delta^{-2-2\gamma}\int_{S_{\ub,u}}|u|^{2+2\gamma-p(3+2\gamma)+2i}|\nablas^i(\Omega\nablas\phi)|^2$$
     $$\|\Lb\phi\|^2_{\mathcal{D}}\triangleq\sup_{\ub,u}\sup_{0\le i\le N}a^{-2} \delta^{-2}\int_{S_{\ub,u}}|u|^{2-2p+2i}|\nablas^i(\widetilde{\Lb\phi})|^2$$
     $$\Es^2\triangleq\sum_{\Phi\in\{L\phi,\Omega\nablas\phi,\Lb\phi\}}\|\Phi\|^2_{\mathcal{D}}$$
     \item [2)]top order derivatives
     \begin{align*}
     \| L\phi\|^2_{\mathcal{F}}\triangleq\sup_{1\le i\le N} &\left[\sup_{u}a^{-2}\delta^{-1-2\gamma}\int_{C_u}|u|^{2+2\gamma-p(3+2\gamma)+2i}|\nablas^i(\Omega\nablas L\phi)|^2\right.\\
     &\left.+a^{-2}\delta^{-1-2\gamma}\int_{\mathcal{M}}|u'|^{1+2\gamma-p(3+2\gamma)+2i}|\nablas^i(\Omega\nablas L\phi)|^2\right]
     \end{align*}
     \begin{align*}
     \|\Omega\nablas\phi\|^2_{\mathcal{F}}\triangleq\sup_{1\le i\le N} &\left[\sup_{\ub}a^{-2}\delta^{-1-2\gamma}\int_{\Cb_{\ub}}|u'|^{2+2\gamma-p(3+2\gamma)+2i}|\nablas^i(\Omega^2\nablas\nablas\phi)|^2\right.\\
     &\left.+\sup_{u}a^{-2}\delta^{-2-2\gamma}\int_{C_u}|u|^{3+2\gamma-p(3+2\gamma)+2i}|\nablas^i(\Omega^2\nablas\nablas\phi)|^2\right]
     \end{align*}
    $$\|\Lb\phi\|^2_{\mathcal{F}}\triangleq\sup_{\ub}\sup_{1\le i\le N} a^{-2}\delta^{-2-2\gamma}\int_{\Cb_{\ub}}|u'|^{3+2\gamma-p(3+2\gamma)+2i}|\nablas^i(\Omega\nablas\Lb\phi)|^2$$
    $$\mathcal{E}^2\triangleq\sum_{\Phi\in\{L\phi,\Omega\nablas\phi,\Lb\phi\}}\|\Phi\|^2_{\mathcal{F}}$$
     \end{itemize}
\begin{remark}We have introduced the spacetime integral in the definition of $\| L\phi\|^2_{\mathcal{F}}$. \end{remark}

    \end{itemize}

\begin{theorem}\label{main3}
Consider the initial value problem as described above. Then if $\varepsilon$ is sufficiently small depending on $u_0$ and $\delta|u_1|^{p-1}a\le \varepsilon$, then the regular solution of \eqref{ES}  exists in the region $\ub\in[0,\delta], u\in[u_0,u_1]$ and the estimates hold
$$\mathcal{R}, \Os, \mathcal{O}, \mathcal{E}, \Es\lesssim 1.$$
\end{theorem}

Here we  use the notation $A\lesssim B$ to denote $A\le CB$ for some constant $C$ {\bf  independent of $a$} and $\varepsilon$ depends on this $C$. The constant $C$ also  depends on $u_0$,  also increases with $\frac{1}{p}$ and $\frac{1}{k^2-\epsilon-p}$, and tends to $\infty$ as either of these two quantities tends to infinity. These two constants only appear in estimating top order derivatives of $\tr\chib$ and top order derivatives of $L\phi$ respectively.  The dependence of $\frac{1}{p}$ can be removed because we essentially consider the case when $p$ is approaching $k^2$ in this paper. We may remove the dependence on $\frac{1}{k^2-\epsilon-p}$ if we assume $\nablas L\phi$ on $C_{u_0}$ vanishes in a sufficiently fast way as approaching $\ub=0$, which is needed in constructing family of gravitational perturbations that are uniformly bounded at the threshold. We will comment on this when we reach the corresponding estimates (Remarks \eqref{dependenceonk2-epsilon-p} and \eqref{dependenceon1/p}).

First of all, it is direct to check as in \cite{Chr08} that if the initial data satisfies the estimates \eqref{initialbound}, then the above norms restricted on $\ub=0$ or $u=u_0$ is bounded.

We first briefly explain the choice of the weights $\delta$ and $|u|$ in different quantities. As discussed in Introduction, we need to utilize the upper bound of $\Omega_0$: \begin{equation}\label{theupperboundofOmega}\Omega_0^2(u)\le |u|^{p}.\end{equation}

   We first analyze the structure of the equations. To each connection coefficient, curvature component and derivatives of scalar field, which are written in the null pair $e_3,e_4$, we can assign a signature as below, first introduced in \cite{CK93}.

   \begin{table}[H]
   \renewcommand{\arraystretch}{2.0}
   \setlength{\tabcolsep}{1.6\tabcolsep}
   \centering
   \caption{Signature table}
   \begin{tabular}{c*{13}{c}}
   \hline
   & 
   $\beta$ & 
   $\Ks$ & 
   $\sigma$ & 
   $\betab$ & 
   $\chi$ & 
   $\chib$ & 
   $\Omega^{-1}\omega$ & 
   $\Omega^{-1}\omegab$ & 
   $\eta$ & 
   $\etab$ & 
   $\Omega^{-1}L\phi$ & 
   $\nablas\phi$ & 
   $\Omega^{-1}\Lb\phi$ \\
   \hline
   $s$ & 
   $-1$ & 
   $0$ & 
   $0$ &
   $1$ & 
   $-1$ & 
   $1$ & 
   $-1$ & 
   $1$ & 
   $0$ & 
   $0$ &
   $-1$ & 
   $0$ & 
   $1$ \\
   \hline
   \end{tabular}
   \end{table}
Here the signature $s$ is defined to be 
$$s(\psi) = \# (e_3)- \# (e_4),\   s(\Psi) = \# (e_3)- \# (e_4)$$
where $\# (e_3), \# (e_4)$ denotes the number of $e_3$'s and $e_4$’s which show up in the definition of $\psi$ and $\Psi$, where $\psi$ is the connection coefficients ($\chi, \chib,\eta,\etab$ and $\Omega^{-1}\omega,\Omega^{-2}\omegab)$ and $\Psi$ is the curvature components ($\beta,\betab, \Ks, \sigma$). Note here that we assign signatures to $\Omega^{-1}\omega$ and $\Omega^{-1}\omegab$ instead of $\omega$ and $\omegab$, because we have an addition factor $\Omega$ in the definition of $\omega,\omegab$, following \cite{Chr08}.

 The most important fact is that,  among all null structure equations and Bianchi equations, the signatures of both sides of the equations are the same.  But we note that in the null structure equations and null Bianchi equations listed in the above section, which are not written in normalized null frame $(e_3, e_4)$ but in $(\Lb, L)$,  there may be some additional $\Omega$ factors appearing on the right hand sides.  If we rewrite them in terms of 
   $$\Omega\chih, \Omega\tr\chi, \Omega\chibh, \Omega\tr\chib, \Omega\eta, \Omega\etab, \omega, \omegab,$$
   and 
   $$  \Omega^2\beta, \Omega^2\betab, \Omega^2\widetilde{\Ks}, \Omega^2\check{\sigma},$$
   and we assign a factor $\Omega$ to each $\nablas$, then the equations will contain no additional $\Omega$ factors on both sides. For example, we rewrite 
   $$D\eta = (\Omega\chi)\cdot\etab-(\Omega\beta+L\phi\ds\phi)$$
   as
   \begin{equation}\label{exampleofOmegaequation}D(\Omega\eta) = \omega\cdot\Omega\eta+(\Omega\chi)\cdot\Omega\etab-(\Omega^2\beta+L\phi\cdot\Omega\ds\phi).\end{equation}
   We will rewrite the other equations in the course of the proof.  Because of this structure, we choose to write down the estimates in terms of $\Omega\psi$ and $\Omega^2\Psi$ and the estimates are easier to follow. 
   
The function $u$ is chosen such that $u=-r$ on $\Cb_0$, so $u\partial_u$ restricted on $\Cb_0$ is the conformal Killing field of the self-similarity. When we want to construct naked singularity solutions,  the connection coefficients, curvature components and derivatives of scalar field written in $\Lb$, which equals $\partial_u$ on $\Cb_0$, and its normalized conjugate $L'$, should satisfy self-similar bounds in the construction of naked singularity solutions. More precisely, it means that $|u|^2\Omega^{s(\Psi)}\Psi$ and $|u|\Omega^{s(\psi)}\psi$  are bounded as in \cite{R-Sh}. For example, we should have
$$\Omega^{-1}\chih\lesssim |u|^{-1}.$$
 In the current situation, we do not expect such bounds because the initial data is large at the threshold, and on the other hand $\Omega_0^2(u)$ obeys a weaker bound $|u|^{p}$ where $p\le k^2$.  In this case we only expect to prove the bound
 $$\Omega\chih(=\Omega^2\cdot\Omega^{-2}\chih)\lesssim |u|^{-1+p}.$$
 To each component, we will assign additional power of $\frac{\delta}{|u|^{1-p}}$, which will be chosen to be small, in order to close the estimates. For $\Omega\chih$, we assume that it obeys the bound $\delta^{\gamma}a$, then the bound we expect to obtain for $\Omega\chih$ is
  $$\Omega\chih\sim \left(\frac{\delta}{|u|^{1-p}}\right)^\gamma \frac{a}{|u|^{1-p}}.$$
 When $\gamma>\frac{p}{1-p}$, we will have a spacetime integral term with favorable sign in  the energy estimates for the curvature components, which is not needed in this paper

To summarize, the expected bounds can be listed below:

    \textit{Connection coefficients:}

    \begin{align*}
    \Omega\chih\sim\left(\frac{\delta}{|u|^{1-p}}\right)^{\gamma}\frac{a}{|u|^{1-p}},&\quad\omega \sim \frac{a}{|u|^{1-p}},\\
   \Omega\chibh \sim\left(\frac{\delta}{|u|^{1-p}}\right)^{1+\gamma}\frac{a}{|u|},& \quad  \widetilde{\Omega\tr\chib}, \omegab\sim\left(\frac{\delta}{|u|^{1-p}}\right)\frac{a}{|u|},\\
   \widetilde{\Omega\tr\chi}\sim\left(\frac{\delta}{|u|^{1-p}}\right)\frac{a^2}{|u|^{1-p}}, &
  \quad \Omega\eta, \Omega\etab\sim\left(\frac{\delta}{|u|^{1-p}}\right)^{1+\gamma}\frac{a}{|u|^{1-\frac{1}{2}p}}    \end{align*}

Note that if a quantity survives in spherical symmetry, then the quantity itself (not its derivatives) satisfies a weaker estimate with $\gamma=0$.

    \textit{Curvature components:}
    \begin{align*}\Omega^2\beta\sim\left(\frac{\delta}{|u|^{1-p}}\right)^{\gamma}\frac{a}{|u|^{2-\frac{3}{2}p}},&\quad\Omega^2\betab\sim\left(\frac{\delta}{|u|^{1-p}}\right)^{\frac{3}{2}+\gamma}\frac{a^{\frac{3}{2}}}{|u|^{2-\frac{1}{2}p}},\\
    \quad \Omega^2\widetilde{\Ks}\sim\left(\frac{\delta}{|u|^{1-p}}\right)\frac{a}{|u|^{2-p}}, &\quad \Omega^2\sigmac\sim\left(\frac{\delta}{|u|^{1-p}}\right)^{1+\gamma}\frac{a}{|u|^{2-p}}\end{align*}

        \textit{(Top order derivatives of) the scalar field:}
    $$\Omega\nablas L\phi\sim\left(\frac{\delta}{|u|^{1-p}}\right)^{\gamma}\frac{a}{|u|^{2-\frac{3}{2}p}},\quad\Omega^2\nablas\nablas\phi\sim\left(\frac{\delta}{|u|^{1-p}}\right)^{\frac{1}{2}+\gamma}\frac{a}{|u|^{2-p}},
    \quad\Omega\nablas\Lb\phi\sim\left(\frac{\delta}{|u|^{1-p}}\right)^{1+\gamma}\frac{a}{|u|^{2-\frac{1}{2}p}}$$

    \textit{Top order derivatives of the connection coefficients:}
    \begin{align*}
    &\Omega^2\chih,\Omega\omega\sim\left(\frac{\delta}{|u|^{1-p}}\right)^{\gamma}\frac{a}{|u|^{1-\frac{3}{2}p}},\quad\Omega^2\chibh\sim\left(\frac{\delta}{|u|^{1-p}}\right)^{1+\gamma}\frac{a}{|u|^{1-\frac{1}{2}p}},
    \quad\Omega^2\tr\chi\sim\left(\frac{\delta}{|u|^{1-p}}\right)^{1+\gamma}\frac{a^2}{|u|^{1-\frac{3}{2}p}},\\
    &\Omega^2\tr\chib\sim\left(\frac{\delta}{|u|^{1-p}}\right)^{1+\gamma}\frac{a}{|u|^{1-\frac{1}{2}p}}, \quad\Omega^2\eta,\Omega^2\etab\sim\left(\frac{\delta}{|u|^{1-p}}\right)^{\frac{1}{2}+\gamma}\frac{a}{|u|^{1-p}},\quad\Omega\omegab\sim\left(\frac{\delta}{|u|^{1-p}}\right)^{1+\gamma}\frac{a}{|u|^{1-\frac{1}{2}p}},
    \end{align*}

Finally, the norms will also have their scales: 
	$$\int_{S_{\ub,u}}\sim |u|^2, \quad\int_{C_u}\sim\delta|u|^2, \quad\int_{\Cb_u}\sim|u|^3, \quad\int_\mathcal{M}\sim \delta|u|^3,\quad \nablas^i\sim |u|^{-i}$$

One important aspect of this paper is that these expected bounds are compatible with the equations written in the form as \eqref{exampleofOmegaequation}, that is, in terms of $\Omega\psi$ and $\Omega^2\Psi$. Regardless of the dimensionless factor $\delta|u|^{p-1}$, the powers of $|u|$ of both sides of the each equation are equal.  The underlying reason is precisely that the two sides of the equations have the same signature. The (lower order) bounds of the connection coefficients and curvature components are typically 
$$\Omega\psi\lesssim\frac{1}{|u|^{1-\frac{1}{2}p+\frac{1}{2}ps(\psi)}}, \quad\Omega^2\Psi\lesssim\frac{1}{|u|^{2-p+\frac{1}{2}ps(\Psi)}}$$
and the derivatives of scalar field verify similar bounds.

\begin{remark}
The top order estimates will  suffer a loss of $\frac{\Omega}{|u|^{\frac{1}{2}p}}$ because of the elliptic estimates. 
\end{remark}

   \begin{remark}
    Compared to the previous work \cite{Li-Liu1}, we  use a more precise bound of $\Omega$, which is (\ref{theupperboundofOmega}), while in \cite{Li-Liu1} we  only use $\Omega\lesssim1$, i.e., $p=0$, and the estimates in \cite{Li-Liu1} are done for $\gamma=0$. To achieve the optimal instability in the case of non-spherically symmetry, an accurate estimate of $\Omega$ is necessary.
      \end{remark}

The proof is began by assuming 
\begin{equation}\label{bootstrap}\mathcal{R}, \Os, \mathcal{O}, \mathcal{E}, \Es\le \varepsilon^{-\kappa},\end{equation}
where $\kappa$ is a small constant (say $\kappa=0.1$), we will show that if $\varepsilon>0$ is sufficiently small depending on  $u_0$  and $p$,  then
$$\mathcal{R}, \Os, \mathcal{O}, \mathcal{E}, \Es\lesssim 1$$
for $\delta|u_1|^{p-1}a\le\varepsilon$.

\subsection{Geometric lemmas}
We start by collecting some geometric lemmas which will be frequently used in the course of the proof. Their proofs are standard (see \cite{Chr08} for example). Under the bootstrap assumptions \eqref{bootstrap}, if $\varepsilon$ is sufficiently small, we have
\begin{itemize}
	\item {\bf Estimate of lapse and area:}
	\begin{equation}\label{lapse}\frac{1}{2}\Omega_0\le\Omega\le 2\Omega_0, \end{equation}
	\begin{equation}\label{area}\frac{1}{2}|u|^2\le\frac{1}{4\pi}\mathrm{Area}(S_{\ub,u})\le 2|u|^2.\end{equation}
	\item {\bf Sobolev inequalities: }
	Given a tangential tensorfield $\theta$, we have,
	\begin{align}\label{Sobolev}
		\|\theta\|_{L^\infty(S_{\ub,u})}&\lesssim\sum_{i=0}^2|u|^{-1}\|(|u|\nablas)^i\theta\|_{L^2(S_{\ub,u})},
	\end{align}
	\item {\bf H\"older inequality:} Define the scale-invariant form of the Sobolev norm, 
	$$\|\theta\|_{H^k(S_{\ub,u})}:=\sum_{i=0}^k \|(|u|\nablas)^i\theta\|_{L^2(S_{\ub,u})}, $$
	we have, for tangential tensorfields $\theta_1,\cdots,\theta_n$, and $k\ge2$, 
	\begin{equation}\label{Holder}
		|u|^{-1}\|\theta_1\cdots\theta_n\|_{H^k(S_{\ub,u})}\lesssim|u|^{-n}\|\theta_1\|_{H^k(S_{\ub,u})}
		\cdots\|\theta_n\|_{H^k(S_{\ub,u})}
	\end{equation}
	\item {\bf Derivatives of lapse:} For any $s\in\mathbb{R}$, 
	\begin{equation*}|u|^{-1}\|\Omega^s\|_{H^{N}(S_{\ub,u})}\le C_s\Omega^s.\end{equation*}
	This estimate tells us that for any tangential tensorfield $\theta$, $2\le k\le N$,
	\begin{equation}\label{derivativelapse}\|\Omega^s\theta\|_{H^k(S_{\ub,u})}\le C_{s,k}\Omega^s\|\theta\|_{H^k(S_{\ub,u})}.\end{equation}
	So $\Omega$ can be moved freely between outside and inside the $H^k$ norm. Moreover, for top order derivatives, writing $\ds\Omega=\frac{1}{2}\Omega(\eta+\etab)$, we will have
	$$|u|^{-1}\|\Omega^s\|_{H^{N+1}(S_{\ub,u})}\le C_s \Omega^{s-1}|u|^{\frac{1}{2}p},$$
	and for any tangential tensorfield $\theta$, 
	\begin{equation}\label{derivativelapsetop}\|\Omega^s\theta\|_{H^{N+1}(S_{\ub,u})}\le C_{s}(\Omega^s\|\theta\|_{H^{N+1}(S_{\ub,u})}+\Omega^{s-1}\|(\Omega\eta,\Omega\etab)\|_{H^N(S_{\ub,u})}\|\theta\|_{H^{N}(S_{\ub,u})}).\end{equation}

	\item {\bf Gronwall type inequalities: } Given a $s$-tangential tensorfield $\theta$, we have 
	\begin{align}
		\label{Gronwallub}|u|^{-1}\|\theta\|_{L^2(S_{\ub,u})}
		\le &C_s |u|^{-1}\left(\|\theta\|_{L^2(S_{0,u})}+\int_0^{\ub}\|D\theta\|_{L^2(S_{\ub',u})}\D\ub'\right).\\
		\label{Gronwallu}\begin{split}\||u|^{s+\nu-1}\theta\|_{L^2(S_{\ub,u})}\lesssim& C_{s,\nu}\left(\||u|^{s+\nu-1}\theta\|_{L^2(S_{\ub,u_0})}\right.\\&\left.+\int_{u_0}^u\||u'|^{s+\nu-1}(\Db\theta+\frac{\nu}{2}\Omega\tr\chib\theta)\|_{L^2(S_{\ub,u'})}\D u'\right).\end{split}
	\end{align}
	The above estimates also hold for angular derivatives up to order $N$ (replacing $L^2$ by $H^N$), or $N+1$ when $\theta$ is a function. If $\theta$ is a trace-free $2$-tangential tensorfield, then $D\theta$ and $\Db\theta$ can be replaced by their trace-free parts $\Dh\theta$ and $\Dbh\theta$.
	\item {\bf Elliptic estimates:} Provided that
	\begin{align*} |u|^{-1}\|\Ks\|_{H^{2}(S_{\ub,u})}\lesssim|u|^{-2},\end{align*}
	we have, for a tangential tensorfield $\theta$ verifying a Hodge system
	\begin{align*}
		\divs\theta=f,
		\curls\theta=g,
	\end{align*} 
	we have
	\begin{align}\label{elliptic}
		\|\theta\|_{H^{N+1}(S_{\ub,u})}\lesssim \||u|(f,g)\|_{H^N(S_{\ub,u})}+(1+ |u|\|\Ks\|_{H^{N-1}(S_{\ub,u})}) \|\theta\|_{H^N(S_{\ub,u})}.
	\end{align}
\end{itemize}

\section{The estimates for lower-order derivatives of connection coefficients and scalar field}
 In this section, we will prove the following.
\begin{proposition}\label{connection-low} In addition to the bootstrap assumptions \eqref{bootstrap}, assume moreover 
	\begin{equation}\label{connection-low-curvaturebound}
		\mathcal{R}\lesssim 1.
	\end{equation}
	Then if $\varepsilon$ is sufficiently small, and $\delta|u|^{p-1}\le\varepsilon$, we have
	\begin{equation}\label{connection-low-bound}\Os, \Es\lesssim 1. \end{equation}
\end{proposition}

\begin{proof}
     For the estimates of the connection coefficients of lower-order derivatives, we need to select a suitable null structure equations. For the equation in the $\Db$ direction, by choosing $s$ from the type of this tensor and selecting the appropriate $\nu$ based on the linear term of the equation, we apply \eqref{Gronwallu} to this equation. For the equation in the $D$ direction, we directly apply \eqref{Gronwallub} to this equation. During this process, the H\"older inequality \eqref{Holder} and the bootstrap assumption will be utilized. To make this statement clearer, we will give an example. Since $\chih$ will appear as a borderline term in some equations, we choose to estimate this first. Considering that the curvature component $\alpha$ does not appear in our paper, the equation of $\chih$ can only be $\Db\chih$ instead of $D\chih$.

	 For $\chih$, we write the equation for $\Db\chih$ as
    $$\Dbh(\Omega\chih)-\frac{1}{2}\Omega\tr\chib\Omega\chih=\Omega^2(\nablas \tensor \eta + \eta \tensor \eta -\frac{1}{2}\tr\chi \chibh+\ds\phi\tensor\ds\phi).$$
	By applying  H\"older inequality \eqref{Holder} and Gronwall \eqref{Gronwallu} for $s=2, \nu=-1$, using \eqref{bootstrap}, we have
    \begin{align*}&|u|^{-1}\|\Omega\eta\widehat{\otimes}\Omega\eta-\frac{1}{2}\Omega\tr\chi\Omega\chibh+\Omega\nablas\phi\widehat{\otimes}\Omega\nablas\phi\|_{H^N(S_{\ub,u})}\\
    \lesssim&(\delta|u|^{p-1})^{2+2\gamma}|u|^{p-2}a^2\varepsilon^{-2\kappa}+|u|^{p-1} a\cdot(\delta|u|^{p-1})^{1+\gamma}a|u|^{-1}\varepsilon^{-2\kappa}+(\delta|u|^{p-1})^{2+2\gamma}|u|^{p-2}a^2\varepsilon^{-2\kappa}\\
    \lesssim&(\delta|u|^{p-1})^{1+\gamma}|u|^{p-2}a\varepsilon^{-2\kappa}.
    \end{align*}
    where we have used
    $$|u|^{-1}\|\Omega\tr\chi\|_{H^N(S_{\ub,u})}\lesssim a|u|^{p-1}$$
    by bootstrap assumption $\|\tr\chi\|_{\mathcal{B}}\le \varepsilon^{-\kappa}$ and the initial data on $\Cb_0$.
   \begin{remark}
   It is instructive to see how the signature argument comes into play. Although not necessary for the proof, it is helpful for understanding the underlying structure of the argument. The second line of the above estimate can be written as
        \begin{align*}&(\delta|u|^{p-1})^{2+2\gamma}(|u|^{\frac{1}{2}p-\frac{1}{2}ps(\eta)-1})^2a^2\varepsilon^{-2\kappa}+|u|^{\frac{1}{2}p-\frac{1}{2}ps(\tr\chi)-1} a\cdot(\delta|u|^{p-1})^{1+\gamma}a|u|^{\frac{1}{2}p-\frac{1}{2}s(\chibh)-1}\varepsilon^{-2\kappa}\\
    &+(\delta|u|^{p-1})^{2+2\gamma}(|u|^{\frac{1}{2}p-\frac{1}{2}ps(\nablas\phi)-1})^2a^2\varepsilon^{-2\kappa}.
    \end{align*}
    The signatures of both sides of the equations equal, so we have
   $$s(\chih)+1= 2s(\eta)=s(\tr\chi)+s(\chibh)=2s(\nablas\phi).$$
   Then the above estimate is bounded above by
      \begin{align*}(\delta|u|^{p-1})^{1+\gamma}\cdot|u|^{p-\frac{1}{2}p(s(\chih)+1)-2}a\varepsilon^{-2\kappa}=(\delta|u|^{p-1})^{1+\gamma}|u|^{p-2}a\varepsilon^{-2\kappa}.
    \end{align*}
   \end{remark} 
        The initial data satisfies, for any $u\in[u_0,u_1]$, 
  $$  \|\Omega\chih\|_{H^N(S_{\ub,u_0})}\lesssim\delta^{\gamma}a\le (\delta|u|^{p-1})^{\gamma}|u|^{\gamma(1-p)}a\lesssim (\delta|u|^{p-1})^{\gamma}|u|^pa.$$
  Then applying Gronwall \eqref{Gronwallu} for $s=2,\nu=-1$,
    \begin{align}\nonumber\|\Omega\chih\|_{H^N(S_{\ub,u})}\lesssim&\|\Omega\chih\|_{H^N(S_{\ub,u_0})}+\int_{u_0}^{u} \|\Dbh(\Omega\chih)-\frac{1}{2}\Omega\tr\chib\Omega\chih\|_{H^N(S_{\ub,u'})}\D u'\\
\nonumber\lesssim&\|\Omega\chih\|_{H^N(S_{\ub,u_0})}+\int_{u_0}^{u}\|\Omega^2\nablas\widehat{\otimes}\eta\|_{H^N(S_{\ub,u'})}\D u'+\int_{u_0}^{u}(\delta|u'|^{p-1})^{1+\gamma}a|u'|^{p-1}\varepsilon^{-2\kappa}\D u'\\
        \label{chihbound}\lesssim&\|\Omega\chih\|_{H^N(S_{\ub,u_0})}+(\delta|u|^{p-1})^{\frac{1}{2}+\gamma}|u|^pa\|\eta\|_{\mathcal{C}}+(\delta|u|^{p-1})^{\gamma}|u|^p\le (\delta|u|^{p-1})^{\gamma}|u|^p a
      \end{align}
	when $\delta|u_1|^{p-1}a\le\varepsilon$ is sufficiently small (and recall that $a\ge1$). This gives the desired estimate for $\chih$. 

    For $L\phi$, from the perspective of the equation in the $\Db$ direction, it is almost identical to $\Omega\chih$. However, the difference in the estimates lies in the fact that $\Omega\tr\chi\Lb\phi$ is worse than $\Omega\tr\chi\Omega\chibh$. Therefore, we choose to estimate the following equation
    $$\Db \widetilde{L\phi}+\frac{1}{2}\Omega\tr\chib \widetilde{L\phi}=\Omega^2\Deltas\phi+2\Omega^2(\eta,\ds\phi)-\frac{1}{2}L\phi|_{\Cb_0}\widetilde{\Omega\tr\chib}-\frac{1}{2}\Omega\tr\chi\widetilde{\Lb\phi}-\frac{1}{2}\Lb\phi|_{\Cb_0}\widetilde{\Omega\tr\chi}.$$
    Use the fact $\|L\phi\|_{H^N(S_{0,u})}\lesssim|u|^pK$ (and recall that $a\ge K$) and the initial data $\|\widetilde{L\phi}\|_{H^N(S_{\ub,u_0})}\lesssim|u|^pa$, we have 
    \begin{equation}\label{lowerLphibound}
    \|L\phi\|_{H^N(S_{\ub,u})}\le\|\widetilde{L\phi}\|_{H^N(S_{\ub,u})}+\|L\phi\|_{H^N(S_{0,u})}\lesssim|u|^pa.
    \end{equation}

	For $\tr\chi$, $\eta$ and $\Omega\nablas\phi$, we rewrite the equation for $D\tr\chi'$, $ D\eta$ and $D\nablas\phi$ as
    \begin{align*}
	&D(\widetilde{\Omega\tr\chi})=-\frac{1}{2}(\Omega\tr\chi)^2+2\omega\Omega\tr\chi-\boxed{|\Omega\chih|^2}-\boxed{2(L\phi)^2},\\
	&D(\Omega\eta) = \omega\cdot\Omega\eta+(\Omega\chi)\cdot(\Omega\etab)-(\Omega^2\beta+L\phi\Omega\nablas\phi),\\
    &D(\Omega\nablas\phi)=\omega\Omega\nablas\phi+\Omega\nablas L\phi.
    \end{align*}
	  By applying \eqref{Gronwallub} to the above equation, we can obtain the desired bound. One should note that in estimating $|\Omega\chih|^2$ and $(L\phi)^2$, so-called \emph{borderline term},  on the right hand side of $D(\widetilde{\Omega\tr\chi})$, we should use the bound \eqref{chihbound}, \eqref{lowerLphibound} instead of \eqref{bootstrap}, since we will not have extra $\varepsilon$ after integrating that term. Note also that this is the reason why we have $a^2$ instead of $a$ in the definitions of $\|\tr\chi\|_{\mathcal{B}}$ and $\|\tr\chi\|_{\mathcal{C}}$.

	For $\chibh$, $\tr\chib$ and $\Lb\phi$, we use the following equations of $D$ direction, in which the right hand sides contain top order derivatives and we should use $\mathcal{O}$ or $\mathcal{E}$. The estimates for $\widetilde{\Omega\tr\chib}$ is obtained by integrating the equation 
    $$D(\widetilde{\Omega\tr\chib})=2\Omega^2\divs\etab+2|\Omega\etab|^2-\Omega\tr\chi\Omega\tr\chib-2\Omega^2\Ks+2|\Omega\nablas\phi|^2.$$
   And we also obtain 
   $$\|\Omega\tr\chib\|_{H^N(S_{\ub,u})}\lesssim1.$$
  $\Omega\chibh$ and $\Lb\phi$ are estimated using
    \begin{align*}
    &\Dh(\Omega\chibh)=\Omega^2\nablas \tensor \etab + (\Omega\etab) \tensor (\Omega\etab) +\frac{1}{2}\Omega\tr\chi\Omega\chibh- \frac{1}{2}\Omega\tr\chib \Omega\chih+(\Omega\nablas\phi)\widehat{\otimes}(\Omega\nablas\phi),\\	
     &D\widetilde{\Lb\phi}+\frac{1}{2}\Omega\tr\chi\Lb\phi=\Omega^2\Deltas\phi+2\Omega^2(\etab,\ds\phi)- \frac{1}{2}\Omega\tr\chib L\phi.
    \end{align*}
	 Moreover, we have 
	 $$\|\Lb\phi\|_{H^N(S_{\ub,u})}\lesssim1.$$

     For $\omegab$, we rewrite the equation for $D\omegab$ as
     $$D  \widetilde{\omegab} =2(\Omega\eta,\Omega\etab)-|\Omega\eta|^2+\Omega^2\Ks+ \frac{1}{4}\Omega\tr\chi\Omega\tr\chib-\frac{1}{2}(\Omega\chih,\Omega\chibh)-|\Omega\nablas\phi|^2-L\phi\Lb\phi.$$

	For $\etab$, we rewrite the equation for $\Db\etab$ in the form
	$$\Db(\Omega\etab) = \omegab\cdot\Omega\etab+ \Omega\chibh\cdot\Omega\eta+\frac{1}{2}\Omega\tr\chib\Omega\eta+(\Omega^2\betab+\Lb\phi\Omega\nablas\phi)-2\Lb\phi\Omega\nablas\phi.$$
	Based on the bootstrap assumption \eqref{bootstrap} (or the bound for $\widetilde{\omegab}$ obtained above), the Sobolev inequality \eqref{Sobolev}, and the initial positive lower bound \eqref{omegablower} of $-\omegab$, we can conclude that $\omegab<0$ throughout the spacetime. Therefore, the first term on the right hand side can be dropped.
	By Gronwall \eqref{Gronwallu}, for $\etab$, $s=1,\nu=0$, so
	\begin{align*}\|\Omega\etab\|_{H^N(S_{\ub,u})}\lesssim& \|\Omega\etab\|_{H^N(S_{\ub,u_0})}+\delta^{1+\gamma}|u|^{p(\frac{3}{2}+\gamma)-1-\gamma}a\lesssim\delta^{1+\gamma}|u|^{p(\frac{3}{2}+\gamma)-1-\gamma}a
	\end{align*}
	when $\varepsilon$ is small enough. This gives the desired bound for $\etab$. For $\omega$, we use the equation
    $$\Db \omega =2(\Omega\eta,\Omega\etab)-|\Omega\etab|^2+\Omega^2\Ks+\frac{1}{4}\Omega\tr\chi\Omega\tr\chib-\frac{1}{2}(\Omega\chih,\Omega\chibh)-|\Omega\nablas\phi|^2-L\phi\Lb\phi,$$
    and apply (3.9) to $\omega$ with $s=\nu=0$.

\end{proof}

  \section{The estimates for curvature components}
	We are going to prove the estimates for curvature components.

	\begin{proposition}\label{curvature}
		Under the assumptions of \eqref{bootstrap}, we have
		\begin{equation}\label{curvaturebound}\mathcal{R}\lesssim1,\end{equation}
	\end{proposition}
		Together with Proposition \ref{connection-low}, the estimate \eqref{connection-low-bound} holds without the additional assumption \eqref{connection-low-curvaturebound}.
	\begin{proof}

		First of all, we derive lower-order estimates for $\Ks$. This can be done from the equation
		$$D\widetilde{\Ks}+\Omega\tr\chi\Ks=\divs\divs(\Omega\chih)-\frac{1}{2}\Deltas(\Omega\tr\chi).$$
		Rewrite this equation as $D(\Omega \widetilde{\Ks})=\omega\Omega\widetilde{\Ks}+\Omega D\widetilde{\Ks}$ and apply \eqref{Gronwallub}, using the bootstrap assumption \eqref{bootstrap}  for (top order derivatives of) $\chih$ and $\tr\chi$, we have
		
		\begin{align*}
		\|\Omega\divs\divs(\Omega\chih)\|_{H^{N-1}(S_{\ub,u})}\lesssim&|u|^{-2}\|\Omega^2\chih\|_{H^{N+1}(S_{\ub,u})}+|u|^{-2}\|\Omega(\eta+\etab)\|_{H^{N}(S_{\ub,u})}\|\Omega\chih\|_{H^{N-1}(S_{\ub,u})}\\
		+&|u|^{-2}\|\Omega(\eta+\etab)\|_{H^{N-1}(S_{\ub,u})}\|\Omega\chih\|_{H^{N}(S_{\ub,u})}\\
		\lesssim&|u|^{-2}\|\Omega^2\chih\|_{H^{N+1}(S_{\ub,u})}+\delta|u|^{\frac{5}{2}p-3}a\varepsilon^{-2\kappa}
	   \end{align*}
		So we have 
		\begin{align*}
		|u|^{-1}\int_{0}^{\delta}\|\Omega\divs\divs(\Omega\chih)\|_{H^{N-1}(S_{\ub,u})}\lesssim&(\delta|u|^{p-1})^{1+\gamma}|u|^{\frac{1}{2}p-2}a\varepsilon^{-\kappa}+(\delta|u|^{p-1})^2|u|^{\frac{1}{2}p-2}a\varepsilon^{-2\kappa}\\
		\lesssim&|u|^{\frac{1}{2}p-2}
	   \end{align*}
		for $\delta|u|^{p-1}a\leq\varepsilon$ small enough.
		Similarly, we can obtain
		\begin{align*}
		|u|^{-1}\int_{0}^{\delta}\|\Omega\Deltas(\Omega\tr\chi)\|_{H^{N-1}(S_{\ub,u})}\lesssim&(\delta|u|^{p-1})^2|u|^{\frac{1}{2}p-2}a^2\varepsilon^{-2\kappa}+(\delta|u|^{p-1})^{2+\gamma}|u|^{\frac{1}{2}p-2}a^2\varepsilon^{-\kappa}\\
		\lesssim&|u|^{\frac{1}{2}p-2}
		\end{align*}
		Therefore, we have
		\begin{equation}\label{Gausscurvaturebound}|u|^{-1}\|\Omega\widetilde{\Ks}\|_{H^{N-1}(S_{\ub,u})}\lesssim|u|^{\frac{1}{2}p-2},\end{equation}
		for $\varepsilon$ small enough. Note if we only consider up to $N-2$ order derivatives of $\Ks$, then we can use $\Os$ alone, then 
		\begin{align*}
		|u|^{-1}\int_{0}^{\delta}\|\divs\divs(\Omega\chih)\|_{H^{N-2}(S_{\ub,u})}+|u|^{-1}\int_{0}^{\delta}\|\Deltas(\Omega\tr\chi)\|_{H^{N-2}(S_{\ub,u})}\lesssim|u|^{-2}
	    \end{align*}
		that is,
		$$|u|^{-1}\|\widetilde{\Ks}\|_{H^{N-2}(S_{\ub,u})}\lesssim|u|^{-2}$$
		and in particular  $|\Ks|\lesssim|u|^{-2}$. So we have
        \begin{equation}\label{rhodaggeritself}
        |\Omega^2\Ks|\lesssim|u|^{p-2}
        \end{equation}
        which is exactly the zeroth-order estimate for $\Ks$ that is missing from the definition of $\|\Ks\|_{\mathcal{A}}$. 
		
		Now we begin the energy estimates for the curvature components. For $\beta$-$(\Ks,\check{\sigma})$ pair, we use the equations
		{\small
        \begin{align*}
        &\Db(\Omega\beta-L\phi\nablas\phi)+\frac{1}{2}\Omega\tr\chib(\Omega\beta-L\phi\nablas\phi)-\Omega\chibh\cdot(\Omega\beta-L\phi\nablas\phi)+\Omega^2\ds \Ks-\Omega^2{}^*\ds\sigmac+3\Omega^2(\eta \Ks-{}^*\eta\sigmac)\\
     &-2\Omega\chih\cdot(\Omega\betab+\Lb\phi\nablas\phi)-\frac{1}{2}\Omega^2(\ds(\chih,\chibh)+{}^*\ds(\chih\wedge\chibh))-\frac{3}{2}\Omega^2(\eta(\chih,\chibh)+{}^*\eta(\chih\wedge\chibh))+\frac{1}{4}\Omega^2\ds(\tr\chi\tr\chib)\\
     &+\frac{3}{4}\Omega^2\tr\chi\tr\chib\eta=-2\Omega^2\Deltas\phi\nablas\phi+\Omega^2\ds|\ds\phi|^2-2\Omega\chih\cdot\nablas\phi\Lb\phi+\Omega\tr\chi\Lb\phi\nablas\phi-2\Omega^2\eta\cdot\nablas\phi\nablas\phi+2\Omega^2\eta|\ds\phi|^2,\\
        &D\Ks+\Omega\tr\chi \Ks+\divs(\Omega\beta-L\phi\nablas\phi)+(\Omega\beta-L\phi\nablas\phi)\cdot\etab-\Omega\chih\cdot\nablas\etab+\frac{1}{2}\Omega\tr\chi\divs\etab-\Omega\chih\cdot\etab\cdot\etab+\frac{1}{2}\Omega\tr\chi|\etab|^2=0\\
     &D\sigmac+\frac{3}{2}\Omega\tr\chi\sigmac+\curls(\Omega\beta-L\phi\nablas\phi)+\etab\wedge(\Omega\beta-L\phi\nablas\phi)+\frac{1}{2}\Omega\chih\wedge(\etab\tensor\etab+\nablas\tensor\etab)=-2\nablas L\phi\wedge\nablas\phi\\
    \end{align*}}
     Then we will reorganize the above equation in the following form
     {\small
    \begin{align*}   
    &\Db(\Omega^2\beta-L\phi\Omega\nablas\phi)+\frac{1}{2}\Omega\tr\chib(\Omega^2\beta-L\phi\Omega\nablas\phi)-\omegab(\Omega^2\beta-L\phi\Omega\nablas\phi)-\Omega\chibh\cdot(\Omega^2\beta-L\phi\Omega\nablas\phi)+\Omega\ds( \Omega^2\Ks)\\
    &+\Omega(2\eta-\etab)\Omega^2\Ks-\Omega{}^*\ds(\Omega^2\sigmac)+{}^*\Omega(-2\eta+\etab)\Omega^2\sigmac-2\Omega\chih\cdot(\Omega^2\betab+\Lb\phi\Omega\nablas\phi)-\frac{1}{2}\Omega\ds(\Omega\chih,\Omega\chibh)\\
    &+\Omega(-\eta+\frac{1}{2}\etab)(\Omega\chih,\Omega\chibh)-\frac{1}{2}\Omega{}^*\ds(\Omega\chih\wedge\Omega\chibh)+{}^*\Omega(-\eta+\frac{1}{2}\etab)(\Omega\chih\wedge\Omega\chibh)+\frac{1}{4}\Omega\ds(\Omega\tr\chi\Omega\tr\chib)\\
    &+\Omega\tr\chi\Omega\tr\chib\Omega(\frac{1}{2}\eta-\frac{1}{4}\etab)=-2\Omega^2\Deltas\phi\Omega\nablas\phi+\Omega^3\ds|\ds\phi|^2-2\Omega\chih\cdot\Omega\nablas\phi\Lb\phi+\Omega\tr\chi\Lb\phi\Omega\nablas\phi-2\Omega\eta\cdot\Omega\nablas\phi\Omega\nablas\phi\\
    &+2\Omega\eta|\Omega\ds\phi|^2,\\
    &D(\Omega^2(\Ks-|u|^{-2}))+\Omega\tr\chi \Omega^2\Ks-2\omega\Omega^2(\Ks-|u|^{-2})+\Omega\divs(\Omega^2\beta-L\phi\Omega\nablas\phi)+\frac{1}{2}(\Omega^2\beta-L\phi\Omega\nablas\phi)\cdot\Omega\etab\\
    &-\frac{1}{2}(\Omega^2\beta-L\phi\Omega\nablas\phi)\cdot\Omega\eta-\Omega\chih\cdot\Omega\nablas(\Omega\etab)-\frac{1}{2}\Omega\chih\cdot\Omega\etab\cdot\Omega\etab+\frac{1}{2}\Omega\chih\cdot\Omega\eta\cdot\Omega\etab+\frac{1}{2}\Omega\tr\chi\Omega\divs(\Omega\etab)\\
    &+\frac{1}{4}\Omega\tr\chi|\Omega\etab|^2-\frac{1}{4}\Omega\tr\chi\Omega\eta\cdot\Omega\etab=0\\
    &D(\Omega^2\sigmac)+\frac{3}{2}\Omega\tr\chi\Omega^2\sigmac-2\omega\Omega^2\sigmac+\Omega\curls(\Omega^2\beta-L\phi\Omega\nablas\phi)+(-\frac{1}{2}\Omega\eta+\frac{1}{2}\Omega\etab)\wedge(\Omega^2\beta-L\phi\Omega\nablas\phi)\\
    &+\frac{1}{2}\Omega\chih\wedge\left((-\frac{1}{2}\Omega\eta+\frac{1}{2}\Omega\etab)\tensor\Omega\etab+\Omega\nablas\tensor(\Omega\etab)\right)=-2\Omega\nablas L\phi\wedge\Omega\nablas\phi\\
    \end{align*}}
    we compute for $0\le i\le N$,
		\begin{align*}
			&\Db\left(\delta^{-1-2\gamma}|u|^{2+2\gamma-p(3+2\gamma)+2i}|\nablas^i(\Omega^2\beta-L\phi\Omega\nablas\phi)|^2\D\mu_{\gs}\right)\\
            &+D\left(\delta^{-1-2\gamma}|u|^{2+2\gamma-p(3+2\gamma)+2i}\left(|\nablas^i(\Omega^2(\Ks-|u|^{-2}))|^2+|\nablas^i(\Omega^2\sigmac)|^2\right)\D\mu_{\gs}\right)\\
			&=\delta^{-1-2\gamma}|u|^{2+2\gamma-p(3+2\gamma)+2i}\divs(\nablas^i(\Omega^2\beta-L\phi\Omega\nablas\phi)\cdot (\nablas^i(\Omega^2(\Ks-|u|^{-2})),\nablas^i(\Omega^2\sigmac)) )\D\mu_{\gs}\\
			&\underline{-((2\gamma-p(3+2\gamma))|u|^{-1}-2\omegab|_{\Cb_0})\left(\delta^{-1-2\gamma}|u|^{2+2\gamma-p(3+2\gamma)+2i}|\nablas^i(\Omega^2\beta-L\phi\Omega\nablas\phi)|^2\D\mu_{\gs}\right)}\\
			&+\delta^{-1-2\gamma}|u|^{2+2\gamma-p(3+2\gamma)+2i}\tau_{1}\D\mu_{\gs}.
		\end{align*}
    The choice of the power $|u|$ is made such that $\tau_1$ does not contain terms like $$(\omegab,\Omega\tr\chib)|\nablas^i(\Omega^2\beta-L\phi\Omega\nablas\phi)|^2.$$ The coefficient of the underlined term is worthy of attention. Recall from \eqref{omegablower} that $-2\omegab\big|_{\Cb_0}\ge p|u|^{-1}$, we have
    $$-2|u|\omegab\big|_{\Cb_0}\ge p\ge 3p-2\gamma(1-p).$$
    So the coefficient of the underlined term is nonpositive. Integrating over $(\mathcal{M}, \D\ub\D u\D\mu_{\gs})$ and drop the underlined term, we have
    \begin{align*}
			&\delta^{-1-2\gamma}\int_{C_u}|u|^{2+2\gamma-p(3+2\gamma)+2i}|\nablas^i(\Omega^2\beta-L\phi\Omega\nablas\phi
            )|^2\\
            &+\delta^{-1-2\gamma}\int_{\Cb_{\ub}}|u|^{2+2\gamma-p(3+2\gamma)+2i}\left(|\nablas^i(\Omega^2(\Ks-|u|^{-2}))|^2+|\nablas^i(\Omega^2\sigmac)|^2\right)\\
			\lesssim&a^2+\delta^{-1-2\gamma}\int_{\mathcal{M}}|u|^{2+2\gamma-p(3+2\gamma)+2i}|\tau_{1}|.
	\end{align*}
    Due to the large number of terms involved in the equations, we will not write out the specific expressions in detail. Instead, we will analyze and explain the components of $\tau$ below. We then group these terms into the following categories:
    \begin{itemize}
        \item (1) In this case, there are three factors involved, among which there are two curvature components. It can be written in the following form (in a schematic form):
        $$\nablas^i(\Omega\psi\cdot(\Omega^2\Psi_1))\cdot\nablas^i(\Omega^2\Psi_2)$$
     where 
    $$\Omega\psi\in \{\Omega\chih, \Omega\chibh, \Omega\tr\chi, \widetilde{\Omega\tr\chib},\Omega\eta, \Omega\etab, \omega, \widetilde{\omegab}\}$$
    $$\Omega^2\Psi_{1,2}\in\{\Omega^2\beta-L\phi\Omega\nablas\phi, \Omega^2\betab+\Lb\phi\Omega\nablas\phi, \Omega^2\widetilde{\Ks}, \Omega^2\sigmac\},$$
 and both  the signatures  and additional  powers of $\Omega$ match on both sides. For example, the subcase $\Omega\psi\in\{\Omega\chih,\Omega\tr\chi, \omega\}$ and $\Omega^2\Psi_{1,2}=\Omega^2\beta-L\phi\Omega\nablas\phi$ is not included. The same applies to the other cases, and we will not repeat this structure below. In this case, we put $\Omega^2\psi$ in $H^j(S_{\ub,u})$, derivatives of $\Omega^2\beta-L\phi\Omega\nablas\phi$ in $L^2(C_u)$ and derivatives of $\Omega^2\betab+\Lb\phi\Omega\nablas\phi, \Omega^2\widetilde{\Ks}, \Omega^2\sigmac$ in $L^2(\Cb_{\ub})$, using the bootstrap assumptions \eqref{bootstrap}.  We consider the following term as an example
    $$\Omega\chih\cdot(\Omega^2\betab+\Lb\phi\Omega\nablas\phi)\cdot (\Omega^2\beta-L\phi\Omega\nablas\phi).$$
    After placing them in suitable norms and taking their signatures into account, they are estimated by
    $$\delta|u|^3\cdot\varepsilon^{-\kappa}\frac{a}{|u|^{1-\frac{1}{2}p+\frac{1}{2}ps(\chih)}}\cdot\varepsilon^{-\kappa}\left(\frac{\delta}{|u|^{1-p}}\right)^{\frac{3}{2}+\gamma}\frac{a^{\frac{3}{2}}}{|u|^{2-p+\frac{1}{2}ps(\betab)}}\cdot\varepsilon^{-\kappa} \left(\frac{\delta}{|u|^{1-p}}\right)^{\gamma}\frac{a}{|u|^{2-p+\frac{1}{2}ps(\beta)}},$$
      where the first factor comes from integration over $\mathcal{M}$,  which is
    $$\varepsilon^{-3\kappa}\frac{\delta}{|u|^{2-\frac{5}{2}p-\frac{1}{2}p(s(\chih)+s(\betab)+s(\beta))}}\left(\frac{\delta}{|u|^{1-p}}\right)^{\frac{3}{2}+2\gamma}a^{\frac{7}{2}}=\varepsilon^{-3\kappa}\frac{\delta}{|u|^{2-\frac{5}{2}p-\frac{1}{2}p\cdot2(s(\beta)+1)}}\left(\frac{\delta}{|u|^{1-p}}\right)^{\frac{3}{2}+2\gamma}a^{\frac{7}{2}},$$
    where $s(\chih)+s(\betab)=s(\beta)+1$ is the relation of the signatures on both sides of the equations. 
   Recall that the weight 
    $$\delta^{-1-2\gamma}|u|^{2+2\gamma-p(3+2\gamma)+2i}$$
    of the spacetime integral is designed as
    $$\delta^{-1}|u|^{-2}\cdot \left((\delta|u|^{p-1})^{-\gamma} \cdot |u|^{2-p+\frac{1}{2}ps(\beta)}\cdot|u|^i\right)^2,$$
    where the first factor corresponds to the integral over $C_u$.     Therefore the spacetime integral of this term
   \begin{align*} &\delta^{-1-2\gamma}\int_{\mathcal{M}}|u|^{2+2\gamma-p(3+2\gamma)+2i}|\nablas^i(\Omega\chih\cdot(\Omega^2\betab+\Lb\phi\Omega\nablas\phi))\cdot \nablas^i(\Omega^2\beta-L\phi\Omega\nablas\phi)|\\
   \lesssim&\varepsilon^{-3\kappa}(\delta|u_1|^{p-1})^{\frac{3}{2}}a^{\frac{7}{2}}\lesssim a^2
   \end{align*}
   if $\varepsilon$ is sufficiently small. The other terms are treated similarly, and the key fact is that we will eventually have an extra $(\delta|u_1|^{p-1}a)$ to some positive power. This fact can be checked term by term without invoking the signature, while the signature structure provides a natural explanation for why it should hold.
     \item (2) In this case, the $0$th-order derivative of $\Ks$ exists.
     $$\nablas^i(\Omega\psi)\cdot\Omega^2\Ks\cdot\nablas^i(\Omega^2\Psi)$$
    where
    $$\Omega\psi\in \{\Omega\chih, \Omega\tr\chi, \omega, \Omega\eta, \Omega\etab\}$$
    $$\Omega^2\Psi\in\{\Omega^2\beta-L\phi\Omega\nablas\phi, \Omega^2\widetilde{\Ks}, \Omega^2\sigmac\}.$$
     $\Omega^2\Ks$ is estimated using \eqref{rhodaggeritself} and the others are the same to Case (1).
     \item (3) This case comes from commutators of $\nablas^i$ with the Hodge operators.
     $$\nablas^{i_1}(\Omega\Ks)\cdot\nablas^{i_2}(\Omega^2\widetilde{\Ks},\Omega^2\sigmac)\cdot\nablas^{i_3}(\Omega^2\beta-L\phi\Omega\nablas\phi)$$
     where $i_1+i_2+i_3=2i-1$ and $i_1+1,i_2,i_3\le i$. The first factor is estimated using \eqref{Gausscurvaturebound}. 
      \item (4) There are four factors involved, among which there are three connection coefficients and one curvature component.
      $$\nablas^i(\Omega\psi_1\cdot\Omega\psi_2\cdot\Omega\psi_3)\cdot\nablas^i(\Omega^2\Psi)$$
      where
       $$\Omega\psi_{1,2,3}\in \{\Omega\chih, \Omega\chibh, \Omega\tr\chi, \widetilde{\Omega\tr\chib},\Omega\eta, \Omega\etab\}$$
    $$\Omega^2\Psi\in\{\Omega^2\beta-L\phi\Omega\nablas\phi, \Omega^2\widetilde{\Ks}, \Omega^2\sigmac\}.$$
    We may place each connection coefficient in $H^j(S_{\ub,u})$ norm and the curvature component in $L^2(C_u)$ and $L^2(\Cb_{\ub})$. 
      \item (5) This case may involve the top-order derivatives of the connection coefficients.
         $$\nablas^i(\Omega\psi_1\cdot\nablas(\Omega^2\psi_2))\cdot\nablas^i(\Omega^2\Psi)$$
      where
      $$\Omega\psi_{1}\in \{\Omega\chih, \Omega\chibh, \Omega\tr\chi, \widetilde{\Omega\tr\chib}\}$$
      $$\Omega\psi_{2}\in \{\Omega\chih, \Omega\chibh, \Omega\tr\chi, \Omega\tr\chib, \Omega\etab\}$$
    $$\Omega^2\Psi\in\{\Omega^2\beta-L\phi\Omega\nablas\phi, \Omega^2\widetilde{\Ks}, \Omega^2\sigmac\}.$$
    The top-order derivatives of the connection coefficients can be viewed as a curvature component and they are estimated in $L^2(C_u)$ or $L^2(\Cb_{\ub})$ using $\mathcal{O}$ in \eqref{bootstrap}.
     \item (6) The following cases include the matter field. This case may involve the top-order derivatives of the matter field.
     $$\nablas^i(\Omega\nablas\Phi\cdot\Omega\nablas\phi)\cdot\nablas^i(\Omega^2\Psi)$$
     where
     $$\Phi\in\{\Omega\nablas\phi, L\phi\}$$
     $$\Omega^2\Psi\in\{\Omega^2\beta-L\phi\Omega\nablas\phi, \Omega^2\sigmac\}.$$
     We put $\Omega\nablas\phi$ in $H^j(S_{\ub,u})$, top order derivatives of scalar field in $L^2(C_u)$. 
     \item (7) There are four factors involved, among which there are two scalar field components, one connection coefficient and one curvature component.
      $$\nablas^i(\Omega\psi\cdot\Phi_1\cdot\Phi_2)\cdot\nablas^i(\Omega^2\beta-L\phi\Omega\nablas\phi)$$
       where $$\Omega\psi\in\{\Omega\chih, \Omega\tr\chi, \Omega\eta\}$$
       $$\Phi_{1,2}\in\{\Lb\phi, \Omega\nablas\phi\}.$$
       Lower order terms can be placed in $H^j(S_{\ub,u})$. 
    \end{itemize}

    Combining all estimates above, when $\varepsilon$ is sufficiently small, we will have
    \begin{equation}\label{curvatureestimate1}
    \|\beta\|_{\mathcal{A}}\lesssim 1.
    \end{equation}

    For $(\Ks,\sigmac)$-$\betab$ pair, we use the equations
    {\small
    \begin{align*}
        &\Db(\Ks-\frac{1}{|u|^2})+\frac{3}{2}\Omega\tr\chib (\Ks-\frac{1}{|u|^2})-\frac{1}{2}\Omega\tr\chib\mu-\divs(\Omega\betab+\Lb\phi\nablas\phi)-(\Omega\betab+\Lb\phi\nablas\phi)\cdot\eta+(\Omega\tr\chib+\frac{2}{|u|})\frac{1}{|u|^2}\\
        &-\Omega\chibh\cdot\nablas\eta-\Omega\chibh\cdot\eta\cdot\eta+\frac{1}{2}\Omega\tr\chib|\eta|^2=0\\
    &\Db\sigmac+\frac{3}{2}\Omega\tr\chib\sigmac+\curls(\Omega\betab+\Lb\phi\nablas\phi)+\eta\wedge(\Omega\betab+\Lb\phi\nablas\phi)-\frac{1}{2}\Omega\chibh\wedge(\eta\tensor\eta+\nablas\tensor\eta)=2\nablas \Lb\phi\wedge\nablas\phi\\
    &D(\Omega\betab+\Lb\phi\nablas\phi)+\frac{1}{2}\Omega\tr\chi(\Omega\betab+\Lb\phi\nablas\phi)-\Omega\chih\cdot(\Omega\betab+\Lb\phi\nablas\phi)-\Omega^2\ds \Ks-\Omega^2{}^*\ds\sigmac-3\Omega^2(\etab \Ks+{}^*\etab\sigmac)\\
    &-2\Omega\chibh\cdot(\Omega\beta-L\phi\nablas\phi)+\frac{1}{2}\Omega^2(\ds(\chih,\chibh)-{}^*\ds(\chih\wedge\chibh))+\frac{3}{2}\Omega^2(\etab(\chih,\chibh)+{}^*\etab(\chih\wedge\chibh))-\frac{1}{4}\Omega^2\ds(\tr\chi\tr\chib)\\
    &-\frac{3}{4}\Omega^2\tr\chi\tr\chib\etab=2\Omega^2\Deltas\phi\nablas\phi-\Omega^2\ds|\ds\phi|^2+2\Omega\chibh\cdot\nablas\phi L\phi-\Omega\tr\chib L\phi\nablas\phi+2\Omega^2\etab\cdot\nablas\phi\nablas\phi-2\Omega^2\etab|\ds\phi|^2.
    \end{align*}}
    Then we will reorganize the above equation in the following form
    {\small
    \begin{align*}
    &\Db(\Omega^2(\Ks-\frac{1}{|u|^2}))+\frac{3}{2}\Omega\tr\chib\Omega^2(\Ks-\frac{1}{|u|^2})-2\omegab\Omega^2(\Ks-\frac{1}{|u|^2})-\Omega\divs(\Omega^2\betab+\Lb\phi\Omega\nablas\phi)\\
    &+(\Omega^2\betab+\Lb\phi\Omega\nablas\phi)\cdot(-\frac{1}{2}\Omega\eta+\frac{1}{2}\Omega\etab)+\Omega^2(\Omega\tr\chib+\frac{2}{|u|})\frac{1}{|u|^2}-\frac{1}{2}\Omega\tr\chib\Omega^2\mu-\Omega\chibh\cdot\Omega\nablas(\Omega\eta)-\frac{1}{2}\Omega\chibh\cdot(\Omega\eta-\Omega\etab)\cdot\Omega\eta\\
    &+\frac{1}{2}\Omega\tr\chib|\Omega\eta|^2=0\\
    &\Db(\Omega^2\sigmac)+\frac{3}{2}\Omega\tr\chib\Omega^2\sigmac-2\omegab\Omega^2\sigmac+\Omega\curls(\Omega^2\betab+\Lb\phi\Omega\nablas\phi)+(\frac{1}{2}\Omega\eta-\frac{1}{2}\Omega\etab)\wedge(\Omega^2\betab+\Lb\phi\Omega\nablas\phi)\\
    &-\frac{1}{2}\Omega\chibh\wedge\left((\frac{1}{2}\Omega\eta-\frac{1}{2}\Omega\etab)\tensor\Omega\eta+\Omega\nablas\tensor(\Omega\eta)\right)=2\Omega\nablas \Lb\phi\wedge\Omega\nablas\phi\\
    &D(\Omega^2\betab+\Lb\phi\Omega\nablas\phi)+\frac{1}{2}\Omega\tr\chi(\Omega^2\betab+\Lb\phi\Omega\nablas\phi)-\omega(\Omega^2\betab+\Lb\phi\Omega\nablas\phi)-\Omega\chih\cdot(\Omega^2\betab+\Lb\phi\Omega\nablas\phi)-\Omega\ds(\Omega^2\Ks)\\
    &+\Omega(\eta-2\etab)\Omega^2\Ks-\Omega{}^*\ds(\Omega^2\sigmac)+{}^*\Omega(\eta-2\etab)\Omega^2\sigmac-2\Omega\chibh\cdot(\Omega^2\beta-L\phi\Omega\nablas\phi)+\frac{1}{2}\Omega\ds(\Omega\chih,\Omega\chibh)\\
    &+\Omega(-\frac{1}{2}\eta+\etab)(\Omega\chih,\Omega\chibh)-\frac{1}{2}\Omega{}^*\ds(\Omega\chih\wedge\Omega\chibh)+\Omega(\frac{1}{2}\eta+2\etab)(\Omega\chih\wedge\Omega\chibh)-\frac{1}{4}\Omega\ds(\Omega\tr\chi\Omega\tr\chib)+\Omega(\frac{1}{4}\eta-\frac{1}{2}\etab)\Omega\tr\chi\Omega\tr\chib\\
    &=2\Omega^2\Deltas\phi\Omega\nablas\phi-\Omega^3\ds|\ds\phi|^2+2\Omega\chibh\cdot\Omega\nablas\phi L\phi-\Omega\tr\chib L\phi\Omega\nablas\phi+2\Omega\etab\cdot\Omega\nablas\phi\Omega\nablas\phi-2\Omega\etab|\Omega\ds\phi|^2.
    \end{align*}}

    Note that $\Omega^2(\Omega\tr\chib+\frac{2}{|u|})\frac{1}{|u|^2}$ and $\Omega\tr\chib\Omega^2\mu$ are two bad terms in the equation of $\Db(\Omega^2\widetilde{\Ks})$, that is, they will cause a logarithmic loss in the above integral. The cause of this problem and the solution are similar to \cite{An-Luk, Li-Liu1}. So we first consider the estimates of the derivatives of $\Omega^2\Ks, \Omega\sigmac, \Omega^2\betab+\Lb\phi\Omega\nablas\phi$, and the zero-order estimates will be obtained by directly integrating the null Bianchi equations, as in deriving \eqref{Gausscurvaturebound}. The derivative of $\widetilde{\Omega\tr\chib}$ and $\mu$ will have better estimates. Here the estimates of the derivatives of $\Omega\tr\chib$ in \eqref{bootstrap} is insufficient. We need to derive a better estimate through the equation $D\nablas(\Omega\tr\chib)$. 

    We compute for $1\le i\le N$, 
    {\small
		\begin{align*}
			&\Db\left(\delta^{-3-2\gamma}|u|^{4+2\gamma-2p(2+\gamma)+2i}\left(|\nablas^i(\Omega^2(\Ks-|u|^{-2}))|^2+|\nablas^i(\Omega^2\sigmac)|^2\right)\D\mu_{\gs}\right)\\
            &+D\left(\ub^{-3-2\gamma}|u|^{4+2\gamma-2p(2+\gamma)+2i}|\nablas^i(\Omega^2\betab+\Lb\phi\Omega\nablas\phi)|^2\D\mu_{\gs}\right)\\
			=&\delta^{-3-2\gamma}|u|^{4+2\gamma-2p(2+\gamma)+2i}\divs((\nablas^i(\Omega^2(\Ks-|u|^{-2})),\nablas^i(\Omega^2\sigmac
            ))\cdot\nablas^i(\Omega^2\betab+\Lb\phi\Omega\nablas\phi))\D\mu_{\gs}\\
			&\underline{-((2\gamma-2p(2+\gamma)|u|^{-1}-4\omegab|_{\Cb_0}))\left(\delta^{-3-2\gamma}|u|^{4+2\gamma-2p(2+\gamma)+2i}\left(|\nablas^i(\Omega^2(\Ks-|u|^{-2}))|^2+|\nablas^i(\Omega^2\sigmac)|^2\right)\D\mu_{\gs}\right)}\\
			&+\delta^{-3-2\gamma}|u|^{4+2\gamma-2p(2+\gamma)+2i}\tau_{2}\D\mu_{\gs}.
		\end{align*}}
   The coefficient of the underlined term is also nonpositive.  We can obtain the desired estimates through a similar method as above. Here we only carefully discuss the estimates of the derivatives of two bad terms. We only need to consider  the case when all $\nablas^i$  apply to $\Omega\tr\chib$ or $\mu$.
    \begin{itemize}
        \item (1) $$\int_{\mathcal{M}}\delta^{-3-2\gamma}|u|^{4+2\gamma-2p(2+\gamma)+2i}\Omega^2\nablas^i(\widetilde{\Omega\tr\chib})\frac{1}{|u|^2}\cdot\nablas^i(\Omega^2\widetilde{\Ks})\D\mu_{\gs},\quad i\geq1$$
    From the following equation, 
    $$D(\Omega\tr\chib)=2\Omega^2\divs\etab+2|\Omega\etab|^2-\Omega\tr\chi\Omega\tr\chib-2\Omega^2\Ks+2|\Omega\nablas\phi|^2$$
    by taking derivative and applying \eqref{Gronwallub}, and in particular we use $\mathcal{O}$ in \eqref{bootstrap}, we can obtain
    $$|u|\|\nablas(\Omega\tr\chib)\|_{H^{N-1}(S_{\ub,u})}\lesssim\epsilon^{-\kappa}(\delta|u|^{-1+p})^{\frac{3}{2}+\gamma}a^{\frac{3}{2}},$$    
    so we have
    \begin{align*}
    &\int_{\mathcal{M}}\delta^{-3-2\gamma}|u|^{4+2\gamma-2p(2+\gamma)+2i}\Omega^2\nablas^i(\widetilde{\Omega\tr\chib})\frac{1}{|u|^2}\cdot\nablas^i(\Omega^2\widetilde{\Ks})\D\mu_{\gs}\\
    \lesssim&\int_{\mathcal{M}}|u|^{-\frac{3}{2}+\frac{1}{2}p}\cdot(\delta|u|^{-1+p})^{-\frac{3}{2}-\gamma}|u|^i\nablas^i(\Omega\tr\chib)\cdot\delta^{-\frac{3}{2}-\gamma}|u|^{2+\gamma-p(2+\gamma)+i}\nablas^i(\Omega^2\Ks)\\
    \lesssim&(\delta|u|^{-1+p})^{\frac{1}{2}}\varepsilon^{-2\kappa}a^3\lesssim a^3
    \end{align*}
    \item (2)
    $$\int_{\mathcal{M}}\delta^{-3-2\gamma}|u|^{4+2\gamma-2p(2+\gamma)+2i}\Omega\tr\chib\nablas^i(\Omega^2\mu)\cdot\nablas^i(\Omega^2\widetilde{\Ks})\D\mu_{\gs},\quad i\geq1$$
    $\mu$ satisfies the following equation:
    $$D\mu+\Omega\tr\chi\mu=-\Omega\tr\chi\frac{1}{|u|^2}+\divs(2\Omega\chih\cdot\eta-\Omega\tr\chi\etab)+2\nablas L\phi\cdot\ds\phi+2L\phi\Deltas\phi$$
    so we have
    $$|u|\|\nablas(\Omega^2\mu)\|_{H^{N-1}(S_{\ub,u})}\lesssim\epsilon^{-\kappa}(\ub|u|^{-1+p})^{\frac{3}{2}+\gamma}|u|^{-1+p}a^{\frac{3}{2}}$$
    and
    \begin{align*}
    &\int_{\mathcal{M}}\delta^{-3-2\gamma}|u|^{4+2\gamma-2p(2+\gamma)+2i}\Omega\tr\chib\nablas^i(\Omega^2\mu)\cdot\nablas^i(\Omega^2\widetilde{\Ks})\D\mu_{\gs}\\
    \lesssim&\int_{\mathcal{M}}|u|^{-\frac{3}{2}+\frac{1}{2}p}\cdot|u|\Omega\tr\chib\cdot(\delta|u|^{-1+p})^{-\frac{3}{2}-\gamma}|u|^{1-p+i}\nablas^i(\Omega^2\mu)\cdot\delta^{-\frac{3}{2}-\gamma}|u|^{2+\gamma-p(2+\gamma)+i}\nablas^i(\Omega^2\Ks)\\
    \lesssim&(\delta|u|^{-1+p})^{\frac{1}{2}}a^3\varepsilon^{-2\kappa}\lesssim a^3.
    \end{align*}
    \item  (3) We may also look at the term in $\tau_2$
    $$\nablas^i(\Omega(\eta,\etab))\Omega\tr\chi\Omega\tr\chib\cdot\nablas^i(\Omega\betab+\Lb\phi\Omega\nablas\phi)$$
    which comes from the equation $D( (\Omega\betab+\Lb\phi\Omega\nablas\phi))$. If we only look at the power of $\delta$ and $a$, this term looks like $\delta\cdot(\delta a\cdot a\cdot 1\cdot \delta^{\frac{3}{2}}a^{\frac{3}{2}})$ where the first $\delta$ represents the region of integration. By absorbing $\delta^{\frac{1}{2}}$, we will have $(\delta a)^3$, which is the expected bound of $\Omega\betab+\Lb\phi\Omega\nablas\phi$. So the power of $a$ in the definition of $\|\betab\|_{\mathcal{A}}$ is $\frac{3}{2}$ instead of $1$ as in most definitions of norms of different quantities. 
    \end{itemize}
   
    Combining the previous Bianchi pair's estimate, we then have
	$$\|\Ks,\sigma\|_{\mathcal{A}}+\|\betab\|_{\mathcal{A}}\lesssim 1,$$
	except 
    where the estimates of $\sigmac$ and $\betab$ themselves, we directly integrate the Bianchi equations to obtain what we want.

	\end{proof}

   \section{The estimates for top-order derivatives of scalar field}\label{sec:toporderscalar}
   \begin{proposition}
    Under the assumptions of \eqref{bootstrap},  if $\varepsilon$ is small enough, for $\delta|u|^{-1+p}\le\varepsilon$, we have
    $$\mathcal{E}\lesssim 1$$
   \end{proposition}
   \begin{proof}
    We rewrite the wave equation in the following form:
    \begin{align*}
    &\Db(\Omega\nablas L\phi)+\frac{1}{2}\Omega\tr\chib\Omega\nablas L\phi-\omegab\Omega\nablas L\phi\\=&\Omega\nablas(\Omega^2\Deltas\phi)-\frac{1}{2}\Omega\nablas(\Omega\tr\chib)L\phi+2\Omega\nablas\left((\Omega\eta,\Omega\nablas\phi)\right)-\frac{1}{2}\Omega\nablas(\Omega\tr\chi\Lb\phi)\\
    &D(\Omega^2\nablas\nablas\phi)-2\omega\Omega^2\nablas\nablas\phi=\Omega\nablas(\Omega\nablas L\phi)-\frac{1}{2}\Omega(\eta+\etab)\Omega\nablas L\phi+\Omega\nablas(\Omega\chi)\Omega\nablas\phi\\
    &\Db(\Omega^2\nablas\nablas\phi)-2\omegab\Omega^2\nablas\nablas\phi=\Omega\nablas(\Omega\nablas\Lb\phi)-\frac{1}{2}\Omega(\eta+\etab)\Omega\nablas\Lb\phi+\Omega\nablas(\Omega\chib)\Omega\nablas\phi\\
    &D(\Omega\nablas\Lb\phi)+\frac{1}{2}\Omega\tr\chi\Omega\nablas\Lb\phi-\omega\Omega\nablas\Lb\phi\\=&\Omega\nablas(\Omega^2\Deltas\phi)-\frac{1}{2}\Omega\nablas(\Omega\tr\chi)\Lb\phi+2\Omega\nablas((\Omega\etab,\Omega\nablas\phi))
    -\frac{1}{2}\Omega\nablas(\Omega\tr\chib L\phi)
    \end{align*}

    For the pair $\Omega\nablas L\phi$--$\Omega^2\nablas\nablas\phi$, we can compute
     \begin{align*}
    &\Db\left(\delta^{-1-2\gamma}|u|^{2+2\gamma-p(3+2\gamma)+2i}|\nablas^i(\Omega\nablas L\phi)|^2\D\mu_{\gs}\right)\\
    &+D\left(\delta^{-1-2\gamma}|u|^{2+2\gamma-p(3+2\gamma)+2i}|\nablas^i(\Omega^2\nablas\nablas\phi)|^2\D\mu_{\gs}\right)\\
    =&\delta^{-1-2\gamma}|u|^{2+2\gamma-p(3+2\gamma)+2i}\divs(\nablas^i(\Omega\nablas L\phi)\cdot\nablas^i(\Omega^2\nablas\nablas\phi))\D\mu_{\gs}\\
    &-\left((2\gamma-p(3+2\gamma))-2|u|\omegab|_{\Cb_{0}}\right)|u|^{-1}\left(\delta^{-1-2\gamma}|u|^{2+2\gamma-p(3+2\gamma)+2i}|\nablas^i(\Omega\nablas L\phi)|^2\D\mu_{\gs}\right)\\
    &+\delta^{-1-2\gamma}|u|^{2+2\gamma-p(3+2\gamma)+2i}\tau_{3}\D\mu_{\gs}.
     \end{align*}
      In view of \eqref{omegablower}, the coefficient 
     $$(2\gamma-p(3+2\gamma))-2|u|\omegab\big|_{\Cb_{0}}\ge k^2-\epsilon-p>0,$$
     which is positive. Integrating over $(\mathcal{M},\D\ub\D u\D\mu_{\gs})$, we have
     \begin{align*}
    &\delta^{-1-2\gamma}\int_{C_u}|u|^{2+2\gamma-p(3+2\gamma)+2i}|\nablas^i(\Omega\nablas L\phi)|^2+\delta^{-1-2\gamma}\int_{\Cb_{\ub}}|u|^{2+2\gamma-p(3+2\gamma)+2i}|\nablas^i(\Omega^2\nablas\nablas\phi)|^2\\
    &+\delta^{-1-2\gamma}\int_{\mathcal{M}}|u|^{1+2\gamma-p(3+2\gamma)+2i}|\nablas^i(\Omega\nablas L\phi)|^2\\
    \lesssim&a^2+\delta^{-1-2\gamma}\int_{\mathcal{M}}|u|^{2+2\gamma-p(3+2\gamma)+2i}|\tau_{3}|.
     \end{align*}

   By conducting a similar analysis on the terms composed of $\tau_3$, just as was done for the previous curvature estimates, we have
   $$\delta^{-1-2\gamma}\int_{\mathcal{M}}|u|^{2+2\gamma-p(3+2\gamma)+2i}|\tau_{3}|\lesssim a^2.$$
   So we can obtain 
   $$\|L\phi\|_{\mathcal{F}}\lesssim 1.$$
 and in particular
   \begin{equation}\label{Lphispacetime}\delta^{-1-2\gamma}\int_{\mathcal{M}}|u'|^{1+2\gamma-p(3+2\gamma)+2i}|\nablas^i(\Omega\nablas L\phi)|^2\lesssim a^2\end{equation}
which will be used later.

\begin{remark}\label{dependenceonk2-epsilon-p}
This is where the dependence on $k^2-\epsilon-p$ arises. So the estimates blows up when $k^2-\epsilon-p\to0$. The estimate of $\tau_3$ is in fact better
   $$\delta^{-1-2\gamma}\int_{[0,\delta]\times[u_0,u]\times S^2}|u|^{2+2\gamma-p(3+2\gamma)+2i}|\tau_{3}|\lesssim a^2\cdot(\delta|u|^{p-1}a)^{\frac{1}{2}}.$$
   For convenience, if we assume $\nablas L\phi=0$ for $\ub\in[0,\delta]$ on $C_{u_0}$, which is allowed in our Theorem \ref{main1}, then we have
  $$ \delta^{-1-2\gamma}\int_{C_u}|u|^{2+2\gamma-p(3+2\gamma)+2i}|\nablas^i(\Omega\nablas L\phi)|^2\lesssim a^2\cdot (\delta|u|^{p-1}a)^{\frac{1}{2}} $$
  and hence the spacetime integral \eqref{Lphispacetime} still holds true without invoking the lower bound from the damping term. This removes the dependence on $k^2-\epsilon-p$.
\end{remark}
   
    For the other pair $\Omega^2\nablas\nablas\phi$--$\Omega\nablas\Lb\phi$, we can compute
    \begin{align*}
    &\Db\left(\delta^{-2-2\gamma}|u|^{3+2\gamma-p(3+2\gamma)+2i}|\nablas^i(\Omega^2\nablas\nablas\phi)|^2\D\mu_{\gs}\right)\\&+D\left(\delta^{-2-2\gamma}|u|^{3+2\gamma-p(3+2\gamma)+2i}|\nablas^i(\Omega\nablas\Lb\phi)|^2\D\mu_{\gs}\right)\\
    =&\delta^{-2-2\gamma}|u|^{3+2\gamma-p(3+2\gamma)+2i}\divs(\nablas^i(\Omega^2\nablas\nablas\phi)\cdot\nablas^i(\Omega\nablas\Lb\phi))\D\mu_{\gs}\\
    &-\left((1+2\gamma-p(3+2\gamma))-4|u|\omegab|_{\Cb_{0}}\right)|u|^{-1}\left(\delta^{-2-2\gamma}|u|^{3+2\gamma-p(3+2\gamma)+2i}|\nablas^i(\Omega^2\nablas\nablas\phi)|^2\D\mu_{\gs}\right)\\
    &+\delta^{-2-2\gamma}|u|^{3+2\gamma-p(3+2\gamma)+2i}\tau_{4}\D\mu_{\gs}
     \end{align*}
    and
    \begin{align*}
    &\delta^{-2-2\gamma}\int_{C_u}|u|^{3+2\gamma-p(3+2\gamma)+2i}|\nablas^i(\Omega^2\nablas\nablas\phi)|^2+\delta^{-2-2\gamma}\int_{\Cb_{\ub}}|u|^{3+2\gamma-p(3+2\gamma)+2i}|\nablas^i(\Omega\nablas\Lb\phi)|^2\\
    &+\left((1+2\gamma-p(3+2\gamma))-4\omegab|_{\Cb_{0}}|u|\right)\delta^{-2-2\gamma}\int_{\mathcal{M}}|u|^{2+2\gamma-p(3+2\gamma)+2i}|\nablas^i(\Omega^2\nablas\nablas\phi)|^2\\
    \lesssim&a^2+\delta^{-2-2\gamma}\int_{\mathcal{M}}|u|^{3+2\gamma-p(3+2\gamma)+2i}|\tau_{4}|
     \end{align*}
     where 
     $$\left((1+2\gamma-p(3+2\gamma))-4\omegab|_{\Cb_{0}}|u|\right)\ge 1-p+2(k^2-\epsilon)>1+p>0.$$
     Among the spacetime integrant $\delta^{-2-2\gamma}|u|^{3+2\gamma-2p(1+\gamma)+2i}|\tau_{4}|$, the following term is the one that  requires the most attention:
     \begin{align*}
    & \int_{\mathcal{M}}\delta^{-2-2\gamma}|u|^{3+2\gamma-p(3+2\gamma)+2i}|\underbrace{\Omega\tr\chib}_{\lesssim\frac{1}{|u|}}\cdot\underbrace{\nablas^i(\Omega\nablas L\phi)}_{\eqref{Lphispacetime}}\cdot \underbrace{\nablas^i(\Omega\Lb\phi)}_{\|\Lb\phi\|_{\mathcal{F}}}|\\
    \lesssim&\left(\delta^{-1-2\gamma}\int_{\mathcal{M}}|u'|^{1+2\gamma-p(3+2\gamma)+2i}|\nablas^i(\Omega\nablas L\phi)|^2\right)^{\frac{1}{2}}\\
    &\times\left(\int_0^\delta\delta^{-1}\D\ub\int_{\Cb_{\ub}}\delta^{-2-2\gamma}|u'|^{3+2\gamma-p(3+2\gamma)+2i}|\nablas^i(\Omega\nablas\Lb\phi)|^2\right)^{\frac{1}{2}}\\
    \lesssim& a^2\|\Lb\phi\|_{\mathcal{F}}.
     \end{align*}
   where we have used  the crucial spacetime estimate \eqref{Lphispacetime} for $\Omega\nablas L\phi$. The above term is sublinear and can be absorbed by the left hand side. We omit the remaining terms that won't cause problems, and thus we can obtain the estimates we want, that is:
     $$\|\Omega \nablas\phi\|_{\mathcal{F}}+\|\Lb\phi\|_{\mathcal{F}}\lesssim 1.$$
     So we are done.
     
    \end{proof}

\section{The estimates for top-order derivatives of connections}\label{sec:toporderconnection}

	\begin{proposition}Under the assumptions \eqref{bootstrap},  if $\varepsilon$ is sufficiently small, we have
		\begin{equation}\label{connection-top-bound}\mathcal{O}\lesssim 1. \end{equation}
	\end{proposition}
	\begin{proof}
    We have established 
    $$\Os,\Es,\mathcal{R},\mathcal{E}\lesssim1.$$
    The only remaining bootstrap assumption is
    \begin{equation}\label{bootstrapON+1Ob}
    \mathcal{O}\le\varepsilon^{-\kappa}.
    \end{equation}
    This proof is mainly carried out in six groups. The estimates are obtained by using the elliptic-transport systems. Since we have \eqref{Gausscurvaturebound}, elliptic estimate \eqref{elliptic} holds in the following form
		\begin{equation}\label{elliptic1}
			\|\Omega\theta\|_{H^{N+1}(S_{\ub,u})}\lesssim \||u|(\Omega \divs\theta, \Omega \curls\theta)\|_{H^N(S_{\ub,u})}+|u|^{\frac{1}{2}p}\|\theta\|_{H^N(S_{\ub,u})},
		\end{equation}
	Among these six groups, what we need to pay the most attention to are the estimates of $\tr\chib$ and $\omegab$. It should also be noted that for the top order derivative of connection, there will be a loss of $\frac{\Omega}{|u|^{\frac{1}{2}p}}$ as compared to the corresponding lower order estimates.

	For $\chih, \tr\chi$, we use
	\begin{align*}
		&\begin{dcases}\divs(\Omega\chih)&=\frac{1}{2}\Omega^2\ds\tr \chi'+\Omega\chih\cdot\etab+\frac{1}{2}\Omega\tr \chi\eta-(\Omega\beta-L\phi\nablas\phi),\\
			D\tr\chi'&=-\frac{1}{2}(\Omega\tr\chi')^2-|\chih|^2-2(\widehat{L}\phi)^2.
		\end{dcases}
	\end{align*}
	From $\Os<\infty$, we have 
	$$|u|^{-1}\|\Omega\chih\cdot\Omega\etab+\frac{1}{2}\Omega\tr \chi\cdot\Omega\eta\|_{H^N(S_{\ub,u})}\lesssim(\delta|u|^{p-1})^{1+\gamma}|u|^{-2+\frac{3}{2}p}a^2.$$
	The assumption \eqref{bootstrapON+1Ob} for $\tr\chi$ gives
	$$|u|^{-1}\|\Omega\cdot\Omega^2\ds\tr \chi'\|_{H^N(S_{\ub,u})}\lesssim\varepsilon^{-\kappa}(\delta|u|^{p-1})^{1+\gamma}|u|^{-1+\frac{3}{2}p}a^2$$
	together with the $\mathcal{R}$ bound for $\beta$, when $\varepsilon$ is small enough,  we have 
	$$\delta^{-1-2\gamma}\int_0^{\delta}|u|^{2+2\gamma-p(3+2\gamma)}\|\Omega\divs(\Omega\chih)\|^2_{H^N(S_{\ub,u})}\D\ub\lesssim a^2.$$
	We will apply \eqref{elliptic1} for $\chih$, and then integrate over $\ub$ with suitable weight. We need to utilize the lower order estimates of $\chih$, and it does not contain $\varepsilon^{-\kappa}$.
	By applying \eqref{elliptic1}, we have the desired bound
	\begin{equation}\label{chihtopbound}\delta^{-1-2\gamma}\int_0^{\delta}|u|^{2\gamma-p(3+2\gamma)}\|\Omega^2\chih\|^2_{H^{N+1}(S_{\ub,u})}\D\ub\lesssim a^2.\end{equation}
	Commute $\nablas$ with the equation for $\tr\chi'$, we have
	$$D\nablas(\tr\chi')=-\Omega\tr\chi'\nablas(\Omega\tr\chi')- 2\chih\cdot\nablas\chih-4\widehat{L}\phi\nablas\widehat{L}\phi.$$
	The assumption \eqref{bootstrapON+1Ob} for $\tr\chi$ (together with its lower order estimates and \eqref{derivativelapse}) gives
	\begin{align*}
&|u|^{-1}\|\Omega\tr\chi'\nablas(\Omega\tr\chi')\|_{H^N(S_{\ub,u})}\\\lesssim&|u|^{-1}\|\Omega^{-3}\Omega(\eta+\etab)\cdot\Omega\tr\chi\cdot\Omega\tr\chi\|_{H^N(S_{\ub,u})}+|u|^{-1}\|\Omega^{-3}\Omega\tr\chi\cdot\nablas(\Omega^2\tr\chi)\|_{H^N(S_{\ub,u})}\\
	\lesssim&(\Omega|u|^{-\frac{1}{2}p})^{-3}\delta^{-1}(\delta|u|^{-1+p})^{2+\gamma}|u|^{-2}a^3\varepsilon^{-\kappa},
	\end{align*}
	By \eqref{chihtopbound}, we have
    $$\int_0^{\delta}|u|^{-1}\|\chih\cdot\nablas\chih\|_{H^N(S_{\ub,u})}\D\ub\lesssim(\Omega|u|^{-\frac{1}{2}p})^{-3}(\delta|u|^{-1+p})^{1+\gamma}|u|^{-2}a^2.$$
	Using the estimates of $\Omega\nablas L\phi$, in a similar way, we have
	\begin{align*}&\int_0^{\delta}|u|^{-1}\|\widehat{L}\phi\nablas\widehat{L}\phi\|_{H^N(S_{\ub,u})}\D\ub\lesssim(\Omega|u|^{-\frac{1}{2}p})^{-3}(\delta|u|^{-1+p})^{1+\gamma}|u|^{-2}a^2.\end{align*}
    Applying \eqref{Gronwallub} for $\nablas\tr\chi'$, for $\varepsilon$ sufficiently small, we have
    $$|u|^{-1}\|\nablas\tr\chi'\|_{H^N(S_{\ub,u})}\lesssim(\Omega|u|^{-\frac{1}{2}p})^{-3}(\delta|u|^{-1+p})^{1+\gamma}|u|^{-2}a^2,$$
    and hence
    $$|u|^{-1}\|\nablas(\Omega^2\tr\chi)\|_{H^N(S_{\ub,u})}\lesssim(\delta|u|^{-1+p})^{1+\gamma}|u|^{-2}|u|^{\frac{3}{2}p}a^2,$$
    which is the desired bound.
	
	We turn to $\chibh$ and $\tr\chib$, which are based on the equations:
	\begin{align*}
		&\begin{dcases}\divs(\Omega\chibh)&=\frac{1}{2}\Omega^2\ds\tr \chib'+\Omega\chibh\cdot\eta+\frac{1}{2}\Omega\tr \chib\etab+(\Omega\betab+\Lb\phi\nablas\phi),\\
			\Db\tr\chib'&=-\frac{1}{2}(\Omega\tr\chib')^2-|\chibh|^2-2(\widehat{\Lb}\phi)^2,
		\end{dcases}
	\end{align*}
    The estimate for top order derivatives of $\tr\chib$ is the most delicate estimate in the whole proof. Before starting the proof of this pair, we first analyze the estimates of the top order derivative of $\tr\chib$. Commute $\nablas$ with the second equation, we rewrite it as
    $$\Db\nablas(\Omega\tr\chib)+\Omega\tr\chib\nablas(\Omega\tr\chib)-2\omegab\nablas(\Omega\tr\chib)=2\Omega\tr\chib\nablas\omegab-2\Omega\chibh\cdot\nablas(\Omega\chibh)-4\Lb\phi\nablas\Lb\phi.$$
     The worst terms on the right hand side of the above equation are $\Omega\tr\chib\nablas\omegab$ and $\Lb\phi\nablas\Lb\phi$. If we have more $\nablas$, then we only need to pay attention to the terms in which all $\nablas$'s apply to $\omegab$ and one $\Lb\phi$. To handle this, we first note that the linear term $-2\omegab\nablas(\Omega\tr\chib)$ has a favorable sign, and we choose not to simply drop it, but instead exploit its positivity. We rewrite this equation further as
        $$\Db(\Omega^{-2}\nablas(\Omega\tr\chib))+\Omega\tr\chib\Omega^{-2}\nablas(\Omega\tr\chib)=2\Omega^{-2}\Omega\tr\chib\nablas\omegab-2\Omega^{-2}\Omega\chibh\cdot\nablas(\Omega\chibh)-4\Omega^{-2}\Lb\phi\nablas\Lb\phi,$$
     We apply Gronwall \eqref{Gronwallu} to $\Omega^{-2}\nablas(\Omega\tr\chib)$ with $s=1, \nu=2$. The first term on the right hand side is estimated as
    \begin{align*}
    &\int_{u_0}^u\||u'|^2\Omega^{-2}\Omega\tr\chib\nablas\omegab\|_{H^N(S_{\ub,u'})}\D u'\\
    \lesssim&\left(\int_{u_0}^{u}\Omega^{-6}\delta^{2+2\gamma}|u'|^{-1-2\gamma+p(3+2\gamma)}\D u'\right)^{\frac{1}{2}}\left(\int_{u_0}^{u}\delta^{-2-2\gamma}|u'|^{3+2\gamma-p(3+2\gamma)}\|\nablas(\Omega\omegab)\|^2_{H^N(S_{\ub,u'})}\D u'\right)^{\frac{1}{2}}\\
    \lesssim&\Omega^{-2}(\Omega|u|^{-\frac{p}{2}})^{-1}(\delta|u|^{-1+p})^{1+\gamma}|u|\|\omegab\|_{\mathcal{C}}
    \end{align*}
    where we have used the following elementary inequality for a positive function $f$
 \begin{equation}\label{eleinequality} \int_{u_0}^u \frac{1}{f(u')|u'|^{1-\sigma}}\D u'\le\frac{1}{\alpha-\sigma}\frac{1}{f(u)|u|^{-\sigma}}, u\in [u_0,u_1]\end{equation}
     provided that
    $$\frac{f(u')}{f(u'')}\le\left(\frac{|u'|}{|u''|}\right)^\alpha, u_0\le u''\le u'\le u_1<0$$
   where $\alpha>\sigma$ are two numbers.   In the above estimate we have set
   $$f(u)=\Omega^6(u), \quad \alpha=3p, \quad\sigma=-2\gamma+p(3+2\gamma)\le p.$$
    \begin{remark}\label{dependenceon1/p}
    This is where the dependence on $\frac{1}{p}$ arises.  In the work \cite{Li-Liu1}, where the case $p=\gamma=0$ is considered, there is a logarithmic loss in this estimate. No quantitative assumptions about $\Omega_0$ are imposed there.  This is the only place where the stronger condition \eqref{relativecontrol} is needed; elsewhere, the weaker point-wise bound $\Omega^2(u)\le |u|^p$ suffices.
    \end{remark}
   The third term of the right hand side is estimated similarly
    \begin{align*}\int_{u_0}^u\||u'|^2\Omega^{-2}\Lb\phi\nablas\Lb\phi\|_{H^N(S_{\ub,u'})}\D u'  \lesssim&\Omega^{-2}(\Omega|u|^{-\frac{p}{2}})^{-1}(\delta|u|^{-1+p})^{1+\gamma}|u|a
\end{align*}
   and the term $\Omega^{-2}\Omega\chibh\cdot\nablas(\Omega\chibh)$ is easier to treat.    Combining all the estimates above, and if $\varepsilon$ is small enough, we have the  bound,
   \begin{equation}\label{trchibtopbound}|u|^{-1}\|(|u|\nablas)(\Omega^2\tr\chib)\|_{H^N(S_{\ub,u})}\lesssim (\delta|u|^{-1+p})^{1+\gamma}|u|^{-1+\frac{1}{2}p}a(1+\|\omegab\|_{\mathcal{C}}).\end{equation}
Below we will show that $\|\omegab\|_{\mathcal{C}}\lesssim a$ so that this is the desired bound. Applying \eqref{elliptic1} to the equation $\divs(\Omega\chibh)$, using the estimate \eqref{trchibtopbound} above (together with the lower order bounds) we have the desired estimate
\begin{equation*}\delta^{-1-2\gamma}\int_{u_0}^u|u'|^{2\gamma-2p(1+\gamma)}\|\Omega^2\chibh\|_{H^{N+1}(S_{\ub,u'})}^2\D u'\lesssim a^2\end{equation*}
provided that $\|\omegab\|_{\mathcal{C}}\lesssim 1$.

We then turn to $\eta$. We use the equations
\begin{align*}
	&\begin{dcases}\divs\eta=\Ks-\frac{1}{|u|^2}-\mu,\\
		\curls\eta=\sigma-\frac{1}{2}\chih\wedge\chibh,\\
		D\mu+\Omega\tr\chi\mu=-\Omega\tr\chi\frac{1}{|u|^2}+\divs(2\Omega\chih\cdot\eta-\Omega\tr\chi\etab)+2\nablas L\phi\cdot\nablas\phi+2L\phi\Deltas\phi\end{dcases}
\end{align*}
By applying \eqref{Gronwallub}, the bootstrap assumptions \eqref{bootstrapON+1Ob} for $\chih$, $\tr\chi$, $\eta$ and $\etab$, and the lower order bounds, we can estimate each term on the right hand side. We obtain the following estimates when $\varepsilon$ is sufficiently small 
    $$|u|^{-1}\|\Omega^2\mu\|_{H^N(S_{\ub,u})}\lesssim\delta|u|^{p-1}|u|^{-2+p}a+(\delta|u|^{-1+p})^{\frac{3}{2}+\gamma}|u|^{-2+p}a^2\varepsilon^{-\kappa}.$$
Applying \eqref{elliptic1}, using in particular the curvature bounds, we have
\begin{equation}\label{etatopbound}\delta^{-2-2\gamma}\int_0^{\delta}|u|^{1+2\gamma-p(3+2\gamma)}\|\Omega^2\eta\|^2_{H^{N+1}(S_{\ub,u})}\D\ub\lesssim a^2\end{equation}
and 
\begin{equation*}\delta^{-1-2\gamma}\int_{u_0}^{u}|u'|^{2\gamma-p(3+2\gamma)}\|\Omega^2\eta\|^2_{H^{N+1}(S_{\ub,u'})}\D u'\lesssim a^2.\end{equation*}

For $\etab$, we use the equations
\begin{align*}
	&\begin{dcases}
		\divs\etab=\Ks-\frac{1}{|u|^2}-\mub,\\
		\curls\etab=-\sigma+\frac{1}{2}\chih\wedge\chibh,\\
		\Db\mub+\Omega\tr\chib\mub=-(\Omega\tr\chib+\frac{2}{|u|})\frac{1}{|u|^2}+\divs(2\Omega\chibh\cdot\etab-\Omega\tr\chib\eta)+2\nablas\Lb\phi\cdot\nablas\phi+2\Lb\phi\Deltas\phi\end{dcases}
\end{align*}

We apply \eqref{Gronwallu} to $\mub$ with $s=0, \nu=2$. For top order derivatives of $\chibh, \tr\chib$, we use the bootstrap assumption \eqref{bootstrapON+1Ob}, together with the lower order bounds, when $\varepsilon$ is small enough, we have
\begin{align*}\||u|\Omega^2\mub\|_{H^N(S_{\ub,u})}\lesssim&\||u_0|\Omega^2\mub\|_{H^N(S_{\ub,u_0})}+(\delta|u|^{-1+p})^{1+\gamma}|u|^pa+\varepsilon^{-\kappa}(\delta|u|^{-1+p})^{\frac{3}{2}+2\gamma}|u|^pa^2\\
	&+\int_{u_0}^u\|\nablas(\Omega^2\eta)\|_{H^{N}(\ub,u')}\D u'+\int_{u_0}^u\delta|u'|^{-1+p}a\|\nablas(\Omega^2\etab)\|_{H^{N}(\ub,u')}\D u'\\
    &+(\delta|u|^{-1+p})^{\frac{1}{2}+\gamma}|u|^pa.\end{align*}

For $\varepsilon$ sufficiently small, we have
\begin{align*}&\delta^{-2-2\gamma}\int_{0}^{\delta}|u|^{1+2\gamma-p(3+2\gamma)}\||u|\Omega^2\mub\|^2_{H^N(S_{\ub,u})}\D\ub\\
	\lesssim&a^2+\delta^{-2-2\gamma}\int_{0}^{\delta}|u|^{1+2\gamma-p(3+2\gamma)}\left(\int_{u_0}^u\delta|u'|^{-1+p}a\|\nablas(\Omega^2\etab)\|_{H^{N}(S_{\ub,u'})}\D u'\right)^2\D\ub\\
    &+\delta^{-2-2\gamma}\int_{0}^{\delta}|u|^{1+2\gamma-p(3+2\gamma)}\left(\int_{u_0}^u\|\nablas(\Omega^2\eta)\|_{H^{N}(S_{\ub,u'})}\D u'\right)^2\D\ub\\
	\lesssim&a^2,
\end{align*}
where we have used the estimate of $\eta$ that has already been proven above, the bootstrap assumption \eqref{bootstrapON+1Ob} for $\etab$. Applying \eqref{elliptic1} to $\etab$, we obtain the desired estimate
\begin{equation}\label{etabtopbound}\delta^{-2-2\gamma}\int_0^{\delta}|u|^{1+2\gamma-p(3+2\gamma)}\|\Omega^2\etab\|^2_{H^{N+1}(S_{\ub,u})}\D\ub\lesssim a^2.\end{equation}

We turn to $\omegab$, for which we use
\begin{align*}
&\begin{dcases}
  \Deltas\omegab&=\omegabs+\divs(\Omega\betab+\Lb\phi\nablas\phi),\\
D\omegabs&+\Omega\tr\chi\omegabs+2\Omega\chih\cdot\nablas\nablas\omegab+2\divs(\Omega\chih)\cdot\nablas\omegab-\frac{1}{2}\divs(\Omega\tr\chi(\Omega\betab+\Lb\phi\nablas\phi))\\
&+\nablas(\Omega^2)\cdot(\ds(\rho+\frac{1}{6}\mathbf{R})+{}^*\ds\sigma)+\Deltas(\Omega^2)(\rho+\frac{1}{6}\mathbf{R})-\Deltas(\Omega^2(2\eta\cdot\etab-|\eta|^2))\\
&-\divs(\Omega\chih\cdot(\Omega\betab+\Lb\phi\nablas\phi)-2\Omega\chibh\cdot(\Omega\beta-L\phi\nablas\phi)+3\Omega^2\etab(\rho+\frac{1}{6}\mathbf{R})-3\Omega^2{}^*\etab\sigma)\\
\phantom{\Delta}=&-\divs\{2\Omega^2\ds\phi\Deltas\phi+\Lb\phi\nablas L\phi+L\phi\nablas\Lb\phi-\Omega\tr\chib L\phi\ds\phi\\
&+2\Omega\chibh\cdot\ds\phi L\phi+2\Omega^2\etab\cdot\nablas\phi\ds\phi+\Omega^2\etab|\ds\phi|^2)\}  \end{dcases}
\end{align*}

It is direct to check that, from the second equation, the following spacetime estimate
$$\delta^{-1-2\gamma}|u|^{-p}\int_{[0,\delta]\times[u_0,u]}|u'|^{5+2\gamma-2p(1+\gamma)}\|\Omega D\omegabs+\Omega\tr\chi\Omega\omegabs\|_{H^{N-1}(\ub,u')}^2\D\ub \D u'\lesssim a^2$$
holds for $\varepsilon$ small enough. We have used the bootstrap assumptions \eqref{bootstrapON+1Ob} for $\omegab$, $\eta$ and $\etab$, the lower order estimates, the curvature estimates \eqref{curvaturebound} and the estimates of the scalar field. The term $\Lb\phi\Deltas L\phi$ is the worst term and the spacetime estimate \eqref{Lphispacetime} is utilized. By applying \eqref{Gronwallub} for $\omegabs$, squaring and integrating over $[u_0,u]$, we have
\begin{align*}&\delta^{-2-2\gamma}\int_{u_0}^u|u'|^{5+2\gamma-p(3+2\gamma)}\|\Omega\omegabs\|_{H^{N-1}(\ub,u')}^2\D u'\\
	\lesssim&\delta^{-2-2\gamma}\int_{u_0}^u|u'|^{5+2\gamma-p(3+2\gamma)}\left(\int_{0}^{\delta}\|\Omega D\omegabs+\Omega\tr\chi\Omega\omegabs\|_{H^{N-1}(\ub',u')}\D\ub'\right)^2\D u'\\
	\lesssim&\delta^{-1-2\gamma}\int_{[0,\delta]\times[u_0,u]} |u'|^{5+2\gamma-2p(3+2\gamma)}\|\Omega D\omegabs+\Omega\tr\chi\Omega\omegabs\|_{H^{N-1}(\ub',u')}^2\D\ub' \D u'\\
	\lesssim&a^2.
\end{align*}
The desired estimate for $\omegab$ follows from applying elliptic estimate \eqref{elliptic1}. But note that the lower order estimate for $\omegab$, which is established in $H^j(S_{\ub,u})$, is not sufficient. We must again utilize the spacetime estimate \eqref{Lphispacetime}. One can check the equations for $D\omegab$ to see that $\nablas\omegab$ (in order to apply \eqref{elliptic1}, $\omegab$ itself is not needed) obeys the estimate 
$$|u|^{-1}\|(|u|\nablas)\omegab\|_{H^{N-1}(S_{\ub,u})}\lesssim(\delta|u|^{-1+p})^{2+\gamma}|u|^{-1}a^2+|u|^{-1}\int_0^{\delta}\|(|u|\nablas)(L\phi\Lb\phi)\|_{H^{N-1}(S_{\ub,u})}\D \ub.$$
Squaring and integrating along $[u_0,u_1]$ gives the desired estimate for lower order derivatives of $\omegab$, and the top order estimate for $\omegab$ follows. 

The last thing is $\omega$, for which we use the equations:
\begin{align*}
&\begin{dcases}\Deltas\omega=\omegas-\divs(\Omega\beta)\\
\Db\omegas+\Omega\tr\chib\omegas+2\Omega\chibh\cdot\nablas\nablas\omega+2\divs(\Omega\chibh)\cdot\nablas\omega+\frac{1}{2}\divs(\Omega^2\tr\chib\beta)+\nablas(\Omega^2)\cdot(\ds(\rho+\frac{1}{6}R)+{}^*\ds\sigma)\\
+\Deltas(\Omega^2)(\rho+\frac{1}{6}R)-\Deltas(\Omega^2(2\eta\cdot\etab-|\etab|^2))-\divs(\Omega^2(-\chibh\cdot\beta+2\chih\cdot\betab+3\eta(\rho+\frac{1}{6}R)+3{}^*\eta\sigma))\\
=\divs(\Omega^2\ds\phi\Deltas\phi-2L\phi\nablas\Lb\phi-\Lb\phi\nablas L\phi-\frac{1}{2}\Omega\tr\chi\Lb\phi\ds\phi+\Omega\chibh\cdot\ds\phi L\phi+(\eta+\etab)\Lb\phi L\phi+\Omega^2\eta|\ds\phi|^2).\end{dcases}\\
\end{align*}
The worst term is still $\Lb\phi\Deltas L\phi$. We can compute that 
\begin{align*}
    &\delta^{-1-2\gamma}\int_{0}^{\delta}|u|^{2+2\gamma-p(3+2\gamma)}\left(\int_{u_0}^{u}\||u'|\Omega\Lb\phi\Deltas L\phi\|_{H^{N-1}(\ub,u')}\D u'\right)^2\D \ub\\
    \lesssim&\delta^{-1-2\gamma}\int_{0}^{\delta}|u|^{2+2\gamma-p(3+2\gamma)}\left(\int_{u_0}^{u}|u'|^{-3-2\gamma+2p(1+\gamma)}\D u'\right)\left(\int_{u_0}^{u}|u'|^{1+2\gamma-2p(1+\gamma)}\|\Omega\nablas L\phi\|^2_{H^{N-1}(\ub,u')}\D u'\right)\D \ub\\
    \lesssim&\delta^{-1-2\gamma}\int_{0}^{\delta}\int_{u_0}^{u}|u'|^{1+2\gamma-p(3+2\gamma)}\|\Omega\nablas L\phi\|^2_{H^{N-1}(\ub,u')}\D u'\D\ub\\
    \lesssim&a^2
\end{align*}
 Applying \eqref{Gronwallu} to $\omegas$  with $s=0, \nu=2$, squaring and integrating over $[u_0,u]$, we have
\begin{align*}&\delta^{-1-2\gamma}\int_{0}^{\delta}|u|^{2+2\gamma-p(3+2\gamma)}\||u|\Omega\omegas\|_{H^{N-1}(\ub,u)}^2\D \ub\\
	\lesssim&a^2+\delta^{-1-2\gamma}\int_{0}^{\delta}|u|^{2+2\gamma-p(3+2\gamma)}\left(\int_{u_0}^{u}\||u'|\Omega(\Db\omegas+\Omega\tr\chib\omegas)\|_{H^{N-1}(\ub,u')}\D u'\right)^2\D \ub\\
	\lesssim &a^2+\delta^{-1-2\gamma}\int_{[0,\delta]\times[u_0,u]}|u'|^{5+2\gamma-p(3+2\gamma)}\|\Omega(\Db\omegas+\Omega\tr\chib\omegas)\|_{H^{N-1}(\ub,u')}^2\D\ub \D u'\\
	\lesssim&a^2.
\end{align*}
The desired estimate for $\omega$ can be obtained by applying \eqref{elliptic1} and improved lower order estimates for $\omega$, as done in the case of $\omegab$.

	\end{proof}

\section{The trapped surface formation}\label{sec:trappedsurface}

In this section, we begin to prove Theorems \ref{main1} and \ref{main2}. Besides this, we will also state an exterior trapped surface formation for data large at the threshold. 

 Consider the problem considered in Theorem \ref{main3}. All assumptions and conclusions in Theorem \ref{main3} hold. Consider the equation below rewritten from the null structure equation of $\Db(\Omega\chih)$ and $\Db(L\phi)$:
\begin{align*}
\frac{\partial}{\partial u}\left(|u|^2|\Omega\chih|^2+2|u|^2|\widetilde{L\phi}|^2\right)=&2|u|^2\langle\Omega^2\left(\nablas\widehat{\otimes}\eta+\eta\widehat{\otimes}\eta+\nablas\phi\widehat{\otimes}\nablas\phi\right)-\nablas_b(\Omega\chih),\Omega\chih\rangle\\
&+4|u|\left(\Omega^2\Deltas\phi+2\Omega^2(\eta,\nablas\phi)\right)\cdot(|u|L\phi)\\
&-\widetilde{\Omega\tr\chib}\left(|u|^2|\Omega\chih|^2+2|u|^2|L\phi|^2\right)-\Omega\tr\chi|u|^2\langle\Omega\chibh,\Omega\chih\rangle\\
&-4\left(|u|\nablas_bL\phi\cdot|u|L\phi-|u|\nablas_b\left(\frac{\varphi}{|u|}\right)\cdot\varphi\right)\\
&-2\left(\widetilde{\Omega\tr\chi|u|\Lb\phi\cdot|u|L\phi}\right)
\end{align*}
where $b$ satisfies $Db=-4\Omega^2\zeta^{\#}$ and $b|_{\Cb_0}=0$. The vector $b$ satisfies $\partial_u+b=\Lb.$ We then have
\begin{align*}&\left|\partial_u\left(|u|^2|\Omega\chih|^2+2|u|^2|\widetilde{L\phi}|^2\right)(\ub,u,\vartheta)\right|\\
\lesssim& \delta|u|^{-2+3p}a^2+|\Omega\tr\chi|\delta|u|^{2p-1}a^2+\Omega_0^2|\tr\chi'(\ub,u,\vartheta)-\tr\chi'(0,u,\vartheta)| |u|^{p}a. \end{align*}
Let us denote
$$M(u,\vartheta)=\int_0^\delta\left(|u|^2|\Omega\chih|^2+2|u|^2|\widetilde{L\phi}|^2\right)(\ub,u,\vartheta)\D\ub.$$
By  Raychaudhuri equation
 $$D\tr\chi'=-\frac{1}{2}(\Omega\tr\chi')^2-|\chih|^2-2(\Lh\phi)^2,$$
 by plugging in the initial bound $|u||\Omega\tr\chi\big|_{\Cb_0}|\le 2$, and absorbing the first term on the right hand side, and we have
 $$|\Omega\tr\chi|\lesssim |u|^{p-1}+\delta|u|^{2p-2}K^2+|u|^{-2}M(u,\vartheta).$$
 Plugging this in the Raychaudhuri equation we have 
 $$\Omega_0^2|\tr\chi'(\ub,u,\vartheta)-\tr\chi'(0,u,\vartheta)|\lesssim |u|^{-2}M(u,\vartheta)+\delta|u|^{2p-2}K^2.$$
Recall that $K$ is introduced in \eqref{initialboundCb0}. It is important that $a$ does not appear in the above two estimates. So we have
\begin{align*}\left|\partial_u\left(|u|^2|\Omega\chih|^2+2|u|^2|\widetilde{L\phi}|^2\right)(\ub,u,\vartheta)\right|\lesssim \delta|u|^{3p-2}a^2+\delta|u|^{3p-2}K^2a+|u|^{p-2}a M\end{align*}
and (after integrating over $\ub\in[0,\delta]$, and noting that $K\ge1$)
\begin{equation*}|\partial_u M(u,\vartheta)|\lesssim K^2\delta^2|u|^{3p-2}a^2+\delta|u|^{p-2}a M(u,\vartheta).\end{equation*}
 Integrating over $u\in[u_0,u_1]$, for $\varepsilon$ sufficiently small, we have
\begin{equation}\label{equationofM}M(u_1,\vartheta)\ge \frac{1}{2}M(u_0,\vartheta)-cK^2\delta^2|u_1|^{-1+3p}a^2\end{equation}
for some $c$ depending on $u_0$, $\frac{1}{p}$ and $\frac{1}{k^2-\epsilon-p}$. 

Now let us assume
\begin{equation}\label{lowerbound1}M(u_0,\vartheta)=\int_{0}^{\delta}\left(|u|^2|\Omega\chih|^2+2|u|^2|\widetilde{L\phi}|^2\right)(\ub,u_0,\vartheta)\D\ub\geq c_1|u_1|^{1+p},\ \forall\vartheta\in S^2\end{equation}
for some $c_1$ which we will pick it up later. By \eqref{equationofM} we have
\begin{align*}
\int_{0}^{\delta}\left(|u|^2|\Omega\chih|^2+2|u|^2|\widetilde{L\phi}|^2\right)(\ub,u_1,\vartheta)\D\ub&\geq \frac{1}{2}c_1|u_1|^{1+p}-cK^2(\delta|u_1|^{-1+p}a)^2|u_1|^{1+p}\\
&\ge \frac{1}{2}c_1|u_1|^{1+p}-cK^2\varepsilon^2|u_1|^{1+p}\
\end{align*}
and hence
\begin{align*}
\int_{0}^{\delta}\left(|u|^2|\Omega\chih|^2+2|u|^2|L\phi|^2\right)(\ub,u_1,\vartheta)\D\ub\geq &\frac{1}{4}c_1|u_1|^{1+p}-\frac{1}{2}cK^2\varepsilon^2|u_1|^{1+p}-2\delta|u|^{2p}K^2\\
\ge&\frac{1}{4}(c_1-2cK^2\varepsilon^2-2\varepsilon K)|u_1|^{1+p}.
\end{align*}
We use again the Raychaudhuri equation, we have
\begin{align*}\tr\chi'(\delta,u_1,\theta)\leq&\frac{2}{|u_1|}-\frac{1}{4|u_1|^2}\Omega_0^{-2}(u_1)\int_{0}^{\delta}\left(|u|^2|\Omega\chih|^2+2|u|^2|L\phi|^2\right)(\ub,u_1,\vartheta)\D\ub\\
\leq&\frac{2}{|u_1|}-\frac{3}{|u_1|}\leq-\frac{1}{|u_1|}<0
\end{align*}
provided that
\begin{equation}\label{choiceofc1}
\frac{1}{4}(c_1-2cK^2\varepsilon^2-2\varepsilon K)\ge 12.
\end{equation}

Instead of \eqref{lowerbound1}, if we assume
\begin{equation}\label{lowerbound2}M(u_0,\vartheta)=\int_{0}^{\delta}\left(|u|^2|\Omega\chih|^2+2|u|^2|\widetilde{L\phi}|^2\right)(\ub,u_0,\vartheta)\D\ub\geq c_1(a\varepsilon^{-1})^{\frac{1+p}{1-p}}\delta^{\frac{1+p}{1-p}},\ \forall\vartheta\in S^2\end{equation}
and moreover $\delta|u_1|^{p-1}a=\varepsilon$, then \eqref{lowerbound1} holds. To ensure that data satisfying \eqref{lowerbound2} exist, in view of \eqref{initialbound} we must require that
\begin{equation}\label{acondition}c_1(a\varepsilon^{-1})^{\frac{1+p}{1-p}}\le a^2\end{equation}
which is possible since $c_1$ is independent of $a$ and $\frac{1+p}{1-p}\le 2$. 
Therefore, we have proved the following.
\begin{theorem}\label{main4}
Consider the characteristic initial value problem of \eqref{ES} as described in Theorem \ref{main3}. In addition to the assumptions of Theorem \ref{main3}, if we assume moreover \eqref{lowerbound2} with $c_1$ given by  \eqref{choiceofc1}, then $S_{\delta,u_1}$ is a closed trapped surface.  

Based on Remark \ref{dependenceonk2-epsilon-p}, the theorem also holds true for $p=k^2-\epsilon$ where $\epsilon\in [0,k^2)$ if we assume that $\nablas L\phi=0$ for $\ub\in[0,\delta]$ on $C_{u_0}$, or more generally that it vanishes sufficiently fast as $\ub\to0$.
\end{theorem}

This is a theorem of trapped surface formation of \eqref{ES} for data large at the threshold without any symmetries, a nonzero-$k$ analogue of the scale-critical trapped surface formation in vacuum of An--Luk \cite{An-Luk}. To see this, Let $\Cb_0$  be the past null cone $\mathcal{N}$ of the singularity of the $k$-self-similar naked singularity solution. Then $\Omega_0(u)=|u|^{k^2}$ for all $u\in[u_0,0)$ and we can pick $p=k^2$ (with $\epsilon=0$).  If on $C_{u_0}$ we choose 
$$\chih(\ub,u_0,\vartheta)\approx\sqrt{c_1(a\varepsilon^{-1})^{\frac{1+k^2}{1-k^2}}} \delta^{\frac{k^2}{1-k^2}}\chih_0(\frac{\ub}{\delta},\vartheta)$$
 for $\ub\in[0,\delta]$, where $\chih_0$ is some fixed seed data. It can be seen the $H^{\frac{1}{2}+\frac{k^2}{1-k^2}}$ norm  in $[0,\delta]\times S^2$ (or $C^{\frac{k^2}{1-k^2}}_{\ub}$) of $\chih$ is large ($\approx \sqrt{c_1(a\varepsilon^{-1})^{\frac{1+k^2}{1-k^2}}} $) and its $H^s$ (or $C^\alpha_{\ub}$) norm tends to $0$ when $s<\frac{1}{2}+\frac{k^2}{1-k^2}$ ($\alpha<\frac{k^2}{1-k^2}$) and tends to $\infty$ when $s>\frac{1}{2}+\frac{k^2}{1-k^2}$  ($\alpha>\frac{k^2}{1-k^2}$) as $\delta\to0$. This simultaneously  gives the instability to trapped surface formation in all regularities below the threshold under exterior (gravitational) perturbations.

\begin{proof}[Proof of Theorem \ref{main1}]
Once the trapped surface formation theorem is established, the remainder of the proof proceeds as in \cite{Li25}, with a modification in order that the family we constructed is uniformly bounded at the threshold. 
Let us first review some basic facts about the $k$-self-similar solutions.  A $k$-self-similar solution ($k\ne0$) considered in \cite{Chr94} is the following. Let $(\mathcal{M},g)$ be a spherically symmetric solution to \eqref{ES}, $h$ be the quotient metric on $\mathcal{Q}=\mathcal{M}/SO(3)$ and $r$ be the area radius of the orbits. Then $(\mathcal{Q},h,r,\phi)$ is called $k$-self-similar if there exists a one-parameter group of diffeomorphism  $f_a$ for $a>0$ on $\mathcal{Q}$ such that
$$f_a^*h=a^2h, \ f_a^*r=ar,\  f_a^*\phi=\phi-k\log a.$$
The vector field $S$ generated by $f_a$ obeys
$$\mathcal{L}_Sh=2h,\ Sr=r,\ S\phi=-k$$
and is then conformal Killing
 $$\mathcal{L}_Sg=2g.$$
 Written in self-similar Bondi coordinates $(u,r), u<0, r>0$ (This $u$ should cause no confusion with the $u$ used above), in which the quotient metric reads
\begin{equation}\label{Bondimetric}h=-\mathrm{e}^{2\nu}\D u^2-2\mathrm{e}^{\nu+\lambda}\D u\D r,\end{equation}
where $\nu,\lambda$ and $\phi+k\log(-u)$ are functions of $s=\log(-\frac{r}{u})$ only. After a suitable rescaling on $u$, the conformal vector field $S$ generated by $f_a$ has the form
$$S=u\frac{\partial}{\partial u}+r\frac{\partial}{\partial r}.$$
Let $l,n$ be outgoing and incoming future directed null vector fields so that $l\cdot n=-2$, then the following dimensionless quantities
$$r\frac{l\phi}{lr},\quad r\frac{n\phi}{n r},\quad lr\cdot nr$$
are functions of $s$. In such a coordinate system, the equation \eqref{ES} can be written as a $2\times2$ autonomous system. For $k^2\in(0,\frac{1}{3})$, Christodoulou was able to show the existence of $k$-self-similar naked singularity solutions.  In this case, the autonomous system can be solved for all $s\in(-\infty,+\infty)$ where $s=-\infty$ represents the center and $s=+\infty$ represents the future Cauchy horizon. There is some $s_*\in\mathbb{R}$ such that $s=s_*$ represents the past null cone $\mathcal{N}$ of the singularity.  We may choose the future null cone $C_o=C_{-1}^+: u=-1$ as initial data of this $k$-self-similar naked singularity. 

Let us denote $C_s^-$ be an incoming null cone starting from on a spherical section on  $C_{-1}$ with the self-similar parameter taking value $s$. We will consider the ``exterior'' problem to $C_s^-$ and let $s\nearrow s_*$.  Fix any $s_1<s_*$, which will determine the regularity of the perturbation.  For any $\epsilon>0$, $s_1$ can be chosen such that 
\begin{equation}\label{nphinrbound}k^2-\epsilon\le r^2\left|\frac{n\phi}{n r}\right|^2\le k^2+\epsilon, s\in [s_1,s_*]\end{equation}
because when $s=s_*$, i.e., on the past null cone of the singularity, $r\frac{n\phi}{n r}=-k$.  Also, because $r\frac{l\phi}{lr}=\frac{1}{k}$ at $s=s_*$ we may assume
\begin{equation}\label{lphilrbound}r\left|\frac{l\phi}{l r}\right|\le\frac{2}{k}, s\in [s_1,s_*].\end{equation}

For each $s\in[s_1,s_*)$,  we apply Theorems \ref{main3} and \ref{main4} for 
$$\Cb_{0,s}=\Cb_s,\quad C_{u_0,s}=C_{-1}^+$$
and construct double null coordinates $(\ub_s,u_s)$ as described before Theorem \ref{main3}. With 
$$|u_{1,s}|=\text{the radius of $\Cb_s\cap\{s=s_1\}$},$$
the assumption \eqref{initialboundCb0} holds for $K=\frac{2}{k}$, the assumption \eqref{omegablower} holds for $p=k^2-\epsilon$. By Theorem \ref{main4} and the discussions below it, for each $s\in[s_1,s_*)$, we can find an initial data $\chih_s$ defined on $\ub_s\in[0,\delta_s]\times S^2$  where $\delta_s$ is determined by 
\begin{equation}\label{deltau1}\delta_s|u_{1,s}|^{k^2-\epsilon-1}a=\varepsilon\end{equation}
 ($\epsilon$ and $\varepsilon$ are distinct parameters, and $a$ can be chosen independent of $s$). The maximal future development contains a closed trapped surface $S_{\delta_s, u_{1,s}}$ . When $s\nearrow s_*$,  $u_{1,s}\to0$ and hence $\delta_s\to0$, and the $C_{\ub_s}^\alpha$ norm converges to zero for any $\alpha<k^2-\epsilon$ when $s\nearrow s_*$. By shifting the coordinate and taking smooth extension we obtain a family of perturbations converging to zero in $C_{\ub}^\alpha C^\infty_\vartheta$ for any $\alpha<k^2-\epsilon$ where $\ub$ is the parameter on $C_{-1}^+$. Moreover, the convergence is smooth away from $S_0=\mathcal{N}\cap C_{-1}^+$, we just need to note that away from $S_0$, the initial shear $\chih$ eventually becomes identically zero as $\delta\to0$. 

Perturbations constructed in this way are supported before $S_0$, we only need to show that $\delta_s=o(s_*-s)$.  This follows from the fact that $|u_{1,s}|=O\left((s_*-s)^{\frac{1}{1-k^2}}\right)$ which follows by a more detailed analysis in \cite{Li25} about the $k$-self-similar solutions and we omit the details.

But note that the perturbations constructed above is not uniformly bounded but tends to $\infty$ as $s\nearrow s_*$. In order to find more regular families, we need to let $s_1\nearrow s_*$ in a single family, and simultaneously $\epsilon=\epsilon(s_1)$ will then change. The construction proceeds as follows. We start by giving a $\delta$, and then we choose $\epsilon=\epsilon(\delta)$ such that $\delta^{\epsilon(\delta)}\to1$ as $\delta\to0$. For this $\epsilon$, we choose $s_1$ so that \eqref{nphinrbound} and \eqref{lphilrbound} hold. Choose $a$ so that \eqref{acondition} holds for $p=k^2$. Then $\eqref{acondition}$ holds for all $p$ near $k^2$. Then we determine $|u_{1,s}|$, which is a number determined on $s_1$ and $s$, according to \eqref{deltau1}, and this finally determine $s$.  The coordinate system $(\ub_s,u_s)$ then depends on $\delta$.  So if we define a $\delta$ family of shear\footnote{Strictly speaking, $\chih$ cannot be prescribed directly in this way, since it may not be trace-free with respect to the resulting metric; see \cite{Chr08} for how the conformal metric is specified as initial data. Here we only keep track of the scaling relation.}
$$\chih_\delta(\ub_s,u_{0,s},\vartheta)=\sqrt{c_1(a\varepsilon^{-1})^{\frac{1+(k^2-\epsilon(\delta))}{1-(k^2-\epsilon(\delta))}}} \delta^{\frac{k^2-\epsilon(\delta)}{1-(k^2-\epsilon(\delta))}}\chih_0(\frac{\ub_s}{\delta},\vartheta)$$
where $\chih_0$ is a fixed seed data. Then in the coordinates $(\ub_s,u_s)$, $S_{\delta, u_{1,s}(\delta)}$ is a closed trapped surface. We only need to check the Sobolev norm at the threshold
$$\|\chih_\delta\|_{H^{\frac{1}{2}+\frac{k^2}{1-k^2}}}\approx\sqrt{c_1(a\varepsilon^{-1})^{\frac{1+(k^2-\epsilon(\delta))}{1-(k^2-\epsilon(\delta))}}}\delta^{\frac{k^2-\epsilon(\delta)}{1-(k^2-\epsilon(\delta))}-\frac{k^2}{1-k^2}}\to  \sqrt{c_1(a\varepsilon^{-1})^{\frac{1+k^2)}{1-k^2}}}, \delta\to0.$$
This completes the proof. Note that this limit cannot be chosen small based on the above argument. This is consistent to the fact that small perturbations at threshold (within spherical symmetry) lead to stability.

\end{proof}

We begin to prove Theorem \ref{main2}. The works of anisotropic trapped surface formation can be found in \cite{K-L-R}, \cite{An-Han}, \cite{A}. The equations mentioned below and some basic estimates can be found there.

\begin{proof}[Proof of Theorem \ref{main2}]

As in the above proof of Theorem \ref{main1}, the interior instability can be established once an exterior trapped surface formation theorem is obtained.  So we assume in addition to Theorem \ref{main3}  that  on $\Cb_0$, $|u|^{k^2+\epsilon}\le \Omega_0^2(u)\le|u|^{k^2-\epsilon}$ for some $\epsilon\ge0$ for $u\in[u_0,u_1]$.  We apply  Theorem \ref{main3} for some $p= k^2-\epsilon$ (we could also choose $p<k^2-\epsilon$ in order to allow $L\phi$ to be perturbed as well)  and $\gamma=\frac{p}{1-p}$.

Suppose that the closed trapped surface on $\Cb_\delta$ has the form $\{(\ub,u,\vartheta):\ub=\delta, u=-R(\vartheta)\}$ where $R$ is a function on $\vartheta$.  The new outgoing null expansion is
$$\overline{\tr\chi}=\tr_g\chi-2\Omega\Delta_g R-4\Omega\eta\cdot\nabla_g R-4\Omega^2\chibh_{bc}\nabla_g^bR\nabla_g^cR-\Omega^2\tr_g\chib|\nabla_g R|^2+4\Omega\omegab|\nabla_g R|^2$$
where $g$ is the induced metric on the set $\{(\ub,u,\vartheta):\ub=\delta, u=-R(\vartheta)\}$. All unbarred quantities in the original null frame. In order to have a trapped surface, so we need to satisfy $\overline{\tr\chi}<0$. This is equivalent to the following
\begin{equation}\label{ellipticinequality1}\Delta_gR+\frac{1}{2}\Omega\tr_g\chib|\nabla_gR|^2-\frac{1}{2}\Omega^{-1}\tr_g\chi-2\omegab|\nabla_gR|^2+2\Omega\chibh(\nabla_gR,\nabla_gR)+2g(\eta,\nabla_gR)>0\end{equation}
 By the estimates in Theorem \ref{main3}, we have (evaluated at $(\ub,-R(\vartheta),\vartheta)$)
$$|\Omega\chibh(\nabla_gR,\nabla_gR)|\lesssim(\delta R^{p-1})^{1+\gamma}R^{-1}a|\nabla_gR|^2,$$
$$|g(\eta,\nabla_gR)|\lesssim (\Omega_0^{-1}(-R)R^{\frac{1}{2}p})(\delta R^{p-1})^{1+\gamma}R^{-1}a|\nabla_gR|,$$
$$|\widetilde{\omegab}||\nabla_gR|^2\lesssim (\delta R^{p-1})R^{-1}a|\nabla_gR|^2.$$
Denote $h=h(u)=\frac{|u|\tr\chi'}{2}\big|_{\Cb_0}$, in view of \eqref{equationofM} we have
\begin{align*}-\Omega^{-1}\tr_g\chi(\ub,-R(\vartheta),\vartheta)\ge &-\frac{2h(-R)}{R}+\int_{0}^{\delta}\left(|\chih|^2+2|\Lh\phi|^2\right)(\ub,-R(\vartheta),\vartheta)\D\ub\\
\ge&-\frac{2h(-R)}{R}+\frac{1}{8}\frac{M(u_0,\vartheta)}{\Omega_0^{2}(-R)R^2}-c(\Omega_0^{-2}(-R)R^p)K^2\delta R^{-2+p}a
\end{align*}
and also
$$-2\omegab\big|_{\Cb_0}\ge pR^{-1}>0,$$
$$h(-R)\le1.$$
In order that \eqref{ellipticinequality1} holds, it suffices to find $R=R(\vartheta)$ so that
\begin{equation}\label{ellipticinequality2}
\begin{aligned}&\Delta_gR-R^{-1}|\nabla_gR|^2-\frac{1}{R}+\frac{1}{16}\frac{M(u_0,\vartheta)}{\Omega_0^2(-R)R^2}\\
&-(\delta R^{p-1})R^{-1}a|\nabla_gR|^2-(\Omega_0^{-1}(-R)R^{\frac{1}{2}p})(\delta R^{p-1})^{1+\gamma}R^{-1}a|\nabla_gR|-c(\Omega_0^{-2}(-R)R^p)K^2\delta R^{-2+p}a>0.
\end{aligned}\end{equation}
We make a transformation $R(\vartheta)=C\mathrm{e}^{-\phi(\vartheta)}$ where $C$ is to be determined and define the metric $\gamma$ as $g=R^2\gamma$, so we can write the relationship equation between the two:
$$
\left\{
\begin{aligned}
	\nabla_g\phi&=&R^{-2}\nabla_\gamma \phi \\
	\Delta_g\phi&=&R^{-2}\Delta_\gamma \phi
\end{aligned}
\right.
$$
and
$$
\left\{
\begin{aligned}
	\nabla_gR&=-R^{-1}\nabla_\gamma\phi\\
	\Delta_gR&=R\left[|\nabla_g\phi|^2_g-\Delta_g\phi\right]=R^{-1}\left[|\nabla_\gamma\phi|^2_\gamma-\Delta_\gamma\phi\right]
\end{aligned}
\right.
$$
Therefore, \eqref{ellipticinequality2} becomes
\begin{equation}\label{ellipticinequality3}
\begin{aligned}
	&-R^{-1}\Delta_\gamma\phi-\frac{1}{R}+\frac{1}{16}\frac{M(u_0,\vartheta)}{\Omega_0^2(-R)R^2}\\&-c(\delta R^{p-1})R^{-1}a|\nabla_\gamma \phi|_\gamma^2-c(\Omega_0^{-1}(-R)R^{\frac{1}{2}p})(\delta R^{p-1})^{1+\gamma}R^{-1}a|\nabla_\gamma \phi|_\gamma-c(\Omega_0^{-2}(-R)R^p)K^2\delta R^{-2+p}a>0.
\end{aligned}
\end{equation}Theorem \ref{main3} 
Moreover, denote $\cir{\gamma}$ be the standard round metric, we have
$$|\nabla_\gamma \phi|_\gamma\lesssim|\nabla_{\cir{\gamma}}\phi|_{\cir{\gamma}}.$$
and
$$|\Delta_\gamma\phi-\Delta_{\cir{\gamma}}\phi|\lesssim \delta R^{p-1}a(1+\phi)(|\nabla^2_{\cir{\gamma}}\phi|_{\cir{\gamma}}+|\nabla_{\cir{\gamma}}\phi|_{\cir{\gamma}}+|\nabla_{\cir{\gamma}}\phi|_{\cir{\gamma}}^2).$$
Therefore,  in order to find $R$ so that \eqref{ellipticinequality3} holds, it suffices to show
\begin{equation}\label{ellipticinequality4}
\begin{aligned}
	&\Delta \phi+ 1-\frac{1}{16}\frac{M(u_0,\vartheta)}{\Omega_0^2(-R)R}\\
	&+c(\delta R^{p-1})a(1+\phi)(|\nabla^2  \phi|+|\nabla  \phi|+|\nabla  \phi|^2)\\
	&+c(\Omega_0^{-1}(-R)R^{\frac{1}{2}p})(\delta R^{p-1})^{1+\gamma}a|\nabla  \phi| +c(\Omega_0^{-2}(-R)R^p)K^2\delta R^{-1+p}a<0
\end{aligned}
\end{equation}
where $\Delta,\nabla$ are with respect to standard round metric. We use the following construction introduced in \cite{K-L-R}.

\begin{lemma}
There is a function $\phi$  on $S^2$ verifying the differential inequality
$$\Delta\phi+1<\frac{1}{2}M_0\mathrm{e}^\phi$$
so that
$$\log\left(\frac{1}{M_*\lambda^2}\right)-O(1)\le \phi\le \log\left(\frac{1}{M_*\lambda^5}\right)+O(1),$$
$$|\nabla\phi|\le O(\lambda^{-1}), |\nabla^2\phi|\le O(\lambda^{-2}).$$ 
Here $M_0$ satisfy
$$M(\vartheta)\ge M_*>0, \vartheta\in B_\lambda(p)$$
for some $p\in S^2, \lambda\in (0,\pi)$. 
\end{lemma}

We choose the initial data in the form
\begin{equation}\label{initialdataalphabeta}\chih_\delta(\ub,u_0,\vartheta)=\delta^\beta \chih_0(\frac{\ub}{\delta}, \frac{\vartheta}{\delta^\alpha}),\end{equation}
where $\chih_0$ is defined on $[0,1]\times S^2$ and supported in $(0,1)\times B(p)$ where $B(p)\subset S^2$ is a small ball centered at some $p\in S^2$ and $\vartheta_p=0$.  Here $\alpha,\beta\ge0$ should not be confused with the curvature components denoted by the same symbols. We will first fix $\alpha,\beta$ to construct a family of perturbations that converges to zero according to an arbitrarily given  regularity below the threshold, then by allowing $\epsilon$, $\alpha$ and $\beta$ to vary with $\delta$ as in the proof of Theorem \ref{main1}, we will obtain a family which converges zero at all regularities below the threshold (but not necessarily uniformly bounded at the threshold).

We may leave the scalar field unperturbed, then 
$$M(u_0,\vartheta)=\delta^{1+2\beta}A_0(\vartheta)$$
where $A_0(\vartheta)=\Omega_0^2(u_0)\int_0^1|u_0|^2|\chih_0(v,\frac{\vartheta}{\delta^\alpha})|^2\D v$. We choose $\phi$ as above with $M_0(\vartheta)=\frac{1}{8}A_0(\vartheta)$.  Then $M_*$ is a number depending only on $\chih_0$ and $\lambda\approx \delta^\alpha$ with the implicit constants depending only on $\chih_0$.  So we are able to find a function $\phi$ so that 
$$\log\left(\frac{1}{\delta^{2\alpha}}\right)-O(1)\le \phi\le \log\left(\frac{1}{\delta^{5\alpha}}\right)+O(1),$$
$$|\nabla\phi|\le O(\delta^{-\alpha}), |\nabla^2\phi|\le O(\delta^{-2\alpha}).$$ 
where the implicit constants depending on $\chih_0$. 

Choose $C=\delta^{\frac{1+2\beta}{1+p}}$ and $R(\vartheta)= \delta^{\frac{1+2\beta}{1+p}}\mathrm{e}^{-\phi}$, note that $\delta$ can be chosen small enough so that $\phi\ge0$. Then
$$\frac{M(u_0,\vartheta)}{8\Omega_0^2(-R)R}\mathrm{e}^{-\phi}\ge\frac{\delta^{1+2\beta}A_0(\vartheta)}{8CR^{p}}=\frac{\delta^{1+2\beta}A_0(\vartheta)}{8C^{1+p}\mathrm{e}^{-p\phi}}\ge \frac{1}{8}A_0(\vartheta)=M_0(\vartheta).$$
Note also that $h(-R)\le 1$,  we have
$$\Delta \phi+ 1-\frac{1}{16}\frac{M(u_0,\vartheta)}{\Omega_0^2(-R)R}<0.$$
It suffices to show that the remaining terms in \eqref{ellipticinequality4} can be chosen arbitrary small and hence \eqref{ellipticinequality4} still holes true.  Recall that we apply  Theorem \ref{main3}  for $p=k^2-\epsilon$. For our construction, we will choose $\beta\le \frac{p}{1-p}$, so we should replace the number $a$ in Theorem \ref{main3} by $a\delta^{\beta-\frac{p}{1-p}}$. Also, if we choose moreover $\alpha>0$, then we should replace the number $a$ in Theorem \ref{main3} and the above inequalities by $a\delta^{\beta-\frac{p}{1-p}-2N\alpha}$ where $N$ is the number of derivatives we used in the estimate. We believe that the latter loss can be recovered by a more refined argument as in \cite{K-R12}, but we do not pursue this here. Now we choose $\delta$ so that
\begin{equation}\label{deltau1modified}\delta|u_1|^{p-1}\cdot a\delta^{\beta-\frac{p}{1-p}-2N\alpha}=\varepsilon \delta^{3\alpha}|u_1|^{2\epsilon},\end{equation}
where $\varepsilon$ is the constant determined in Theorem \ref{main3}. Note that $k^2+\epsilon-p\ge 2\epsilon\ge0$,  then the estimates there hold and  the remaining terms in \eqref{ellipticinequality4} satisfy
$$(\delta R^{p-1})\cdot a\delta^{\beta-\frac{p}{1-p}-2N\alpha}\cdot (1+\phi)(|\nabla^2  \phi|+|\nabla  \phi|+|\nabla  \phi|^2)\lesssim \varepsilon|u_1|^{2\epsilon}(1+\log(2+\delta^{-5\alpha}))\delta^{\alpha},$$
$$(\underbrace{\Omega_0^{-1}(-R)R^{\frac{1}{2}p}}_{\lesssim|u_1|^{-\epsilon}})(\delta R^{p-1})^{1+\gamma}\cdot a\delta^{\beta-\frac{p}{1-p}-2N\alpha}\cdot|\nabla  \phi| \lesssim \varepsilon |u_1|^{\epsilon}\delta^{2\alpha},$$
$$(\Omega_0^{-2}(-R)R^p)K^2\delta R^{-1+p}\cdot a\delta^{\beta-\frac{p}{1-p}-2N\alpha}\lesssim \varepsilon K^2\delta^{3\alpha}.$$
Note that if $\alpha$ is small enough, $\delta\to0$ as $|u_1|\to0$, which is allowed when we push the incoming null cones before $\mathcal{N}$ to $\mathcal{N}$. In particular, the right hand sides of the above estimates $\le c\varepsilon K^2$. So if we choose $\varepsilon$ small enough so that 
$$\Delta \phi-p |\nabla\phi|^2+ h(-R)-\frac{1}{16}\frac{M(u_0,\vartheta)}{\Omega_0^2(-R)R}+c\varepsilon K^2<0,$$
the function $\phi$ satisfies \eqref{ellipticinequality4}  and hence $R(\vartheta)=\delta^{\frac{1+2\beta}{1+p}}\mathrm{e}^{-\phi}$ satisfies \eqref{ellipticinequality1}.

We show that this $R(\vartheta)$ is located within the existence region \eqref{deltau1modified}, that is, it should hold $R(\vartheta)=\delta^{\frac{1+2\beta}{1+p}}\mathrm{e}^{-\phi}\ge |u_1|$. From the upper bound of $\phi$, and plugging in \eqref{deltau1modified},  it suffices to show 
\begin{equation}\label{inexistenceregion}C_0\delta^{\frac{1+2\beta}{1+p}+5\alpha}\ge \left(\frac{a}{\varepsilon}\right)^{\frac{1}{1-p+2\epsilon}}
\delta^{\frac{1+\beta-\frac{p}{1-p}-(2N+3)\alpha}{1-p+2\epsilon}}.\end{equation}
 for some $C_0$ depending on $\chih_0$. By comparing the exponents of $\delta$, we only need
 \begin{equation}\label{alpha+beta}
\begin{aligned}
\alpha+\beta\le &\frac{p}{1-p}-\frac{2\epsilon(1+p)}{(1-p)(1+4\epsilon-3p)}
\\&-\alpha\left[\frac{(1+p)\bigl(5(1-p+2\epsilon)+2N+3\bigr)}{1+4\epsilon-3p}-1\right].
\end{aligned}
\end{equation}
Recall that $p=k^2-\epsilon <\frac{1}{3}$, the right hand side is $\le\frac{p}{1-p}=\frac{k^2-\epsilon}{1-(k^2-\epsilon)}$. For each $\epsilon>0$, we can choose $\alpha,\beta>0$ so that the above inequality holds. Moreover, by letting $\epsilon\to0$, $\alpha+\beta$ can be chosen arbitrarily close to $\frac{k^2}{1-k^2}$ from below. Note that initial data given by \eqref{initialdataalphabeta} verifies
$$\|\chih\|_{H^s([0,\delta]\times S^2)} \approx C(\chih_0), s=\frac{1}{2}+\alpha+\beta,$$
uniformly bounded at $H^{\frac{1}{2}+\alpha+\beta}$. Because $|u_1|$ equals to the right hand side of \eqref{inexistenceregion}, the power of $\delta$ on the right hand side is strictly less than $\frac{1}{1-k^2}$, this allows us to construct interior perturbations totally supported in the original interior region. By varying $\epsilon(\delta)\to0$, $\alpha(\delta)+\beta(\delta)\to\frac{k^2}{1-k^2}$ and $\delta^{\alpha(\delta)}\to0$, we can extract a single family and the proof of Theorem \ref{main2} is completed.

\end{proof}
\begin{remark}
The above argument works even when $\epsilon=0$, $\alpha=0$ or $\beta=\frac{p}{1-p}$.  If we choose $\epsilon=\alpha=0, \beta=\frac{k^2}{1-k^2}$ simultaneously, and $C_0$ sufficiently large. Because the power of $a$ in \eqref{inexistenceregion} is less than $2$ (as in the proof of Theorem \ref{main4}),  we have obtained a classical fully anisotropic version of the Theorem \ref{main4}, which we will not state separately. By varying $\epsilon$ and $\beta$ (but setting $\alpha=0$)  with $\delta$ as in the proof of Theorem \ref{main1}, we also obtain a family of interior perturbations with angular support contained in a small ball which is  uniformly bounded in $H^{\frac{1}{2}+\frac{k^2}{1-k^2}}$.

But note that even when $\epsilon=0$, if we choose $\alpha>0$, $\alpha+\beta$ cannot be chosen equal to $\frac{k^2}{1-k^2}$. So based on the above argument, we cannot find data that are simultaneously large at the threshold and localized to an angular region whose size tends to zero. More modestly, if we allow $\epsilon, \alpha$ and $\beta$ to vary with $\delta$ and $\alpha(\delta)>0$, in order that the family is uniformly bounded at the threshold and the angular support tends to zero,  we require that $\delta^{\alpha(\delta)+\beta(\delta)-\frac{k^2}{1-k^2}}\to1$ and $\delta^{\alpha(\delta)}\to0$, which is inconsistence with the condition \eqref{alpha+beta}. It would be interesting to see what happens in this endpoint case. 
\end{remark}


\begin{thebibliography}{99}
\bibitem{A} X. An, \textit{Naked Singularity Censoring with Anisotropic Apparent Horizon}, arXiv:2401.02003

\bibitem{An2} X. An, S. Wu, \textit{Naked Singularities beyond Spherical Symmetry: Singular Inner Cauchy Horizons for the Einstein-Scalar Field System}, arXiv:2607.07134v1, 2026.

\bibitem{An3} X. An, S. Wu, \textit{Naked Singularities beyond Spherical Symmetry: Instability of $kappa$-Self-Similar Solutions via an Iteration Scheme}, arXiv:2609.04723 (2026).


\bibitem{An-Han} Xinliang An and Qing Han, \textit{Anisotropic Dynamical Horizons Arising in Gravitational Collapse}, 
arXiv:2010.12524 (2020). 

\bibitem{An-Luk} X. An and J. Luk, \textit{Trapped surfaces in vacuum arising dynamically from mild incoming radiation}, Advances in Theoretical and Mathematical Physics 21 (2017), 1--120. 

\bibitem{A-T} X. An, H. K. Tan, \textit{A Proof of Weak Cosmic Censorship Conjecture for the Spherically Symmetric Einstein-Maxwell-Charged Scalar Field System}, arXiv:2402.16250

  \bibitem{Cho} M.W. Choptuik, \textit{Universality and scaling in gravitational collapse of a massless scalar field}, Phys. Rev. Lett., 70, 9–12, (1993).


\bibitem{Chr87} D. Christodoulou, \textit{A Mathematical Theory of Gravitational Collapse}, Comm. Math. Phys. 109, (1987), 613--647

\bibitem{Chr91} D. Christodoulou, \textit{The formation of black holes and singularities in spherically symmetric gravitational collapse}, Communications on Pure and Applied Mathematics 44, no. 3 (1991): 339-373.

\bibitem{Chr93} D. Christodoulou, \textit{Bounded variation solutions of the spherically symmetric einstein-scalar field equations}, Communications on Pure and Applied Mathematics 46, no. 8 (1993): 1131-1220.

\bibitem{Chr94} D. Christodoulou, \textit{Examples of Naked Singularity Formation in the Gravitational Collapse of a Scalar Field}, Annals of Mathematics, (1994) 140(3), 607--653.

\bibitem{Chr99} D. Christodoulou, \textit{The instability of naked singularities in the gravitational collapse of a scalar field}, Ann. of Math. 149, 183-217 (1999).

\bibitem{Chr99cqg} D. Christodoulou, \textit{On the global initial value problem and the issue of singularities}, Class. Quantum Grav. (1999)  16 A23.

\bibitem{Chr08} D. Christodoulou, \textit{The Formation of Black Holes in General Relativity}, European Mathematical Soc., 2009.


\bibitem{CK93} D. Christodoulou, S. Klainerman, \textit{The Global Nonlinear Stability of the Minkowski Space}, Princeton Mathematical Series, vol. 41, Princeton University Press, Princeton, NJ, 1993.


\bibitem{Ci26} S. Cicortas, \textit{Weak cosmic censorship for the circularly symmetric Einstein-scalar field system in 2+1 dimensions}, arXiv:2605.19143v1



\bibitem{Ci-Ke} S. Cicortas, C. Kehle, \textit{Discretely self-similar exterior-naked singularities for the Einstein-scalar field system}, aXiv:2412.09540



\bibitem{D} M. Dafermos, \textit{Spherically symmetric spacetimes with a trapped surface}, Class. Quantum Grav. 22 2221, 2005

 
 \bibitem{G-H-M} C. Gundlach, D. Hilditch, J. M. Mart\'in-Garc\'ia, \textit{Critical Phenomena in Gravitational Collapse}, arXiv: 2507.07636.
 
 
 
 
\bibitem{G-H-J} Y. Guo, M. Hadzic, M., J. Jang, \textit{Naked Singularities in the Einstein-Euler System}, Ann. PDE 9, 4 (2023).


\bibitem{K-L-R} S. Klainerman, J. Luk and I. Rodnianski, \textit{A fully anisotropic mechanism for formation of trapped surfaces in vacuum}, Invent. Math. 198 (2014), no. 1, 1–26.

\bibitem{K-R12} S. Klainerman and I. Rodnianski, \textit{On the Formation of Trapped Surfaces}, Acta Mathematica 208 (2012), no. 2, 211–333.

 
 \bibitem{Li25} J. Li, \textit{Interior instability of naked singularities of a scalar field}, arXiv:2508.07655v1, 2025.
 
\bibitem{Li-Liu1} J. Li, and J. Liu \textit{Instability of spherical naked singularities of a scalar field under gravitational perturbations}, Journal of Differential Geometry, 120(1): 97-197, 2022.


\bibitem{Liu-Li} J. Liu and J. Li, \textit{A robust proof of the instability of naked singularities of a scalar field in spherical symmetry}, Comm. Math. Phys. 363 (2018), no. 2, 561–578


\bibitem{L-Z3} J. Li and X. P. Zhu, \textit{Scalar perturbations to naked singularities of perfect fluid}, arXiv: 2505.20766.


\bibitem{O-P} A. Ori, T. Piran, \textit{Naked singularities and other features of self-similar general-relativistic gravitational collapse}, Phys. Rev. D 42, 1068, 1990

\bibitem{Pen69} R. Penrose, \textit{Gravitational collapse: the role of general relativity}, Noovo Cimento 1, 252 - 276 (1969).


\bibitem{R-Sh} I. Rodnianski, Y. Shlapentokh-Rothman, \textit{Naked singularities for the Einstein vacuum equations: The exterior solution}, Ann. Math., 198 (2023), 1, 231-391

\bibitem{Sh} Y. Shlapentokh-Rothman, \textit{Naked Singularities for the Einstein Vacuum Equations: The Interior Solution}, arXiv:2204.09891.

\bibitem{JS1} J. Singh, \textit{A construction of approximately self-similar naked singularities for the spherically symmetric Einstein-scalar field system}, arXiv: 2210.11325.

\bibitem{JS2} J. Singh, \textit{High regularity waves on self-similar naked singularity interiors: decay and the role of blue-shift}, arXiv: 2402.00062

\bibitem{SZ26} J. Singh, W. Zheng, \textit{Nonlinear stability of continuously self-similar naked singularities for the Einstein-scalar field equations II: linearized stability}, arXiv:2605.16095

\bibitem{Zheng} W. Zheng, \textit{Nonlinear stability of continuously self-similar naked singularities for the Einstein-scalar field equations I: main results}, arXiv:2605.16235


\end{thebibliography}
\end{document}